\documentclass[11pt,twoside]{article}

\usepackage{amsmath,amsfonts,amsthm,amssymb,amsopn,amstext,amscd}
\usepackage{enumerate}
\usepackage{graphicx}
\usepackage{float}
\usepackage{enumitem}
\usepackage{bigints}
\usepackage{setspace}
\usepackage{rotating}
\usepackage{longtable}

\newcommand{\N}{\mathbb{N}}
\newcommand{\R}{\mathbb{R}}

\newcommand{\cuad}{\begin{flushright}\vspace{-2ex}$\Box$\vspace{-2ex}\end{flushright}}
\newenvironment{Prf}[1][\unskip]{%
\noindent
{\large\textbf{Proof of #1}}\newline
\vspace{-2ex}\noindent{}\newline}\cuad

\newtheorem{thm}{Theorem}
\newtheorem{prop}[thm]{Proposition}

\newtheorem{lem}[thm]{Lemma}
\newtheorem{remark}{Remark}
\newtheorem{ex}{Example}

\begin{document}

\pagenumbering{arabic}

\title{Robust estimation in controlled branching processes: Bayesian estimators via disparities}
\author{M. Gonz\'alez, C. Minuesa, I. del Puerto, A.N. Vidyashankar}

\maketitle

\begin{abstract}
This paper is concerned with Bayesian inferential methods for data from controlled branching processes that account for model robustness through the use of disparities. Under regularity conditions, we establish that estimators built on disparity-based posterior, such as expectation and maximum a posteriori estimates, are consistent and efficient under the posited model. Additionally, we show that the estimates are robust to model misspecification and presence of aberrant outliers. To this end, we develop several fundamental ideas relating minimum disparity estimators to Bayesian estimators built on the disparity-based posterior, for dependent tree-structured data. We illustrate the methodology through a simulated example and apply our methods to a real data set from cell kinetics.
\end{abstract}

\vspace{.2cm}


\vspace{.5cm}

\section{Introduction}\label{Dposterior:sec:introduction}

Branching processes are routinely used to model population evolution in a variety of  scientific disciplines such as cell biology, population demography, biochemical processes, genetics, epidemiology, and actuarial sciences (see, for instance \cite{Devroye}, \cite{Haccou}, \cite{procedings-workshop}, \cite{Kimmel} and \cite{procedings-workshop2}). While several variants of branching processes are available, a particularly useful variant is the controlled branching process (CBP) with random control functions, on which this work is focused.

These processes are discrete time and discrete state stochastic processes -like classical Galton-Watson processes (GWPs)- describing the growth of a population across generations using probability distributions for reproduction. However, unlike the GWPs, the number of progenitors is determined by a random mechanism which is referred to as control functions. A good reference for the  probabilistic theory and inferential issues developed until now for the CBPs is the recent monograph \cite{CBPs}. The role of control functions is to specify the number of progenitors in each generation enabling to model a variety of random migratory movements such as immigration and emigration. These CBPs allow for modeling real-life phenomenon with great flexibility. Nevertheless, such a great flexibility of CBPs comes with a cost; namely, they require specification of multiple distributions such as the offspring distribution and the control distributions. It is well-known (see \cite{art-MDE}) that specifying simultaneously the offspring distribution and control distributions for data analysis is challenging and is prone to misspecifications. Divergence-based methods have been used to provide methodologies for inference in these settings that are robust to presence of aberrant outliers and efficient when the posited model is correct, see \cite{Sriram-2000} for the GWP and  the more general approach given in \cite{art-MDE} for the CBP. While these papers deal with the problem from the frequentist standpoint, divergence-based methods in the Bayesian framework in the context of branching processes with strong theoretical guarantees do not exist. This paper is focussed on addressing this problem.

Bayesian inference in the presence of model misspecification has received much attention in recent years. In the context of classification problems, \cite{JT08} studied the behaviour of the so-called Gibbs posterior under a variety of conditions on the risk functions and hence allowing for potential model misspecification. On the other hand, \cite{Hooker-Vidyshankar-2014} provided an alternative approach for Bayesian inference under model misspecification using divergences. Indeed, in their work they evaluate the effect of model misspecifications by studying the asymptotic behaviour of the posterior estimates under posited model and misspecified model settings. Alternatively, \cite{Ghosh-Basu-2016,Ghosh-Basu-2017} used power divergences to derive robust Bayesian methods. More recently, \cite{MD15} developed an alternative coarsening approach which essentially amounts to assuming that the observed data are within an $\epsilon-$ neighborhood of posited model for a suitably defined neighborhood and $\epsilon$; a variety of examples illustrating their approach is provided in the aforesaid work.

The current paper is concerned with the Bayesian estimation of the offspring distribution, which is assumed to belong to a parametric family, via disparities using CBP data. Considering the sample given by the entire family tree, roughly speaking our proposed method consists in replacing the log-likelihood in the expression of the posterior distribution   with an appropriately scaled disparity measure, and considering  the expectation and the mode of the resulting function we get Bayes estimators,  known as EDAP and MDAP estimators (see Section \ref{Dposterior:sec:method} for definition), respectively. We establish their asymptotic properties such as consistency and asymptotic normality and their robustness properties. Indeed, we show that such estimators are efficient when the posited model is correct.  These results are the first ones for dependent tree-structured data. Additionally, when the control function yields the population size of a generation as progenitors, one obtains results for classical GWPs. We also present two examples in which we use the statistical software and programming environment R. In the first example, we deal with a real data set of a oligodendrocyte cell population, firstly used in \cite{Yakovlev-Stoimenova-Yanev-2008}. The second example is a simulation study to illustrate the goodness of this methodology, paying especial attention to the robustness properties.

Besides this Introduction, the paper is structured as follows: Section \ref{Dposterior:sec:model} describes the CBP model and states the hypotheses considered throughout the manuscript. Section \ref{Dposterior:sec:method} is concerned with the description of Bayes estimators under disparity measures and provides sharp probabilistic bounds for differences between the EDAP and MDAP estimators and the frequentist minimum disparity estimators. Section \ref{Dposterior:sec:asymptotic-properties} is devoted to the asymptotic properties of the estimators while Section \ref{Dposterior:sec:robustness-properties} deals with robustness properties. Section \ref{Dposterior:sec:examples} contains the examples briefly described above and the contributions of the paper are summarized in Section \ref{Dposterior:sec:conclusion}. Finally, the proofs of the results are presented in the Supplementary material.

\section{The probability model}\label{Dposterior:sec:model}

In mathematical terms, a \emph{controlled branching process (CBP) with random control functions} is a process $\{Z_n\}_{n\in\N_0}$ defined recursively as:

\begin{equation}\label{Dposterior:def:model}
Z_0=N,\quad Z_{n+1}=\sum_{j=1}^{\phi_n(Z_{n})}X_{nj},\quad n\in\N_0,
\end{equation}
where $\N_0=\N\cup\{0\}$, $N\in\N$, $\{X_{nj}:\ n\in\N_0;\ j\in\N\}$ and $\{\phi_n(k):n,k\in\N_0\}$ are two independent families of non-negative integer valued random variables defined on a probability space $(\Omega,\mathcal{A},P)$. The empty sum in \eqref{Dposterior:def:model} is considered to be 0. The random variables $X_{nj}$, $n\in\N_0$, $j\in\N$, are assumed to be independent and identically distributed (i.i.d.) and $\{\phi_n(k)\}_{k\in\N_0}$, $n\in\N_0$, are independent stochastic processes with equal one-dimensional probability distributions. Intuitively, this process models an evolving population in which each individual reproduces, independently of each other and of the previous generation population, according to the same probability distribution. Similar to the classical GWP in the reproduction mechanism, the number of reproducing individuals in the $n$-th generation, $\phi_n(Z_n)$, is  however a random function $\phi_n(\cdot)$ of the generation size $Z_n$ rather than $Z_n$ itself. In the case when $\phi_n(x) \equiv x$, one obtains the classical GWP.  As in the GWP, the quantity $X_{nj}$ can be interpreted as the number of offspring produced by the $j$-th progenitor in the $n$-th generation.  The terminology, number of progenitors of the $n$-th generation will sometimes be used to describe $\phi_n(Z_n)$. We denote the offspring distribution by $p=\{p_k\}_{k\in\N_0}$, $p_k=P[X_{01}=k]$, $k\in\N_0$. Moreover, in relation to the moments of the process, we denote by
\begin{equation*}
m=E[X_{01}],\quad \mbox{and}\quad \sigma^2=Var[X_{01}],
\end{equation*}
the  offspring mean and variance (assumed to be finite), while by
\begin{equation*}
\varepsilon(k)=E[\phi_0(k)],\quad \mbox{and}\quad \sigma^2(k)=Var[\phi_0(k)],\quad k\in\N_0,
\end{equation*}
the mean and variance function of the random control functions.

In this paper, we focus on offspring distributions that are parametric; that is, we assume that $p_k(\theta)= p_k=P[X_{01}=k]$, for $k\in\N_0$, $\theta\in\Theta$ and $\Theta \subseteq \R$ with a non-empty interior.  Furthermore, we denote the parametric family by ${\mathcal{F}}_{\Theta}=\{p(\theta): \theta \in \Theta\}$, where $p(\theta)=\{p_k(\theta): k\in\N_0\}$, is the offspring distribution for each $\theta \in \Theta$. Hence, if the offspring distribution generating the CBP is parametric, then we assume that there exists an interior point $\theta_0 \in \Theta$ such that $p_k=p_k(\theta_0)$, for all $k\in\N_0$, and hence, we write  $m=m(\theta_0)$ and $\sigma^2=\sigma^2(\theta_0)$ for $m$ and $\sigma^2$,  respectively. We notice here that the assumption $\Theta \subseteq \R$ can be relaxed substantially at the cost of more cumbersome notations. Furthermore, we also assume that the parametric model satisfies the following condition; namely that

\begin{equation}\label{Dposterior:eq:identif}
p_k(\theta_1)=p_k(\theta_2),\quad \forall k\in\N_0\qquad \Rightarrow\qquad \theta_1=\theta_2.
\end{equation}
We  will  frequently  refer  to  the  above  condition  as  identifiability  condition  in  the  rest  of  the
manuscript.

The main asymptotic results that we describe in the paper will require that the generation sizes diverge to infinity with positive probability. While the assumption of supercriticality of the offspring distribution is sufficient in the ordinary GWP case, one needs a slight modification of such a condition for the CBP. Previously, we need the following assumption that we make throughout the paper.

\begin{enumerate}[label=(H\arabic*),ref=(H\arabic*),start=1]
\item The offspring and control distributions satisfy the following two conditions:
\begin{enumerate}
\item [(a)] $p_0(\theta) > 0$, or $P[\phi_n(k) = 0] > 0$, $k\in\N$, $\theta\in\Theta$.
\item [(b)] $\phi_n(0) = 0$ almost surely (a.s.).
\end{enumerate}
\end{enumerate}

Under these conditions, it can be seen that $k=1,2,\ldots$ are transient states while 0 is an absorbing state. As a consequence, the classical GWP duality holds; namely, $P[Z_n\to 0]+P[Z_n\to\infty]=1$ (see \cite{Yanev-75}).  A summary on sufficient conditions to guarantee that the non-extinction set has positive probability can be found in Chapter 3 in \cite{CBPs}. Assumption \ref{Dposterior:cond:consistencyMLE} in Section \ref{Dposterior:sec:method} fixes the framework we need in relation to this issue.

In this paper, we describe a new methodology for robust Bayesian inference for the parameters of the offspring distribution via the use of disparities based on the observations from a CBP tree; that is, the sample is given by the entire family tree up to generation $n$; i.e.,
\begin{equation*}
\mathcal{Z}_n^*=\{Z_l(k): 0\leq l\leq n-1; k\in\N_0\},\quad \mbox{ with }\quad Z_l(k)=\sum_{i=1}^{\phi_l(Z_l)}I_{\{X_{li}=k\}},
\end{equation*}
where $I_A(\cdot)$ represents the indicator function of the set $A$. Notice that $Z_l(k)$ represents the number of individuals in generation $l$ who have exactly $k$ offspring. The proposed methods are a principled approach to address model misspecification, within  a Bayesian framework, which is frequently encountered when working with complex tree-structured data. Finally, the results presented in the manuscript are the first ones for a tree-structured Markov chain data. Before we state the main results of the paper, we introduce a few additional notations.  Let

\begin{equation*}
Y_{l}(k)=\sum_{j=0}^{l} Z_j(k),\quad Y_l=\sum_{j=0}^l Z_j,\quad \mbox{and}\quad \Delta_l=\sum_{j=0}^l \phi_j(Z_j),\quad l,k\in\N_0.
\end{equation*}
We observe that  $Y_{l}(k)$ represents the total number of progenitors who have exactly $k$ offspring up to generation $l$. Furthermore, $Y_l$ and  $\Delta_l$  represent the total number of individuals and the total number of progenitors until the $l$-th generation.

\section{Bayesian estimators using disparity measures}\label{Dposterior:sec:method}

Let $\pi(\cdot)$ denote a prior density on $\Theta$.
Then, using the Bayes Theorem and Markov property it can be seen that the posterior density of $\theta$ given the sample $\mathcal{Z}_n^*$ is
\begin{equation}\label{Dposterior:eq:posterior-tree}
\pi(\theta | \mathcal{Z}_n^*)\propto e^{-\Delta_{n-1} KL(\hat{p}_n,\theta)}\pi(\theta),
\end{equation}
where $KL(q,\theta)$ is the Kullback-Leibler divergence between a probability distribution $q$ defined on $\N_0$ and $p(\theta)$; that is,
\begin{equation*}
KL(q,\theta)=\sum_{k=0}^\infty \log\left(\frac{q_k}{p_k(\theta)}\right)q_k,
\end{equation*}
and $\hat{p}_n=\{\hat{p}_{k,n}\}_{k\in\N_0}$ is the non-parametric maximum likelihood estimator (MLE) of the offspring distribution, $p$, based on the sample $\mathcal{Z}_n^*$ (see \cite{art-EM}); that is,
\begin{equation}\label{Dposterior:eq:offspring-dist-MLE}
\hat{p}_{k,n}=\frac{Y_{n-1}(k)}{\Delta_{n-1}},\qquad k\in\N_0,\ n\in\N.
\end{equation}

Two summary Bayes estimators for the parameter $\theta$ are given by the posterior mean or expectation a posteriori (EAP) and the posterior mode or maximum a posteriori (MAP). These are defined as:
\begin{equation*}
\theta_n^*=\int_\Theta \theta \pi(\theta|\mathcal{Z}_n^*) d\theta,\quad \text{ and }\quad \theta_n^+=\arg\max_{\theta\in\Theta}\pi(\theta|\mathcal{Z}_n^*),\quad n\in\N,
\end{equation*}
respectively. However, these estimators fail to yield robust estimates for the offspring parameter $\theta_0$ as is shown in the example below.

\vspace*{1ex}

\begin{ex}\label{Dposterior:ex:LD}
We consider a CBP starting with $Z_0=1$ individual and for each $k\in\N_0$ the control variable $\phi_n(k)$ follows Poisson distribution with parameter equal to $\lambda k$, with $\lambda=0.3$. In practice, these control functions are suitable to describe an environment with an expected emigration. For the offspring distribution, we consider a geometric one with parameter $\theta_0=0.3$, but which is affected by the presence of outliers that can occur at the point $L=11$ with probability 0.15. The offspring mean and variance are $m=m(\theta_0)=(1-\theta_0)/\theta_0=2.333$ and $\sigma^2=\sigma^2(\theta_0)=(1-\theta_0)/\theta_0^2=7.778$, respectively.

Using the statistical software \texttt{R}, we have simulated the first 45 generations of such a process, $z_{45}^*$. The evolution of the number of individuals and progenitors is shown in Figure \ref{Dposterior:fig:cont-pop-post-KL} (left), where a growth in both groups is observed. An estimate of the posterior density function of $\theta$ upon the sample $z^*_{45}$ and using a beta distribution with parameters 1/2 and 1/2 as a non-informative prior distribution (see \cite{Berger-Bernardo-1992}) is plotted in Figure \ref{Dposterior:fig:cont-pop-post-KL} (centre), where one notices the poor estimation for the offspring parameter given by such a function. To illustrate this fact, in Figure \ref{Dposterior:fig:cont-pop-post-KL} (right) we have also plotted the evolution of the EAP and MAP estimates for $\theta_0$ through the generations.

\begin{figure}[H]
\centering\includegraphics[width=0.3\textwidth]{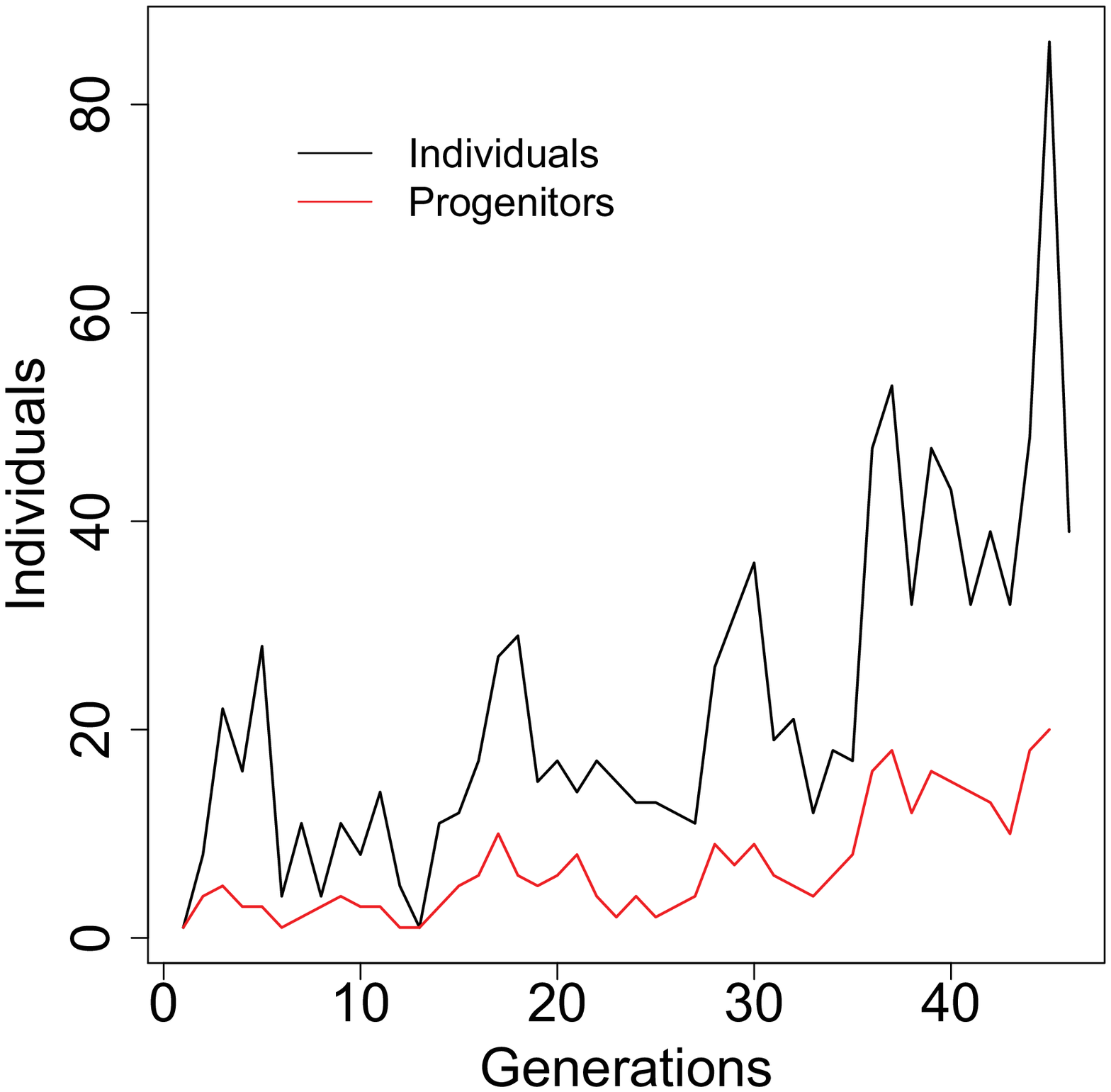}\hspace*{0.03\textwidth}
\includegraphics[width=0.3\textwidth]{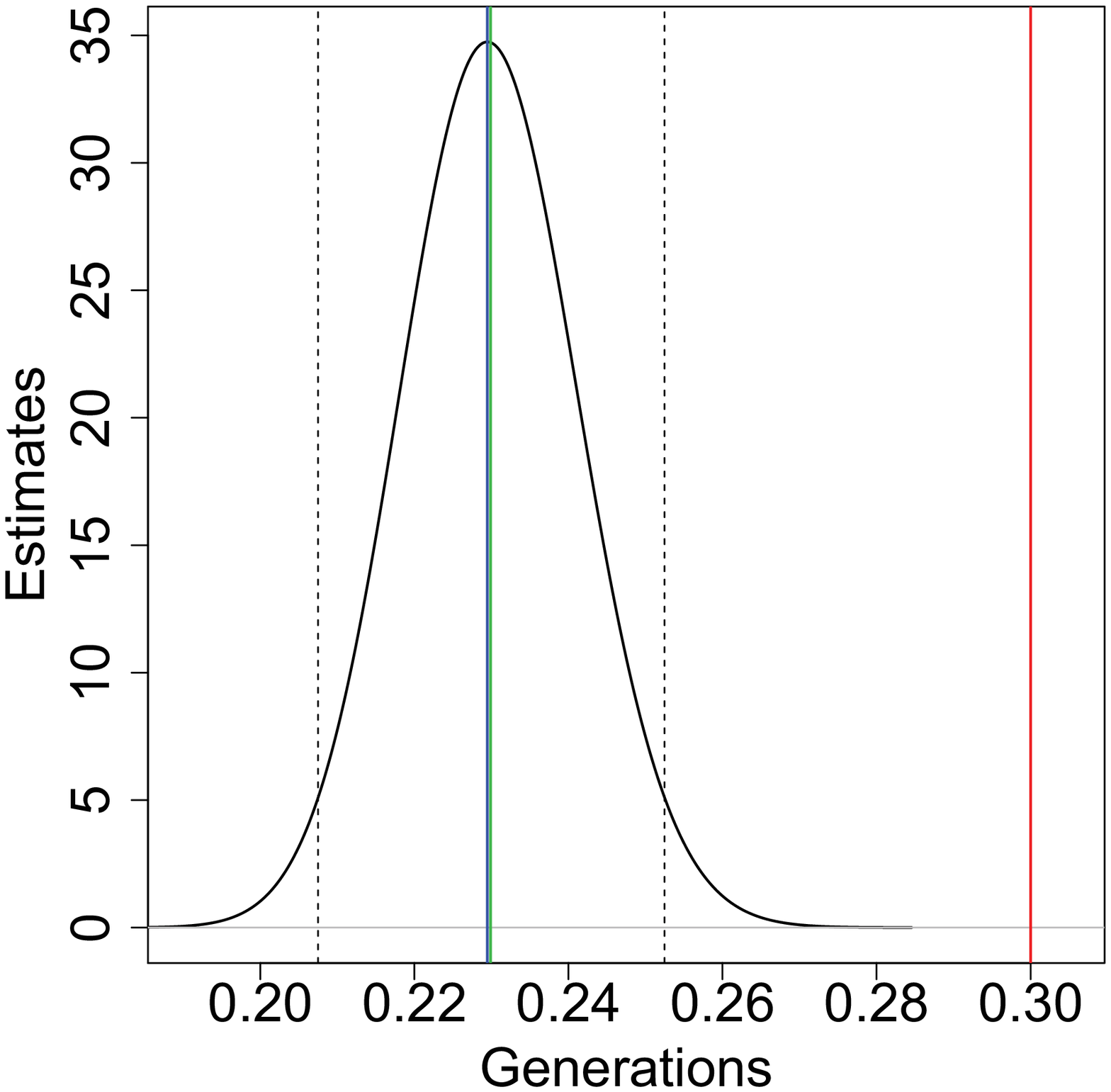}\hspace*{0.03\textwidth}
\centering\includegraphics[width=0.3\textwidth]{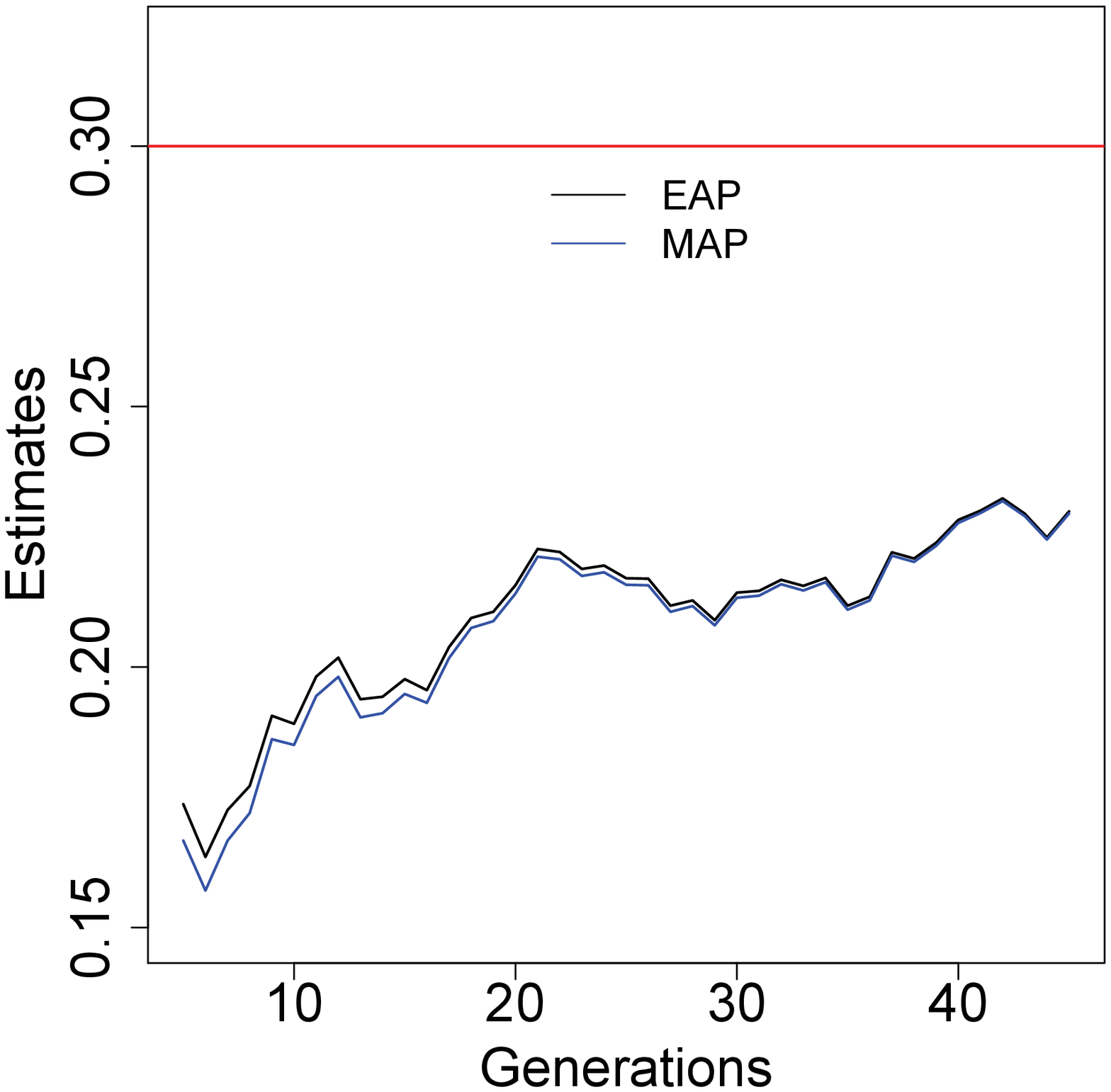}
\caption{Left: evolution of the number of individuals (black line) and progenitors (red line). Centre: estimate of the posterior density of $\theta$ given the sample $z_{45}^*$, together with EAP estimate (blue line) and MAP estimate (green line), high posterior density (HPD) interval (dashed line) and true value of $\theta_0$ (red line). Right: evolution of EAP estimates (black line) and MAP estimates (blue line) of $\theta_0$ through the generations, where the horizontal red line represents the true value of the  parameter.}\label{Dposterior:fig:cont-pop-post-KL}
\end{figure}
\end{ex}

\vspace*{1ex}

Situations similar to the above example clearly showcase the need of robust estimators against outlier contamination. Since outlier contamination can frequently be described using a mixture model, the  robustness of the  Bayes estimators to outliers can be cast in the general framework of model misspecification. The general disparity-based approach, as described in \cite{Hooker-Vidyshankar-2014} and  adapted to the family tree data ${\mathcal{Z}}_n^*$ generated by the CBP, facilitates such an analysis. As in \cite{Hooker-Vidyshankar-2014}, the equation \eqref{Dposterior:eq:posterior-tree} suggests defining a density function by replacing the Kullback-Leibler divergence with a suitable disparity measure $D$ defined on $\Gamma\times\Theta$, where $\Gamma$ is the set of all probability distributions defined on the non-negative integers. A \emph{disparity measure} is defined (see  \cite{lindsay}) using a strictly convex and thrice differentiable function $G:[-1,\infty)\rightarrow \R$ satisfying $G(0)=0$ as follows:
\begin{eqnarray*}
D: \Gamma\times\Theta&\rightarrow& [0,\infty]\nonumber\\
(q,\theta)&\mapsto&D(q,\theta)=\sum_{k=0}^\infty G(\delta(q,\theta,k))p_k(\theta),\label{Dposterior:eq:disparity}
\end{eqnarray*}
where $q=\{q_k\}_{k\in\N_0}$, and $\delta(q,\theta,k)$ denotes the ``Pearson residual at $k$'', that is,
\begin{equation*}\label{Dposterior:eq:Pearson-residual}
\delta(q,\theta,k)=\left\{
                     \begin{array}{ll}
                       \frac{q_k}{p_k(\theta)}-1, & \hbox{if $p_k(\theta)>0$;} \\
                       0, & \hbox{otherwise.}
                     \end{array}
                   \right.
\end{equation*}

The resulting function after replacing the Kullback-Leibler divergence with a disparity-based quantity is referred to as \emph{$D$-posterior density function} at $\hat{p}_n$ and is defined as:
\begin{equation}\label{Dposterior:eq:D-posterior-tree}
\pi_D^n(\theta|\hat{p}_n)=\frac{e^{-\Delta_{n-1}D(\hat{p}_n,\theta)}\pi(\theta)}{\int_\Theta e^{-\Delta_{n-1}D(\hat{p}_n,\theta)}\pi(\theta)d\theta}.
\end{equation}
More generally, the previous definition can be extended for an arbitrary probability distribution $q\in\Gamma$, yielding the $D$-posterior density function at $q$ that is given by
\begin{equation}\label{Dposterior:def:D-posterior-density}
\pi_D^n(\theta|q)=\frac{e^{-\Delta_{n-1}D(q,\theta)}\pi(\theta)}{\int_\Theta e^{-\Delta_{n-1}D(q,\theta)}\pi(\theta)d\theta}.
\end{equation}
Observe that \eqref{Dposterior:def:D-posterior-density} is well defined for each $q\in\Gamma$ if and only if $supp(\pi)\cap\{\theta\in\Theta: D(q,\theta)<\infty\}$ has non-null Lebesgue measure, where $supp(\pi)$ denotes the support of the prior density function. In particular, if the disparity measure $D$ is bounded on $\Gamma\times\Theta$, then $supp(\pi)\cap\{\theta\in\Theta:D(q,\theta)<\infty\}=supp(\pi)$, for each $q\in\Gamma$, and the density function in \eqref{Dposterior:def:D-posterior-density} is well defined. Taking into consideration this fact, we focus our attention on the following set:
\begin{equation*}\label{Dposterior:eq:set-Gamma-tilde}
\widetilde{\Gamma}=\{q\in\Gamma: supp(\pi)\cap\{\theta\in\Theta: D(q,\theta)<\infty\} \text{ has non-null Lebesgue measure}\}.
\end{equation*}
Analogous to the classical Bayes estimators, one can summarize the information from the $D$-posterior using the  estimators \emph{expectation a $D$-posteriori (EDAP}) and \emph{maximum a $D$-posteriori (MDAP}); these are defined as follows:
\begin{itemize}
\item \emph{Expectation a $D$-posteriori (EDAP}):
\begin{equation}\label{Dposterior:def:EDAP-estimator}
\theta_n^{*D}=\int_\Theta \theta \pi_D^n(\theta|\hat{p}_n) d\theta,\quad n\in\N.
\end{equation}
\item \emph{Maximum a $D$-posteriori (MDAP)}:
\begin{equation}\label{Dposterior:def:MDAP-estimator}
\theta_n^{+D}=\arg\max_{\theta\in\Theta}\pi_D^n(\theta|\hat{p}_n),\quad n\in\N.
\end{equation}
\end{itemize}
We next focus on  a few useful examples of non-negative disparity measures and the corresponding $D$-posterior density functions and refer to \cite{cressie-84} and \cite{lindsay} for additional examples.
\begin{ex}
\begin{enumerate}
  \item [(a)] The Kullback-Leibler divergence with the parametric family $\mathcal{F}_\Theta$, defined above, is provided by the function $G(\delta)=(\delta+1)\log(\delta+1)-\delta$. Consequently, the $D$-posterior distribution and the EDAP and MDAP estimators using this disparity coincide with the posterior distribution and the EAP and MAP estimators, respectively.
  \item [(b)] The disparity determined by the function $G(\delta)=2[(\delta+1)^{1/2}-1]^2$ is the \emph{squared Hellinger distance} (or simply Hellinger distance), denoted by $HD(q,\theta)$, for each $q\in\Gamma$, and $\theta\in\Theta$. The $D$-posterior distribution using this disparity is referred to as $HD$-posterior distribution and the EDAP and MDAP estimators are denoted by $\theta_n^{*HD}$ and $\theta_n^{+HD}$, respectively, for each $n\in\N$.
   \item [(c)] The disparity defined by using the function $G(\delta)=e^{-\delta}-1+\delta$ is known as \emph{negative exponential disparity} and denoted by $NED(q,\theta)$ for each $q\in\Gamma$, and $\theta\in\Theta$. The $D$-posterior distribution using this disparity is referred to as $NED$-posterior distribution and the EDAP and MDAP estimators are denoted by $\theta_n^{*NED}$ and $\theta_n^{+NED}$, respectively, for each $n\in\N$.
\end{enumerate}
\end{ex}


For the study of robustness properties of the EDAP and MDAP estimators, which we address in Section \ref{Dposterior:sec:robustness-properties}, we introduce the \emph{EDAP functions}, defined for each $n\in\N$ as follows:
\begin{eqnarray*}\label{Dposterior:def:EDAP-functional}
\overline{T}_n: \Gamma \times \Omega &\rightarrow & \Theta\nonumber\\
(q,\omega) &\mapsto & \overline{T}_n(q)(\omega)=\frac{\int_\Theta\theta e^{-\Delta_{n-1}(\omega)D(q,\theta)}\pi(\theta)d\theta}{\int_\Theta e^{-\Delta_{n-1}(\omega)D(q,\theta)}\pi(\theta)d\theta},
\end{eqnarray*}
and the \emph{MDAP functions}, defined as:
\begin{eqnarray*}\label{Dposterior:def:MDAP-functional}
\widetilde{T}_n: \Gamma\times\Omega &\rightarrow & \Theta\nonumber\\
(q,\omega) &\mapsto &\widetilde{T}_n(q)(\omega)=\arg\min_{\theta\in\Theta}\ (\Delta_{n-1}(\omega)D(q,\theta)-\log (\pi(\theta))),
\end{eqnarray*}
whenever this minimum exists. With these definitions, $\theta_n^{*D}(\omega)=\overline{T}_n(\hat{p}_n)(\omega)$, and $\theta_n^{+D}(\omega)=\widetilde{T}_n(\hat{p}_n)(\omega)$. Note that these functions depend on the total number of progenitors, and  one has different EDAP and MDAP functions for each disparity measure and for each prior distribution; however, in order to ease the notation we will not make explicit this relation.

Despite the similarities between EDAP and MDAP functions defined in the context of branching process and those in the case of i.i.d. random variables described in \cite{Hooker-Vidyshankar-2014}, there is a key difference between both methods that needs to be studied, which is the random nature of EDAP and MDAP functions in the present context (see, for instance \cite{Brown-Purves}). This randomness then leads to questions concerning the existence and measurability of these functionals. To this end, we consider as $\sigma$-field on $\Gamma\times\Omega$ the product of the $\sigma$-fields on $\Gamma$ and $\Omega$, where $\Gamma$ is taken with the Borel $\sigma$-field induced by the topology generated by the $l_1$-metric. In the following, let denote the $l_r$-norm, $r\geq 1$, by $||\cdot||_r$, that is, if $h=\{h_k\}_{k\in\N_0}$ is a sequence of real numbers, then $||h||_r=\left(\sum_{k=0}^\infty |h_k|^r\right)^{1/r}$; when $||h||_r<\infty$, $h$ is said to belong to $l_r$. We also make the following assumption throughout the paper:
\begin{enumerate}[label=(H\arabic*),ref=(H\arabic*),start=2]
\item The prior density $\pi(\cdot)$ posses finite absolute first moment; i.e.\label{Dposterior:mean:prior}
\begin{equation*}
\int_{\Theta} |\theta|\pi(\theta) d\theta < \infty.
\end{equation*}
\end{enumerate}

Our first proposition is concerned with the existence and continuity of the EDAP functional. We note here that the existence is immediate under \ref{Dposterior:mean:prior}.
\begin{prop}\label{Dposterior:prop:existence-EDAP}
Under the assumption that $D(q,\cdot)$ is a continuous function on $\Theta$ for every $q\in\widetilde{\Gamma}$, then, for each $n\in\N$ fixed:
\begin{enumerate}[label=(\roman*),ref=\emph{(\roman*)}]
\item $\overline{T}_n(q)$ exists finitely with probability one and $\overline{T}_n(q)$ is a random variable.\label{Dposterior:prop:existence-EDAP-functional}
\item If $D(\cdot,\theta)$ is continuous in $\widetilde{\Gamma}$ with respect to the $l_1$-metric for each $\theta\in\Theta$, then $\overline{T}_n(\cdot)$ is almost surely continuous on $\widetilde{\Gamma}$ with respect to the $l_1$-metric; that is, $q_j \to q$ in $l_1$, then $\overline{T}_n(q_j)\to \overline{T}_n(q)$, as $j\to\infty$, with probability one. Moreover, $\overline{T}_n$ is a random variable.\label{Dposterior:prop:existence-continuity-EDAP}
\end{enumerate}
\end{prop}
The proof is given in Section \ref{Dposterior:ap:method} of Supplementary material.\vspace*{0.5cm}

For the MDAP function, we will consider the following subclass of the family $\widetilde{\Gamma}$. For each $\omega\in\Omega$, let $\Gamma_{n,\omega}^+$ be a subclass $\Gamma_{n,\omega}^+\subseteq\widetilde{\Gamma}$ which satisfies the following condition: there exists a compact set $C_{n,\omega}^+\subseteq\Theta$ such that for every $q\in\Gamma_{n,\omega}^+$,
\begin{equation*}\label{Dposterior:eq:cond-gamma-plus}
\inf_{\theta\in\Theta\backslash C_{n,\omega}^+} g_n(q,\omega,\theta) > g_n(q,\omega,\theta^+),
\end{equation*}
for some $\theta^+\in C_{n,\omega}^+$, and with $g_n(q,\omega,\theta)=\Delta_{n-1}(\omega)D(q,\theta)-\log (\pi(\theta))$, for each $q\in\Gamma$ and $\theta\in\Theta$.

\begin{prop}\label{Dposterior:prop:existence-MDAP} Let $n\in\N$ fixed and $\Theta$ be a complete and separable subset of $\R$. Assume that for every $q\in\Gamma_{n,\omega}^+$, $D(q,\cdot)$ is a continuous function on $\Theta$, then:
\begin{enumerate}[label=(\roman*),ref=\emph{(\roman*)}]
\item $\widetilde{T}_n(q)$ exists finitely with probability one. In addition, if $\widetilde{T}_n(q)(\omega)$ is unique for each $\omega\in\Omega$, then  $\widetilde{T}_n(q)$ is a random variable.\label{Dposterior:prop:existence-MDAP-functional}
\item The function $\widetilde{T}_n(\cdot)$ is continuous in $q$; that is, $\widetilde{T}_n(q_j)\to \widetilde{T}_n(q)$ with probability one as $j\to\infty$,  as $q_j\to q$ in the sense that $\sup_{\theta\in\Theta} |D(q_j,\theta)-D(q,\theta)|\to 0$. Moreover, $\widetilde{T}_n$ is a random variable.\label{Dposterior:prop:existence-continuity-MDAP-functional}
\end{enumerate}
\end{prop}
Proof is available in Section \ref{Dposterior:ap:method} of Supplementary material. We summarise several observations concerning the properties of disparities in the following remark.

\begin{remark}\label{Dposterior:rem:existence}
\begin{enumerate}[label=(\emph{\alph*})]
\item First, we notice that
if $\Theta$ is a compact set, we can choose $C_{n,\omega}^+=\Theta$, for each $n\in\N$ and $\omega\in\Omega$, and $\Gamma_{n,\omega}^+=\widetilde{\Gamma}$. Moreover,  in that case, $\Theta$ is also complete and separable. Finally, the assumption of compactness can be also removed using conditions similar to those in \cite{CV06}, p.1885.

\item In \ref{Dposterior:prop:existence-MDAP-functional}, the assumption that $\widetilde{T}_n(q)(\omega)$ is unique for each $q\in\Gamma$ and $\omega\in\Omega$ can be weakened by assuming that $\{\theta\in\Theta:\theta=\widetilde{T}_n(q)(\omega)\}$ is a closed subset in $\Theta$. In such a case, by Theorem 4.1 in \cite{Wagner-77} one has that there exists a version of $\widetilde{T}_n$ which is measurable.

\item In \cite{art-MDE}, conditions that guarantee the continuity of $D(q,\cdot)$ on $\Theta$, for each $q\in\Gamma$, are established. In particular, it is sufficient to assume that such a disparity measure is defined by a bounded function $G(\cdot)$ and $p_k(\cdot)$ is a continuous function on $\Theta$, for each $k\in\N_0$.

\item It is known that a sufficient condition for the disparity measure to be bounded is the boundedness of the function $G(\cdot)$. In some cases, when this condition does not hold, it is possible to re-define the disparity measure as a disparity corresponding to a bounded function $G(\cdot)$ with no change in the values of the function $D(q,\cdot)$, for any $q\in\Gamma$. For instance, for the negative exponential disparity, one can also consider the function $\overline{G}(\delta)=e^{-\delta}-1$, which is bounded and satisfies that $NED(q,\theta)=\sum_{k=0}^\infty \overline{G}(\delta(q,\theta,k))p_k(\theta)$, for each $q\in\Gamma$ and $\theta\in\Theta$.

\item In \cite{art-MDE}, it is also proved that the boundedness of the derivative of the function $G(\cdot)$, $G'(\cdot)$, is a sufficient condition for $\sup_{\theta\in\Theta} |D(q_j,\theta)-D(q,\theta)|\to 0$, as $j\to\infty$, if $q_j\to q$ in $l_1$. \label{Dposterior:rem:existence-unif-conv}

\item  For the Hellinger distance, the continuity of $D(q,\cdot)$ is deduced without the boundedness of $G(\cdot)$. Moreover, it is also a bounded disparity in $\Gamma\times\Theta$ despite not being defined by a bounded function $G(\cdot)$ and satisfies the uniform convergence in \ref{Dposterior:rem:existence-unif-conv} (see \cite{art-MDE}).
\end{enumerate}
\end{remark}


Before studying the asymptotic and robustness properties of EDAP and MDAP estimators in Sections \ref{Dposterior:sec:asymptotic-properties} and \ref{Dposterior:sec:robustness-properties}, respectively, we exhibit the strong relation between the behaviour of the EDAP and MDAP functions with their frequentist counterpart.
We begin by considering the \emph{minimum disparity estimator (MDE)} of $\theta_0$ based on $\hat{p}_n$, which is defined as:
\begin{equation*}\label{Dposterior:eq:MDE-offspring-estim}
   \hat{\theta}_n^D = \arg\min_{\theta\in\Theta} D(\hat{p}_n,\theta),
\end{equation*}
and the associated  \emph{disparity function} defined as:
\begin{eqnarray*}
  T: &\Gamma& \rightarrow \Theta\label{Dposterior:def:D-functional}\\
   &q& \mapsto T(q)=\arg\min_{\theta\in\Theta} D(q,\theta),\nonumber
\end{eqnarray*}
whenever this minimum exists. For details about conditions for the existence, uniqueness, and continuity of the function $T$, we refer the reader to Theorems 3.1 - 3.3 in \cite{art-MDE}. Note that analogous to EDAP and MDAP functions, the disparity function depends on the disparity measure; we suppress this in our notations. With this notation, $T(\hat{p}_n)=\hat{\theta}_n^D$.
\vspace*{3ex}

In order to state the result, we strengthen our regularity conditions as below:

\begin{enumerate}[label=(H\arabic*),ref=(H\arabic*),start=3]
\item The parametric family and the disparity measure satisfy:\label{Dposterior:eq:con:bound-disp}
\begin{enumerate}[label=(\emph{\alph*})]
\item $p_k(\cdot)$ is continuous in $\Theta$, for each $k\in\N_0$.\label{Dposterior:cond:cont-par-family}
\item $D$ is a disparity measure associated with a function $G(\cdot)$ which satisfies that $G(\cdot)$ and $G'(\cdot)$ are bounded on $[-1,\infty)$.\label{Dposterior:cond:G-disparity}
\end{enumerate}

\item  The CBP satisfies:\label{Dposterior:cond:consistencyMLE}
\begin{enumerate}[label=(\emph{\alph*})]
 \item There exists $\tau=\lim_{k\to\infty}\varepsilon(k)k^{-1}<\infty$ and the sequence $\{\sigma^2(k)k^{-1}\}_{k\in\N}$ is bounded.
 \item $\tau_m=\tau m >1$ and $Z_0$ is large enough such that $P[Z_n\rightarrow\infty]>0$.
 \item $\{Z_n\tau_m^{-n}\}_{n\in\N}$ converges a.s. to a finite random variable $W$ such that $P[W>0]>0$.
 \item $\{W > 0\}=\{Z_n\to\infty\}$  a.s.
 \end{enumerate}
\end{enumerate}

In \cite{art-EM}, it has been established that under condition \ref{Dposterior:cond:consistencyMLE},  $\Delta_{n}\to\infty$ a.s., and $\hat{p}_{k,n}$ is a strongly consistent estimator of $p_k$, for each $k\in\N_0$,  on the set $\{Z_n\to\infty\}$, where recall that $p=\{p_k\}_{k\in\N_0}$ is the offspring distribution satisfying $p=p(\theta_0)$, for some $\theta_0\in\Theta$. This fact will be useful in the establishment of the asymptotic properties of EDAP and MDAP estimators that we address in next section.

\begin{enumerate}[label=(H\arabic*),ref=(H\arabic*),start=5]
\item The following conditions hold:\label{Dposterior:alimit}
\begin{enumerate}[label=(\emph{\alph*})]
\item $\pi(\cdot)$ is bounded.\label{Dposterior:cond:bound-prior}
\item $\pi(\cdot)$ is thrice differentiable, the third derivative of $\pi(\cdot)$ is bounded, and $\pi(T(q))>0$.\label{Dposterior:cond:3-deriv-prior}
\item For $q\in\Gamma$, $T(q)$ exists, is unique and $T(q)\in int(\Theta)$ (interior of $\Theta$).\label{Dposterior:cond:MDE-ex-uniq-int}
\item For $q\in\Gamma$, $D(q,\theta)$ is twice continuously differentiable in $\theta$ and \linebreak $\frac{\partial^2}{\partial\theta^2} D(q,\theta)_{|\theta=T(q)}>0$.\label{Dposterior:cond:2-deriv-MDE}
\item For $q\in\Gamma$, there exists some $n_0\in\N$ such that $q\in\Gamma_{n,\omega}^+$, and $\widetilde{T}_n(q)(\omega)$ is unique, for each $n\geq n_0$, and $\omega\in\Omega$.\label{Dposterior:cond:uniq-MDAP}
\end{enumerate}
\end{enumerate}

Finally, we introduce the following subset of probability distributions which we need in the statement of our next theorem. Let
\begin{equation}
\Gamma^*=\left\{q\in\widetilde{\Gamma}: \forall\eta>0, \exists\rho>0: \inf_{|\theta-T(q)|>\eta}D(q,\theta)-D(q,T(q))>\rho\right\}.\label{Dposterior:cond:sep}
\end{equation}

Our next result establishes that the EDAP and MDAP functions can be approximated a.s. by the corresponding disparity function up to an error of certain order, on the set $\{Z_n\to\infty\}$.

\begin{thm}\label{Dposterior:thm:aprox-MDE}
Suppose that \ref{Dposterior:cond:consistencyMLE} and \ref{Dposterior:alimit}~\ref{Dposterior:cond:3-deriv-prior}  hold and let $q\in\Gamma^*$ satisfying \ref{Dposterior:alimit}~\ref{Dposterior:cond:MDE-ex-uniq-int}-\ref{Dposterior:cond:2-deriv-MDE}. Then, the following convergences hold:
\begin{enumerate}[label=(\roman*),ref=\emph{(\roman*)}]
\item The EDAP function satisfies\label{Dposterior:thm:aprox-MDE-i-EDAP}
\begin{equation*}
\overline{T}_n(q)-T(q)=o\left(\Delta_{n-1}^{-1/2}\right)\quad \text{ a.s. on } \{Z_n\to\infty\}.
\end{equation*}
\item If $\Theta$ is complete and separable, \ref{Dposterior:alimit}~\ref{Dposterior:cond:bound-prior} holds, and $q$ satisfies \ref{Dposterior:alimit}~\ref{Dposterior:cond:uniq-MDAP}, then the MDAP function satisfies\label{Dposterior:thm:aprox-MDE-ii-MDAP}
\begin{equation*}
\widetilde{T}_n(q)-T(q)=o\left(\Delta_{n-1}^{-1/2}\right)\quad \text{ a.s. on } \{Z_n\to\infty\}.
\end{equation*}
\end{enumerate}
\end{thm}

We refer the reader to Section \ref{Dposterior:ap:method} in the Supplementary material for the proof of this theorem.

\begin{remark}
\begin{enumerate}[label=(\alph*),ref=(\alph*)]
\item We observe that  if the offspring distribution belongs to the parametric family, that is $p=p(\theta_0)$ and $T(p)$ is unique, then $T(p)=\theta_0$; furthermore,
if $p$ satisfies the conditions of Theorem \ref{Dposterior:thm:aprox-MDE}, one has that $\overline{T}_n(p)\to\theta_0$, and $\widetilde{T}_n(p)\to\theta_0$ a.s. on $\{Z_n\to\infty\}$.
\item If $T(q)$ exists and is unique for each $q\in\widetilde{\Gamma}$, under the assumption of the compactness of $\Theta$, one has $\Gamma^*=\widetilde{\Gamma}$. The compactness assumption can be replaced with the assumption that $q$ belongs to the subclass $\overline{\Gamma}\subseteq\Gamma$ which satisfies the following condition: there exists a compact set $\overline{C}\in\Theta$ such that for every $\bar{q}\in\overline{\Gamma}$,
\begin{equation*}\label{Dposterior:eq:cond-gamma-tilde}
 \inf_{\theta\in\Theta\backslash \overline{C}} D(\bar{q},\theta) > D(\bar{q},\theta^*),\quad\text{ for some $\theta^*\in \overline{C}$}.
 \end{equation*}
\end{enumerate}
\end{remark}

\section{Asymptotic properties}\label{Dposterior:sec:asymptotic-properties}

In this section, we focus our attention on the asymptotic behaviour of the EDAP and MDAP estimators. The proofs of the results of this section are collected in Section \ref{Dposterior:ap:asymp} of the Supplementary material. One of our main theorems in this section is that EDAP and MDAP estimators are asymptotically efficient when the posited offspring distribution belongs to the parametric family. This will be obtained as a consequence of a more general result: the ${\mathcal{L}}^1$ almost sure convergence of posterior density to a Gaussian density with mean $0$ and variance equal to the inverse of the Fisher information. Here, $\mathcal{L}^1$ is the space of measurable functions which are Lebesgue integrable equipped with the $L^1$-norm.

To establish the results in this and the following section, we need additional regularity conditions and hence, in the remainder of this paper we assume that for each $q\in\Gamma$, the first and the second derivative of $D(q,\theta)$ with respect to $\theta$ exist and we denote them by $\dot{D}(q,\theta)$ and $\ddot{D}(q,\theta)$  - the reader is referred to Section 4 in \cite{art-MDE} for conditions that guarantee their existence-. Let also denote $I^D(\theta) = \ddot{D}(p,\theta)$, and $I_n^D(\theta) = \ddot{D}(\hat{p}_n,\theta)$, where recall that $p$ is the posited offspring distribution and $\hat{p}_n$ was defined in \eqref{Dposterior:eq:offspring-dist-MLE}. Henceforth, we also assume -without loss of generality (see \cite{lindsay}, p.1089)- that $G'(0)=0$ and $G''(0)=1$. Thus, taking into account the previous hypotheses and under the identifiability condition \eqref{Dposterior:eq:identif} if $p=p(\theta_0)$, one has that $I^{D}(\theta_0)$ reduces to the Fisher information at $\theta_0$ denoted by $I(\theta_0)$ (see \cite{art-MDE}).

In the proofs of the results in this section asymptotic properties of the minimum disparity estimators will play a critical role. These properties were recently obtained in \cite{art-MDE}. We state the required asymptotic properties as an assumption.

\begin{enumerate}[label=(H\arabic*),ref=(H\arabic*),start=6]
\item The following properties concerning the disparity measure and the MDE hold:\label{Dposterior:dis:assum}
\begin{enumerate}[label=(\alph*),ref=(\alph*),start=1]
\item There exists some $C>0$ such that for each $q_1,q_2\in\Gamma$, $\sup_{\theta\in\Theta} |D(q_1,\theta)-D(q_2,\theta)|\leq C||q_1-q_2||_1$.\label{Dposterior:cond:unif-cont-disparity}
\item Let $\theta_p=T(p)$ satisfy $I^D(\theta_p)>0$ and $I_n^D(\theta)\to I^D(\theta_p)$ a.s. on $\{Z_n\to\infty\}$, as $\theta\to\theta_p$ and $n\to\infty$. \label{Dposterior:cond:cont-second-deriv-D}
\item The MDE satisfies that\label{Dposterior:cond:MDE}
\begin{eqnarray}
\hat{\theta}^D_n&\rightarrow&\theta_p\quad\text{a.s. on }\{Z_n\to\infty\},\label{Dposterior:eq:consistencyMDE}\\
\Delta_{n-1}^{1/2}(\hat{\theta}_n^D-\theta_p)&\xrightarrow{d}& N(0,I^D(\theta_p)^{-1}),\label{Dposterior:eq:normalityMDE}
\end{eqnarray}
where $\xrightarrow{d}$ represents the convergence in distribution with respect to the probability $P[\cdot| Z_n\to\infty]$.
\end{enumerate}
\end{enumerate}

\begin{remark}
\begin{enumerate}[label=(\alph*),ref=(\alph*)]

\item We note that \ref{Dposterior:dis:assum}~\ref{Dposterior:cond:unif-cont-disparity} holds for Hellinger distance or for every disparity satisfying \ref{Dposterior:eq:con:bound-disp}~\ref{Dposterior:cond:G-disparity} (see \cite{art-MDE}, p.313).

\item Sufficient conditions for \ref{Dposterior:dis:assum}~\ref{Dposterior:cond:cont-second-deriv-D} can be found in Section 4 in \cite{art-MDE}.

\item The establishment of conditions that guarantee the convergence of \eqref{Dposterior:eq:consistencyMDE} and \eqref{Dposterior:eq:normalityMDE} was succeeded in Theorems 3.4 and 4.1, respectively, in \cite{art-MDE}.
\end{enumerate}
\end{remark}

In the following $\varphi(t; \theta)$ denotes the density function of a normal distribution with mean 0 and variance $I^D(\theta)^{-1}$. Additionally, let $\varphi_n(t)$ denote the density function of a normal distribution with mean $0$ and variance $I_n^D ({\hat{\theta}}_n^D)^{-1}$. The following results hold true for the offspring distribution $p$ without assuming that it belongs to the parametric family. In the case when $p=p(\theta_0)$ and under the identifiability condition \eqref{Dposterior:eq:identif}, then the same results hold with $\theta_p=\theta_0$ and  $I^D(\theta_p)=I(\theta_0)$. We state our first main result of this section.
\begin{thm}\label{Dposterior:thm:consistency}
Let $\overline{\pi}_D^{n}(\cdot|\hat{p}_n)$ denote the $D$-posterior density function of $t=\Delta_{n-1}^{1/2}(\theta-\hat{\theta}^D_n)$. Let $p\in\Gamma^*$ satisfy \ref{Dposterior:alimit}~\ref{Dposterior:cond:MDE-ex-uniq-int}. If  \ref{Dposterior:cond:consistencyMLE} and \ref{Dposterior:dis:assum} hold, then:
\begin{enumerate}[label=(\roman*),ref=\emph{(\roman*)}]
\item $\int |\overline{\pi}_D^{n}(t|\hat{p}_n) -\varphi(t;\theta_p)|dt\rightarrow 0$\quad a.s. on $\{Z_n\to\infty\}.$\label{Dposterior:thm:consistency-l1-param}
\item $\int |t||\overline{\pi}_D^{n}(t|\hat{p}_n) -\varphi(t;\theta_p)|dt\rightarrow 0$\quad a.s. on $\{Z_n\to\infty\}.$\label{Dposterior:thm:consistency-l1-param-t}
\item $\int |\overline{\pi}_D^{n}(t|\hat{p}_n) -\varphi_n(t)|dt\rightarrow 0 \quad a.s. $ on $\{Z_n\to\infty\}.$\label{Dposterior:thm:consistency-l1-estim}
\end{enumerate}
\end{thm}



The following theorem shows that the EDAP estimator mimics the MDE at the rate $\Delta_{n-1}^{1/2}$. This feature leads to the asymptotic normality of the centered and scaled EDAP estimator.

\begin{thm}\label{Dposterior:thm:consistency-EDAP}
Under the hypotheses of Theorem \ref{Dposterior:thm:consistency}, the following convergences hold:
\begin{enumerate}[label=(\roman*),ref=\emph{(\roman*)}]
\item $\Delta_{n-1}^{1/2}(\theta_n^{*D}-\hat\theta_n^D)\rightarrow 0$ a.s. on $\{Z_n\to\infty\}$.\label{Dposterior:thm:conv-MDE-EDAP}
\item $\Delta_{n-1}^{1/2}(\theta_n^{*D}-\theta_p)\xrightarrow{d} N(0,I^D(\theta_p)^{-1})$ with respect to $P[\cdot|Z_n\to\infty]$.\label{Dposterior:thm:normality-EDAP}
\end{enumerate}
\end{thm}


To establish the asymptotic properties of the MDAP estimator, we need the uniform convergence of posterior density to the Gaussian density. This is the content of our next theorem.

\begin{thm}\label{Dposterior:thm:lim-sup-D-post-estim}
Under conditions of Theorem \ref{Dposterior:thm:consistency}, if \ref{Dposterior:alimit}~\ref{Dposterior:cond:bound-prior}-\ref{Dposterior:cond:3-deriv-prior} hold, then:
\begin{equation*}
\lim_{n\to\infty}\sup_{t\in\R}\left|\overline{\pi}_D^{n}(t|\hat{p}_n)-h(t;\theta_p)\right|=0\quad \text{ a.s. on }\{Z_n\to\infty\}.
\end{equation*}
\end{thm}


Now, we turn to the asymptotic behaviour of MDAP estimators. We strengthen the conditions of the previous theorem and assume additionally that there exists $n_0\in\N$ such that $\hat{p}_n(\omega)\in\Gamma_{n,\omega}^+$, and $\hat{\theta}_n^{+D}(\omega)$ is unique, for each $n\geq n_0$, and for each $\omega\in\{Z_n\to\infty\}$. Thus, from Proposition \ref{Dposterior:prop:existence-MDAP} one has that $\theta_n^{+D}=\widetilde{T}_n(\hat{p}_n)$ exists and is unique for $n\geq n_0$.

\begin{thm}\label{Dposterior:thm:consistency-MDAP}
Under that the hypotheses of Theorem \ref{Dposterior:thm:lim-sup-D-post-estim} and if $\Theta$ is complete and separable,  
then the following convergences hold:
\begin{enumerate}[label=(\roman*),ref=\emph{(\roman*)}]
\item $\Delta_{n-1}^{1/2}(\theta_n^{+D}-\hat\theta_n^D)\rightarrow0$ a.s. on $\{Z_n\to\infty\}$.\label{Dposterior:thm:conv-MDE-MDAP}
\item $\Delta_{n-1}^{1/2}(\theta_n^{+D}-\theta_p)\xrightarrow{d} N(0,I^D(\theta_p)^{-1})$ with respect to $P[\cdot|Z_n\to\infty]$.\label{Dposterior:thm:normality-MDAP}
\end{enumerate}
\end{thm}


\section{Robustness properties}\label{Dposterior:sec:robustness-properties}

In this section we describe the robustness properties of Bayesian disparity estimators EDAP and MDAP. The proof of the results of the current section can be found in Section \ref{Dposterior:ap:robust} of Supplementary material. In the context of i.i.d. data some of these issues were investigated in \cite{Hooker-Vidyshankar-2014}. The study of disparities in the frequentist context for branching processes has been investigated for the Hellinger distance in \cite{Sriram-2000} and for a general standpoint in \cite{art-MDE}. It is pertinent to note here that the standard robustness concepts such as influence function, breakdown point (both finite and asymptotic) and $\alpha$-influence curves, $\alpha\in (0,1)$, take a familiar form as the i.i.d. case and hence we will keep the discussion rather brief.

We begin with a brief description of the framework. As is common in such problems, we focus on the gross error contamination model given by
\begin{equation}\label{Dposterior:eq:mixture-model}
    p(\theta,\alpha,L)=(1-\alpha)p(\theta)+\alpha\eta_L,
\end{equation}
where $\theta\in\Theta$, $\alpha\in (0,1)$, $L\in\N_0$, and  $\eta_L$ is a point mass distribution at $L$. Next, we turn to the definition of the $\alpha$-influence function of a random variable $\overline{T}:\Gamma\times\Omega\rightarrow \Theta$. For $\alpha\in (0,1)$, set
\begin{eqnarray*}
IF_{\alpha}(\cdot,\overline{T},p):\N_0\times\Omega &\rightarrow& \R\\
(L,\omega)&\mapsto& IF_{\alpha}(L,\overline{T},p)(\omega)=\alpha^{-1}[\overline{T}(p(\theta_0,\alpha,L))(\omega)-\overline{T}(p(\theta_0))(\omega)].\\
\end{eqnarray*}

In our first result in this section we consider the disparity function $T$ as a degenerate random variable, i.e., $T(q)(\omega)=T(q)$, for each $\omega\in\Omega$ and each $q\in\Gamma$, and establish that the $\alpha-$influence function of the EDAP (respectively, MDAP) function and of the disparity function are close and characterize the difference in terms of the random sample size $\Delta_{n-1}$. Thus, under the conditions of Theorem \ref{Dposterior:thm:aprox-MDE}, the following result on the $\alpha$-influence curves of EDAP and MDAP functions at $p$ is straightforward.

\begin{prop}\label{Dposterior:prop:aprox-IF-EDAP-MDAP-MDE}
Let fix $\alpha\in (0,1)$ and $L\in\N_0$ and assume that $p(\theta_0,\alpha,L)\in\Gamma^*$, and satisfies \ref{Dposterior:alimit}~\ref{Dposterior:cond:MDE-ex-uniq-int}-\ref{Dposterior:cond:2-deriv-MDE}.
\begin{enumerate}[label=(\roman*),ref=\emph{(\roman*)}]
\item Under conditions of Theorem \ref{Dposterior:thm:aprox-MDE}~\ref{Dposterior:thm:aprox-MDE-i-EDAP}, as $n\to\infty$,\label{Dposterior:prop:aprox-IF-EDAP-MDAP-MDE-i}
\begin{equation*}\label{Dposterior:eq:aprox-IF-EDAP-MDE}
  IF_{\alpha}(L,\overline{T}_n,p)-IF_{\alpha}(L,T,p)=o(\Delta_{n-1}^{-1/2})\quad \text{a.s. on } \{Z_n\to\infty\}.
\end{equation*}
\item Under conditions of Theorem \ref{Dposterior:thm:aprox-MDE}~\ref{Dposterior:thm:aprox-MDE-ii-MDAP}, if \ref{Dposterior:alimit}~\ref{Dposterior:cond:uniq-MDAP} holds for $p(\theta_0,\alpha,L)$, then, as $n\to\infty$,\label{Dposterior:prop:aprox-IF-EDAP-MDAP-MDE-ii}
\begin{equation*}\label{Dposterior:eq:aprox-IF-EDAP-MDE}
  IF_{\alpha}(L,\widetilde{T}_n,p)-IF_{\alpha}(L,T,p)=o(\Delta_{n-1}^{-1/2})\quad \text{a.s. on } \{Z_n\to\infty\}.
\end{equation*}
\end{enumerate}
\end{prop}

An immediate consequence of this proposition is that for each $\alpha\in (0,1)$, and $n\in\N$ large enough, the $\alpha$-influence curve of the disparity functional at $p$ provides a good approximation of the $\alpha$-influence curves of $\overline{T}_n$ and of $\widetilde{T}_n$ at this probability distribution. In addition, under conditions given in  Theorem 5.1 in \cite{art-MDE}, one has that $\lim_{L\to\infty}\lim_{n\to\infty} IF_{\alpha}(L,\overline{T}_n,p)=0$ and $\lim_{L\to\infty}\lim_{n\to\infty} IF_{\alpha}(L,\widetilde{T}_n,p)=0$.

The study of the $\alpha$-influence curves as a measure for the robustness of an estimator was driven by the fact that the influence functions do not always provide a good description of the robustness properties of an estimator. Nonetheless, we establish conditions under which the influence functions of the EDAP estimators are bounded, and consequently, from the classical viewpoint, the EDAP estimators are robust. Recall the definition of the influence function for EDAP estimators at $p$ is given by
\begin{eqnarray*}
IF(\cdot,\overline{T}_n,p):\N_0 &\rightarrow& \R\nonumber\\
L&\mapsto& IF(L,\overline{T}_n,p)=\lim_{\alpha\to 0 } IF_{\alpha}(L,\overline{T}_n,p).\label{Dposterior:def:influence-function}
\end{eqnarray*}

\begin{thm}\label{Dposterior:thm:influence-function}
If \ref{Dposterior:eq:con:bound-disp} holds, then $|IF(L,\overline{T}_n,p)|<\infty$, for each $L\in\N_0$ and $n\in\N$.
\end{thm}


Next, we turn our attention to the study of breakdown point. While the notion of influence curves represent the effect of a single outlier in the asymptotic behaviour of the estimator, the breakdown point intuitively represents the percentage of contamination that the estimator can asymptotically bear without taking arbitrarily large values. Classically, the breakdown point of a general function $\overline{T}$ at $q\in\Gamma$ is defined as:
\begin{equation*}\label{Dposterior:def:breakdown-point}
  B(\overline{T},q)=\sup\{\alpha\in(0,1): b(\alpha,\overline{T},q)<\infty\},
\end{equation*}
where $b(\alpha,\overline{T},q)=\sup\ \{|\overline{T}((1-\alpha)q+\alpha\overline{q})-\overline{T}(q)|:\overline{q}\in\Gamma\}$. Note that $b(\alpha,\overline{T},q)=\infty$ is equivalent to the existence of a sequence of probability distributions $\{\overline{q}_L\}_{L\in\N_0}$ satisfying $|\overline{T}((1-\alpha)q+\alpha \overline{q}_L)-\overline{T}(q)|\to\infty$, as $L\to\infty$. That sequence is called sequence of contaminating probability distributions. The notion of breakdown point can be extended in a natural way to a function $\overline{T}: \Gamma\times\Omega\rightarrow\Theta$, such that for each $q\in\Gamma$, $\overline{T}(q)$ is a random variable. We will study the breakdown point of the EDAP and MDAP functions under the following assumption:

\begin{enumerate}[label=(H\arabic*),ref=(H\arabic*),start=7]
\item The contaminating probability distributions and the parametric family satisfy:\label{Dposterior:cond:contam-distr}
\begin{enumerate}[label=(\alph*),ref=(\alph*),start=1]
\item $\lim_{L\to\infty} \sum_{k=0}^\infty \min\left(\overline{q}_{L,k},p_k\right)=0$.\label{Dposterior:cond:ortog-distrib-contam}
\item $\lim_{L\to\infty} \sup_{|\theta|\leq c}\sum_{k=0}^\infty \min\left(\overline{q}_{L,k},p_k(\theta)\right)=0$,\quad $\forall c>0$.\label{Dposterior:cond:ortog-contam-family}
\item $\lim_{|\theta|\to\infty} \sum_{k=0}^\infty \min\left(p_k,p_k(\theta)\right)=0$.\label{Dposterior:cond:ortog-distrib-family}
\end{enumerate}
\end{enumerate}
We notice here that the above conditions represent the worst possible contamination context (see, for instance, \cite{Park-Basu-2004}, p.28). The next result, which in content is analogous to the i.i.d. case, describes the role of the prior in the robustness of the EDAP and MDAP estimators. It is to be noted that the asymptotic breakdown point of the minimum disparity estimators are at least 50\%. For the following results, we also fix the following condition:

\begin{enumerate}[label=(H\arabic*),ref=(H\arabic*),start=8]
\item The function $\delta\in [-1,\infty)\mapsto (\delta+1)G'(\delta)-G(\delta)$ is bounded.\label{Dposterior:cond:bound-RAF}
\end{enumerate}

\begin{thm}\label{Dposterior:thm:breakdown-point-EDAP-MDAP}
Suppose \ref{Dposterior:cond:bound-RAF} holds true.
\begin{enumerate}[label=(\roman*),ref=\emph{(\roman*)}]
\item If $D(q,\theta)$ is bounded for any $q\in\Gamma$ and $\theta\in\Theta$, then the breakdown point of the EDAP at $p$ is 1.\label{Dposterior:thm:breakdown-point-i-EDAP}
\item Assume that \ref{Dposterior:eq:con:bound-disp}~\ref{Dposterior:cond:G-disparity}, and \ref{Dposterior:cond:contam-distr}~\ref{Dposterior:cond:ortog-distrib-family} hold. Then, for any family of contaminating distributions $\{\bar{q}_L\}_{L\in\N_0}$ satisfying \ref{Dposterior:cond:contam-distr}~\ref{Dposterior:cond:ortog-distrib-contam}-\ref{Dposterior:cond:ortog-contam-family}, the breakdown point of the MDAP at $p$ is 1. The result also holds for the Hellinger distance. \label{Dposterior:thm:breakdown-point-ii-MDAP}
\end{enumerate}
\end{thm}

We next turn our attention to a new notion for the study of robustness for a sequence of estimators which are dependent on the sample size. The \emph{asymptotic breakdown point for a sequence of estimators} $\{T_n\}_{n\in\N}$ at $q\in\Gamma$ is defined as:
\begin{equation*}\label{Dposterior:def:breakdown-point}
  B(\{T_n\}_{n\in\N},q)=\sup\{\alpha\in(0,1): \limsup_{n\to\infty} b(\alpha,T_n,q)<\infty\}.
\end{equation*}

Before we describe the behaviour of EDAP and MDAP with respect to the above measure, we introduce the following assumption:
\begin{enumerate}[label=(H\arabic*),ref=(H\arabic*),start=9]
\item The sequence of contaminating probability distributions $\{\bar{q}_L\}_{L\in\N_0}$ satisfies that for any $\alpha\in\left(0,1\right)$ there exists $\delta>0$ such that\label{Dposterior:cond:thm-convergence-Dposterior}
\begin{equation*}
\inf_{L\in\N_0}\inf_{\theta\in\Theta} D(\alpha \overline{q}_L,\theta)>\inf_{\theta\in\Theta} D((1-\alpha)p,\theta)+\delta.
\end{equation*}
\end{enumerate}

Our next theorem describes the asymptotic breakdown point of EDAP and MDAP estimators. To that end, we need some additional regularity conditions on the parametric family; indeed, assume condition (A4) -alternatively, (A6) for the Hellinger distance- in \cite{art-MDE}.

\begin{thm}\label{Dposterior:thm:asympt-breakdown}
Assume that \ref{Dposterior:eq:con:bound-disp}~\ref{Dposterior:cond:G-disparity}, \ref{Dposterior:cond:consistencyMLE}, \ref{Dposterior:alimit}~\ref{Dposterior:cond:3-deriv-prior}, \ref{Dposterior:cond:contam-distr}~\ref{Dposterior:cond:ortog-distrib-family}  hold true. Moreover, assume that \ref{Dposterior:alimit}~\ref{Dposterior:cond:MDE-ex-uniq-int}-\ref{Dposterior:cond:2-deriv-MDE} hold for any $q\in\Gamma^*$. For any family of contaminating distributions $\{\bar{q}_L\}_{L\in\N_0}$ satisfying \ref{Dposterior:cond:contam-distr}~\ref{Dposterior:cond:ortog-distrib-contam}-\ref{Dposterior:cond:ortog-contam-family} and under \ref{Dposterior:cond:bound-RAF} and \ref{Dposterior:cond:thm-convergence-Dposterior} the following results hold.
\begin{enumerate}[label=(\roman*),ref=\emph{(\roman*)}]
\item The asymptotic breakdown point of $\{\overline{T}_n\}_{n\in\N}$ at $p$ is equal to the breakdown point of $T$ at $p$. \label{Dposterior:thm:asympt-breakdown-EDAP}
\item If $\Theta$ is complete and separable, \ref{Dposterior:alimit}~\ref{Dposterior:cond:bound-prior} and \ref{Dposterior:alimit}~\ref{Dposterior:cond:uniq-MDAP} holds for any $q\in\Gamma^*$, then the asymptotic breakdown point of $\{\widetilde{T}_n\}_{n\in\N}$ at $p$ is equal to the breakdown point of $T$ at $p$. \label{Dposterior:thm:asympt-breakdown-MDAP}
\end{enumerate}
The result also holds for the Hellinger distance.
\end{thm}

While the above results are concerned with EDAP and MDAP estimators, in the context of the methods proposed in the manuscript it is important to know the asymptotic behaviour of the corresponding $D$-posterior density functions at a contaminating model. Our next result addresses this issue. Specifically,  given $\alpha\in (0,1)$, and a family of contaminating distributions $\{\bar{q}_L\}_{L\in\N_0}$, the following result establishes the asymptotic behaviour, as $L\to\infty$, of the $D$-posterior density function at $(1-\alpha)p+\alpha \overline{q}_L$.

\begin{thm}\label{Dposterior:thm:convergence-Dposterior}
Let $\alpha\in (0,1)$. Assume that \ref{Dposterior:eq:con:bound-disp}~\ref{Dposterior:cond:G-disparity},  \ref{Dposterior:cond:consistencyMLE},  and \ref{Dposterior:cond:contam-distr}~\ref{Dposterior:cond:ortog-distrib-family} hold true, and \ref{Dposterior:alimit}~\ref{Dposterior:cond:MDE-ex-uniq-int}-\ref{Dposterior:cond:2-deriv-MDE} hold for any $q\in\Gamma^*$. If $\{\bar{q}_L\}_{L\in\N_0}$ is a family of contaminating distribution satisfying \ref{Dposterior:cond:contam-distr}~\ref{Dposterior:cond:ortog-distrib-contam}-\ref{Dposterior:cond:ortog-contam-family}, then under \ref{Dposterior:cond:bound-RAF} and \ref{Dposterior:cond:thm-convergence-Dposterior},
\begin{enumerate}[label=(\roman*),ref=\emph{(\roman*)}]
\item $\lim_{L\to\infty}\overline{\pi}_D^n(\theta | (1-\alpha)p+\alpha \overline{q}_L)=\overline{\pi}_D^n(\theta|(1-\alpha)p).$\label{Dposterior:thm:convergence-Dposterior-i}
\item $\lim_{L\to\infty} \int_\Theta|\overline{\pi}_D^n(\theta | (1-\alpha)p+\alpha \overline{q}_L)-\overline{\pi}_D^n(\theta|(1-\alpha)p)|d\theta=0.$\label{Dposterior:thm:convergence-Dposterior-ii}
\end{enumerate}
The result also holds for the Hellinger distance.
\end{thm}


\section{Examples}\label{Dposterior:sec:examples}

In this section, we illustrate the methodology proposed via several examples. The first example is based on a real data set in the context of cell kinetics populations. In this example, we show that the EDAP and MDAP estimators based on different disparity measures are as good as the ones based on the Kullback-Leibler divergence in absence of contamination. The second example of this section, which is a continuation of Example \ref{Dposterior:ex:LD} introduced in Section \ref{Dposterior:sec:method}, demonstrates the accuracy of the method in a contamination context. Thus, the twofold goodness (in contamination and free-contamination contexts) of EDAP and MDAP estimators is exemplified. In both cases, the results are illustrated by considering the Hellinger distance and the negative exponential disparity.

\subsection{Example - Oligodendrocytes}\label{Dposterior:ex:oligo}

In this example, we consider data from a study on cell proliferation published by \cite{Hyrien-AMNY2006} and which have been also examined in  \cite{Yakovlev-Stoimenova-Yanev-2008}. Among the aims of both papers is to model the proliferation of oligodendrocyte precursor cells,  and their transformation into terminally differentiated oligodendrocytes through multitype age-dependent branching processes.

As mentioned above, the cell populations considered consist of two types of cells: the oligodendrocyte precursor cells, referred as type $T_1$ cells, and the terminally differentiated oligodendrocytes, referred as type $T_2$ cells. The development of both types of cells is as follows: type $T_1$ cells can die without any offspring or divide under normal conditions; when stimulating  to division, precursor cells are capable of producing either direct progeny (two daughter cells of the same type) or a single terminally differentiated nondividing oligodendrocyte. The data are supplied by time-lapse video recording of oligodendrocyte populations. These experimental techniques enable to record all observable events (division, differentiation, or death), as well as their timing, in the development of each individual cell. Due to the design of these experimental techniques it is reasonable to consider absence of contamination of the sample observed (the multi-daughter cell division can be observed although rarely in cancer scenarios - see \cite{tse-weaver-dicarlo-2012} - which is not the case). A special feature is the presence of censoring effects due to migration of precursor cells out of the microscopic field of observation, modelled as a process of emigration of the type $T_1$ cells.

In a frequentist context and in order to estimate the reproduction probabilities, in \cite{Yakovlev-Stoimenova-Yanev-2008} the authors proposed estimators for the offspring distribution by using the discrete time branching process obtained by embedding the branching structure of the aforesaid continuous-time branching process. It is important to note that cell emigration is artificially incorporated into that discrete-time model via a modified reproduction law. As a more natural alternative, we propose a  two-type controlled  branching process to describe the embedded discrete branching structure of the age-dependent branching process previously quoted. In the framework given by this process we deal with the problem of estimating the offspring distribution of the cell population from a Bayesian standpoint by making use of disparity measures.

We introduce a two-type controlled branching process with binomial control, which is a particular case of that introduced in \cite{Gonzalez-Martinez-Mota-2005}.
Indeed, we propose the following controlled two-type  branching process, $\{Z_n\}_{n\in\N_0}$, defined as follows:
\begin{align*}
Z_0=(N,0),\quad Z_{n+1}= \sum_{j=1}^{\phi_n(Z_n)}(X_{n,j}^{(1)},X_{n,j}^{(2)}),\quad n\in\N_0,
\end{align*}
where $Z_n=(Z_n^{(1)},Z_n^{(2)})$, $N\in\N$, and $\{(X_{n,j}^{(1)},X_{n,j}^{(2)}): j\in\N, n\in\N_0\}$ and $\{\phi_n(z): n\in\N_0, z\in\N_0^2\}$ are two independent families of non-negative integer valued random variables satisfying the following conditions:
\begin{enumerate}[label=(\roman*),ref=(\roman*),start=1]
\item For each $z=(z^{(1)},z^{(2)})\in\N_0^2$, the random variables $\{\phi_n(z): n\in\N_0\}$ are i.i.d. following a binomial distribution with parameters $z^{(1)}$ and $\gamma\in (0,1)$.

\item The stochastic processes $\{(X_{n,j}^{(1)},X_{n,j}^{(2)}): j\in\N\}$, $n\in\N_0$ are i.i.d. following the probability distribution
\begin{align*}
p_0&=P\left[X_{n,j}^{(1)}=0,X_{n,j}^{(2)}=0\right],\\
p_1&=P\left[X_{n,j}^{(1)}=2,X_{n,j}^{(2)}=0\right],\\
p_2&=P\left[X_{n,j}^{(1)}=0,X_{n,j}^{(2)}=1\right].
\end{align*}

\item If $n_1,n_2\in\N_0$ are such that $n_1\neq n_2$, then, the sequences $\{(X_{n_1,j}^{(1)},X_{n_1,j}^{(2)}): j\in\N\}$ and  $\{(X_{n_2,j}^{(1)},X_{n_2,j}^{(2)}): j\in\N\}$ are independent.
\end{enumerate}
Intuitively, $Z_n^{(j)}$ denotes the number of cells of type $T_j$, $j=1,2$, in the $n$-th generation, and $X_{n,j}^{(1)}$ and $X_{n,j}^{(2)}$ represent the number of cells of type $T_1$ and $T_2$, respectively, produced by the $j$-th progenitor of type $T_1$ in the generation $n$. The random variables $\phi_n(z)$ are introduced to model the cell emigration, and therefore to determine the number of progenitors of type $T_1$ cells in the $n$-th generation provided the population size at that generation is the vector $z$. Consequently, $\gamma$ is the probability that a cell of type $T_1$ completes successfully its mitotic-cycle regardless of its outcome (and then $1-\gamma$ is the probability of emigration of that $T_1$ cell).

We focus our attention on the offspring distribution, $\{p_0,p_1,p_2\}$, where $p_0$ is the probability that a cell of type $T_1$ dies, $p_1$ is the probability that a cell of type $T_1$ divides into two cells of type $T_1$ and $p_2$ is the probability that a cell of type $T_1$ differentiates into a new cell of type $T_2$. Using the time-lapse video recording, one can observe the entire family tree for that process, that is, the following sample $\mathcal{Z}_n^*=\{Z_l^{(1)}(0),Z_l^{(2)}(2),\Lambda_l:\ l=0,\ldots,n-1\}$, where
\begin{equation*}
Z_l^{(1)}(j)=\sum_{i=1}^{\phi_l(Z_l)} I_{\left\{X_{l,i}^{(1)}=j\right\}},\quad j=0,2;\qquad \Lambda_l=\sum_{i=1}^{\phi_l(Z_l)} I_{\left\{X^{(2)}_{l,i}=1\right\}}.
\end{equation*}

Note that $\phi_l(Z_l)=Z_l^{(1)}(0)+Z_l^{(1)}(2)+\Lambda_l$, $Z_l^{(1)}(j)$ is the number of cells of type $T_1$ in the $l$-th generation having exactly $j$ offspring of type $T_1$, $j=0,2$, and $\Lambda_l$ is the number of cells of type $T_1$ in the $l$th generation having exactly one offspring of type $T_2$. Our aim is to estimate the offspring distribution by using disparity measures, to provide HPD regions for the corresponding $D$-posterior densities and to estimate the probability that the production of precursor cells follows a supercritical probability distribution.

Using Markov's property and the independence between the emigration and reproduction phases, one can prove that the likelihood function  $f(\mathcal{Z}_n^*|p_0,p_1,\gamma)$ is given by
\begin{equation*}
f(\mathcal{Z}_n^*|p_0,p_1,\gamma)\propto p_{0}^{Y_{n-1}^{(1)}(0)} p_{1}^{Y_{n-1}^{(1)}(2)} (1-p_0-p_1)^{\Psi_{n-1}}\gamma^{\Delta_{n-1}} (1-\gamma)^{Y_{n-1}^{(1)}-\Delta_{n-1}},
\end{equation*}
where $Y_{n-1}^{(1)}(j)=\sum_{l=0}^{n-1} Z_l^{(1)}(j)$ is the total number of cells of type $T_1$ in the first $n$ generations having exactly $j$ offspring of type $T_1$, $j=0,2$, $\Psi_{n-1}=\sum_{l=0}^{n-1}\Lambda_l$ is the total number of cells of type $T_1$ in the first $n$ generations having exactly one offspring of type $T_2$, $\Delta_{n-1}=\sum_{l=0}^{n-1} \phi_l(Z_l)$ is the total number of observed progenitor cells of type $T_1$ in the first $n$ generations and $Y_{n-1}^{(1)}=\sum_{l=0}^{n-1} Z_l^{(1)}$ is the total number of individuals of type $T_1$ in the first $n$ generations. Moreover, observe that $\Delta_{n-1}=Y_{n-1}^{(1)}(0)+Y_{n-1}^{(1)}(2)+\Psi_{n-1}$. Thus, the MLEs of $p_0$, $p_1$, $p_2$ and $\gamma$ based on the sample $\mathcal{Z}_n^*$ are given, respectively, by
\begin{equation*}
\widehat{p}_{0,n}=\frac{Y_{n-1}^{(1)}(0)}{\Delta_{n-1}},\quad \widehat{p}_{1,n}=\frac{Y_{n-1}^{(1)}(2)}{\Delta_{n-1}},\quad \widehat{p}_{2,n}=\frac{\Psi_{n-1}}{\Delta_{n-1}},\quad \widehat{\gamma}=\frac{\Delta_{n-1}}{Y_{n-1}^{(1)}}.
\end{equation*}

Let write $p=\{p_0,p_1,p_2\}$, $\widehat{p}_n=\{\widehat{p}_{0,n},\widehat{p}_{1,n},\widehat{p}_{2,n}\}$,  $q=\{q_0,q_1,q_2\}$ representing a probability distribution on the state space $\{(0,0),(2,0),(0,1)\}$, $\pi(\cdot)$ a prior distribution on the space of the probability distributions defined on such a state space and  $\theta=(q_0,q_1)$. With analogous arguments as those in Section \ref{Dposterior:sec:method}, we define the $D$-posterior density function of $\theta=(q_0,q_1)$ as
\begin{equation}\label{Dposterior:eq:D-posterior-2type}
\pi_D^n(\theta|\hat{p}_n)=\frac{e^{-\Delta_{n-1} D(\hat{p}_n,\theta)}\pi(\theta)}{\int_\Theta e^{-\Delta_{n-1} D(\hat{p}_n,\theta)}\pi(\theta)d\theta},
\end{equation}
where $\int_\Theta e^{-\Delta_{n-1} D(\hat{p}_n,\theta)}\pi(\theta)d\theta=\int_\Theta e^{-\Delta_{n-1} D(\hat{p}_n,q)}\pi(q)dq_0 dq_1$ and $\Theta=\{(x_0,x_1)\in (0,1)\times (0,1):x_0+x_1\leq 1\}$. Note that in this case, the parametric family which we consider is $\mathcal{F}_\Theta=\{\{q_0,q_1,1-q_0-q_1\}: (q_0,q_1)\in\Theta\}$. As a consequence, we propose as EDAP estimators of $\theta_0=(p_0,p_1)$ the following ones:
\begin{equation}\label{Dposterior:eq-EDAP-oligo-sim}
(p_{0,n}^{*D},p_{1,n}^{*D})=\left(\int_0^1 q_0 \pi_D^n(q_0|\hat{p}_n)dq_0,\int_0^1 q_1 \pi_D^n(q_1|\hat{p}_n)dq_1\right),
\end{equation}
and as MDAP estimors:
\begin{equation*}
(p_{0,n}^{+D},p_{1,n}^{+D})=\arg\max_{\theta\in\Theta}\pi_D^n(\theta|\hat{p}_n),
\end{equation*}
where $\pi_D^n(\theta|\hat{p}_n)$ was introduced in \eqref{Dposterior:eq:D-posterior-2type}, and $\pi_D^n(q_0|\hat{p}_n)$ and $\pi_D^n(q_1|\hat{p}_n)$ are the marginal density functions of $q_0$ and $q_1$, respectively.

The data that we consider are provided from two experiments developed under two different experimental conditions: in the first one, the cells were cultured in a vehicle solution without agent stimulating differentiation, whereas in the second experiment, they were cultured with a specific hormone promoting oligodendrocyte generation. The first experiment started with 34 precursor cells cultured in a vehicle solution (with no agent promoting oligodendrocyte generation) and whose evolution was observed until the generation $n=7$. On the other hand, the initial number of cells of type $T_1$ in experiment 2 was 30, which were cultured in solution with a specific hormone promoting oligodendrocyte generation and whose branching tree was observed until the generation $n=5$. Based on the observed branching trees, we computed the non-parametric MLE estimators for the parameters related to the reproduction and the emigration processes, which are given in Table \ref{Dposterior:table:ex-yanev-MLE}.

\begin{table}[H]
\centering\begin{tabular}{|c|ccccccc|}
\cline {2-8}
\multicolumn{1}{c|}{ } & $n$ & $N$ &  $Y_{n}^{(1)}$ & $\Delta_n$ & $Y_{n-1}^{(1)}(0)$ & $Y_{n-1}^{(1)}(2)$ & $\Psi_n$ \\
\hline
Experiment 1 & 7 & 34 & $425$ & $410$ & $158$ & $201$ & $51$ \\
Experiment 2 & 5 & 30 & $276$ & $269$ & $37$ & $133$ & $99$ \\
\hline
\end{tabular}
 \caption{Observed branching tree for each experiment.} \label{Dposterior:table:ex-yanev-data}
\end{table}

\begin{table}[H]
\centering\begin{tabular}{|c|ccccc|}
\cline{2-6}
\multicolumn{1}{c|}{ } & $n$ & $\widehat{p}_{0,n}$ & $\widehat{p}_{1,n}$ & $\widehat{p}_{2,n}$ & $\widehat{\gamma}_n$ \\
\hline
Experiment 1 & 7 & $0.3854$ & $0.4902$ & $0.1244$ & $0.9647$ \\
Experiment 2 & 5 & $0.1375$ & $0.4944$ & $0.3680$ & $0.9746$ \\
\hline
\end{tabular}
\caption{MLE for the main parameters of the model.} \label{Dposterior:table:ex-yanev-MLE}
\end{table}

Using the MLE of $p$, $\hat{p}_n=\{\widehat{p}_{0,n},\widehat{p}_{1,n},\widehat{p}_{2,n}\}$, for each experiment, we determined the $D$-posterior density functions of $(q_0,q_1)$ at $\hat{p}_n$ considering the Hellinger distance, the negative exponential disparity and the Kullback-Leibler divergence. In all the cases, we consider the Dirichlet distribution with parameter $(1/2,1/2,1/2)$ as the prior distribution due to the fact that no prior knowledge about the offspring distribution - apart from its support - is available (see \cite{Berger-Bernardo-1992}). The contour plots of these $D$-posterior density functions for experiments 1 and 2 are represented in Figures \ref{Dposterior:fig:contour1} and \ref{Dposterior:fig:contour2}, respectively, together with the 95\% HPD regions for $(q_0,q_1)$. We also calculated the EDAP and MDAP estimators of the parameters $p_0$ and $p_1$ for the aforementioned disparity measures, which are presented in Tables \ref{Dposterior:table:ex-yanev-exp-1} and \ref{Dposterior:table:ex-yanev-exp-2} for the first and the second experiment, respectively. Note that, these estimates are close to the non-parametric estimators in Table \ref{Dposterior:table:ex-yanev-MLE}. Moreover, although the results are not presented, it is important to mention that a sensitivity analysis was performed in order to examine the influence of the parameter of the Dirichlet distribution on the estimates; however, no significant difference is observed on the estimates due to the change of such a parameter (see Supplementary material for further details). The analogous behaviour of the three estimators suggests that the data have not significant (if any) contamination and that our HD and NED estimators are as good as the one based on the KL divergence in absence of contamination.

\begin{figure}[H]
\centering\includegraphics[width=0.3\textwidth]{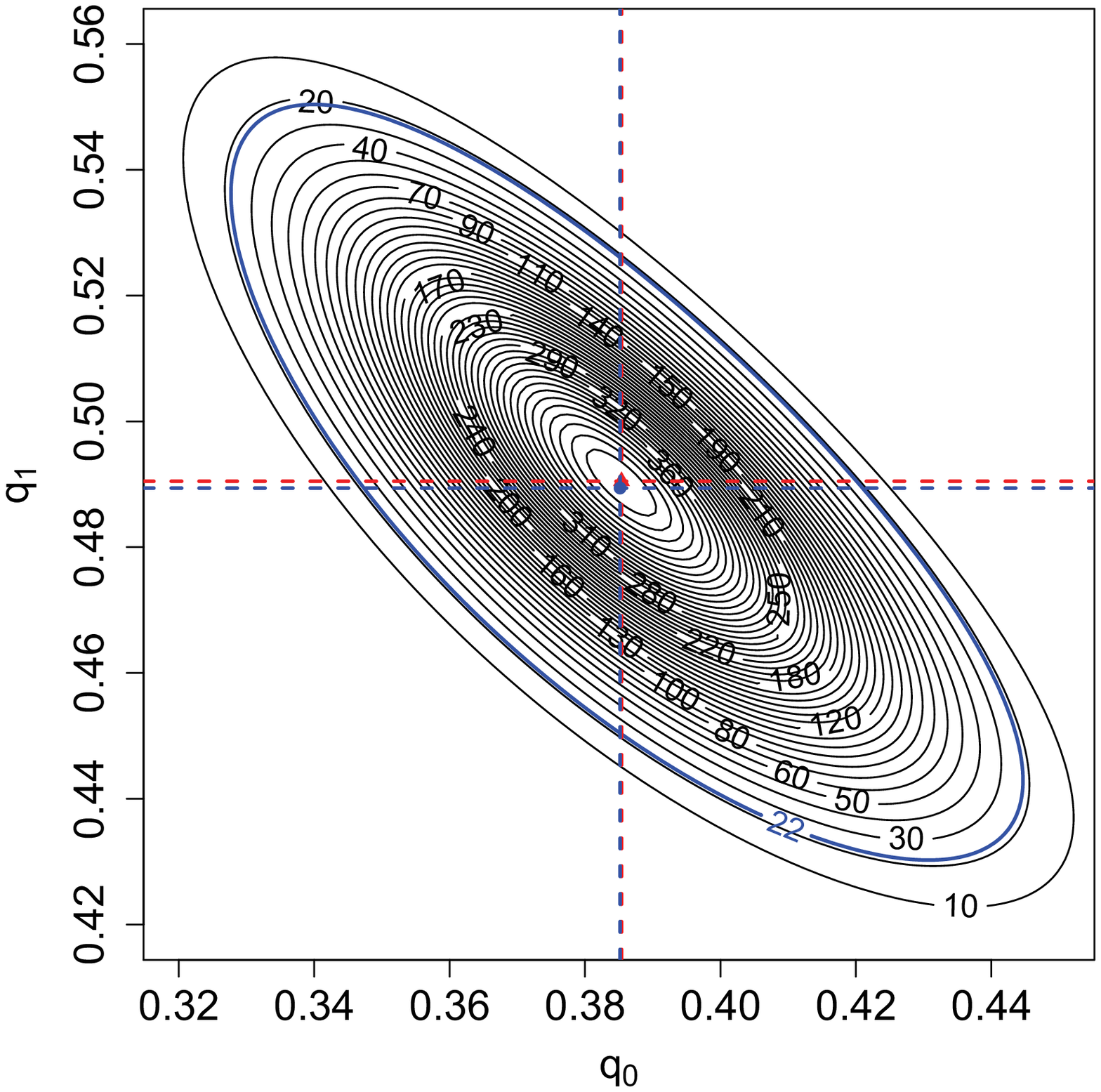}\hspace*{0.03\textwidth}
\includegraphics[width=0.3\textwidth]{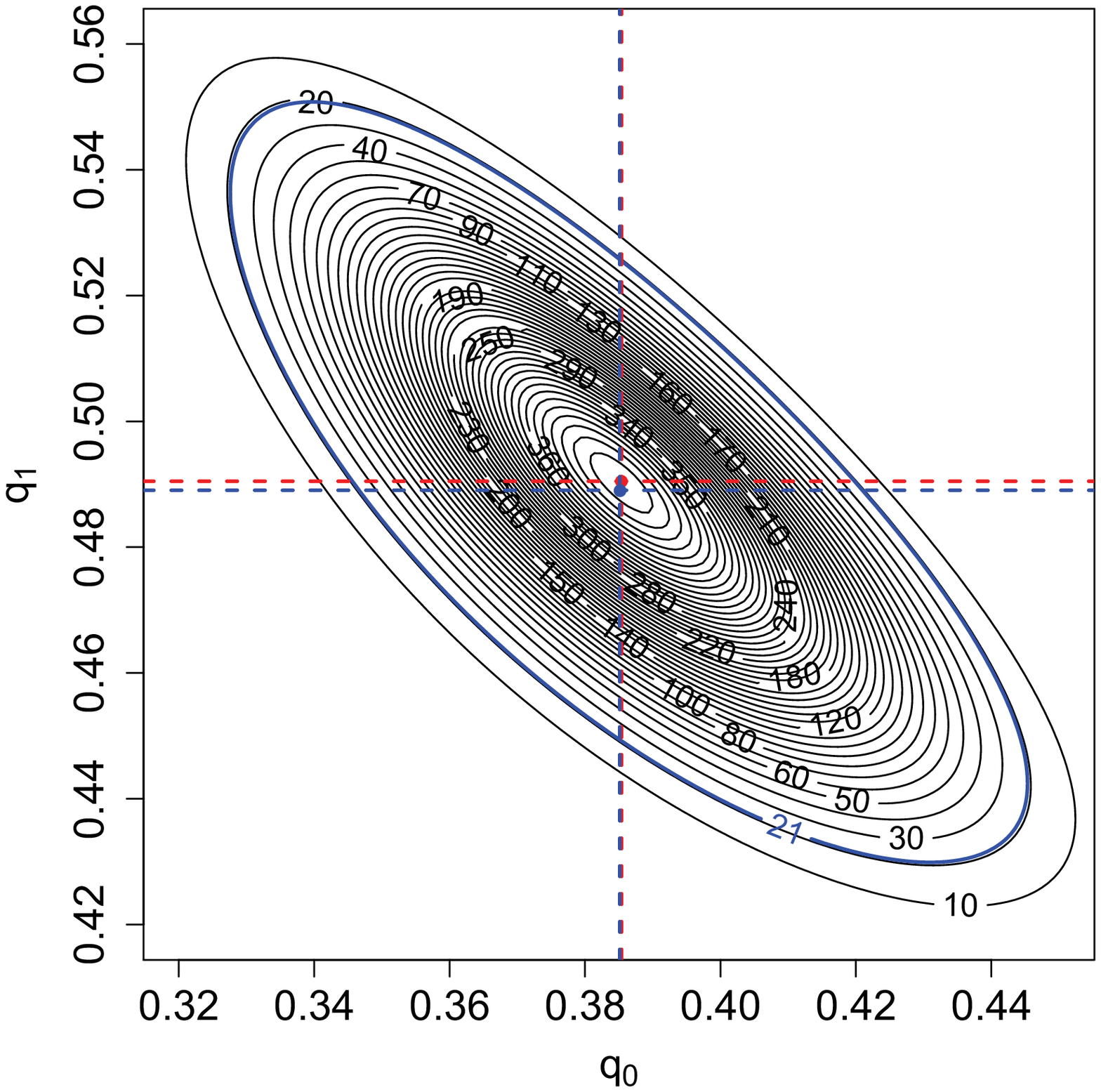}\hspace*{0.03\textwidth}
\includegraphics[width=0.3\textwidth]{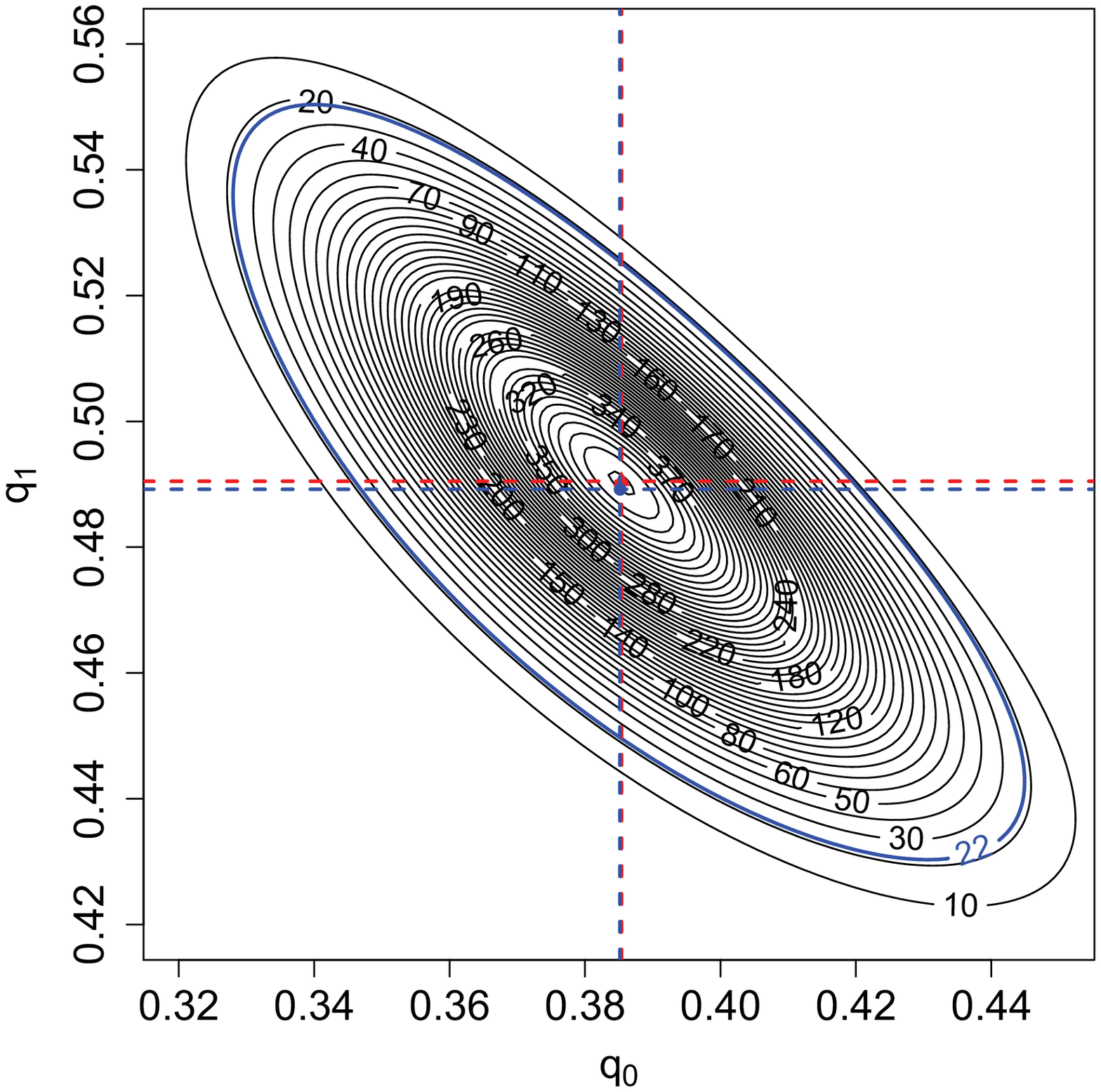}\caption{Contour plot for the $D$-posterior density of $(q_0,q_1)$ of the first experiment, together with the EDAP estimate (intersection of blue dashed lines) and the MDAP estimate (intersection of red dashed lines). Blue contour line represents the 95\% HPD region. Left: HD. Centre: NED. Right: KL.}\label{Dposterior:fig:contour1}
\end{figure}

\begin{figure}[H]
\centering\includegraphics[width=0.3\textwidth]{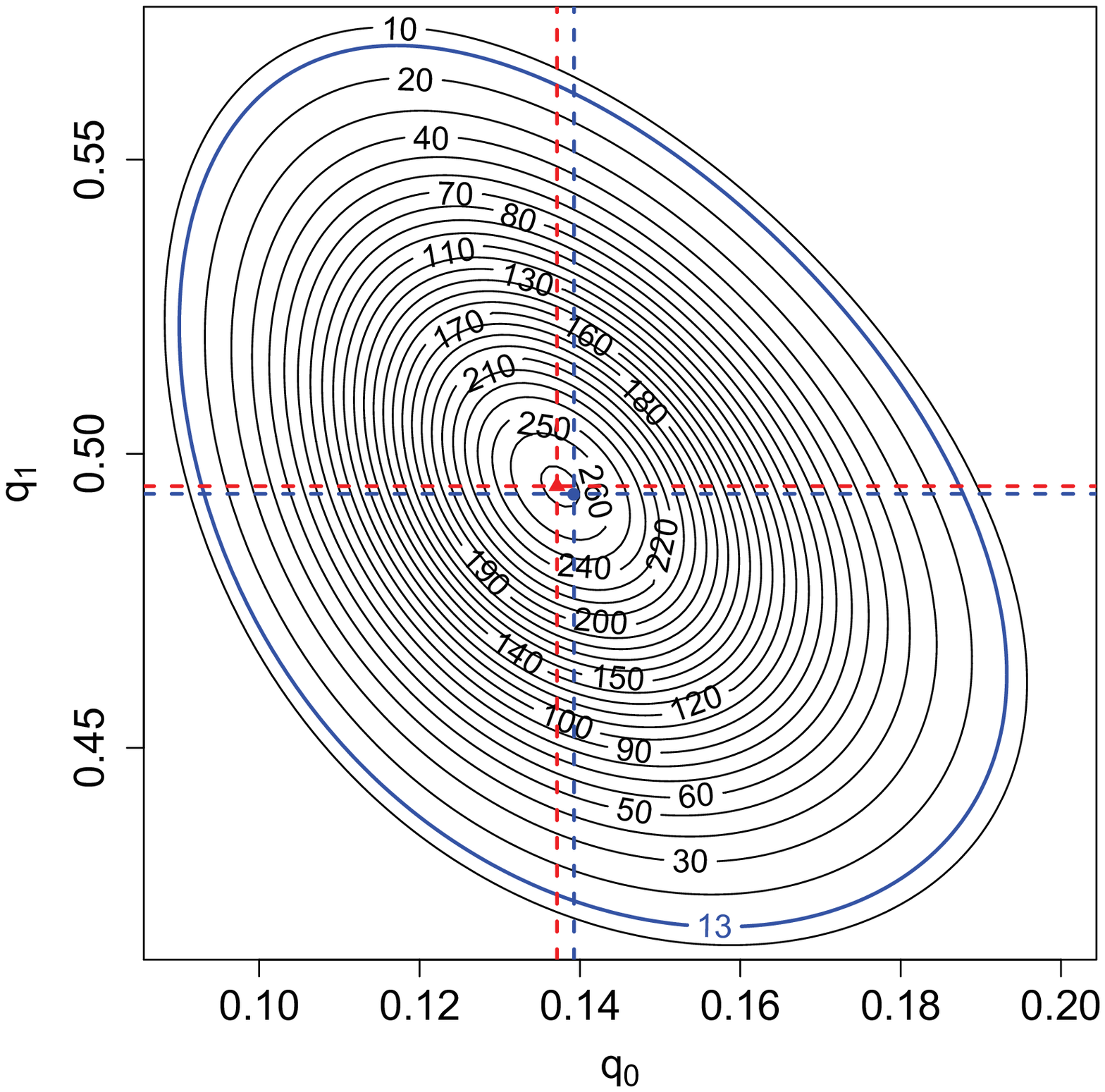}\hspace*{0.03\textwidth}
\includegraphics[width=0.3\textwidth]{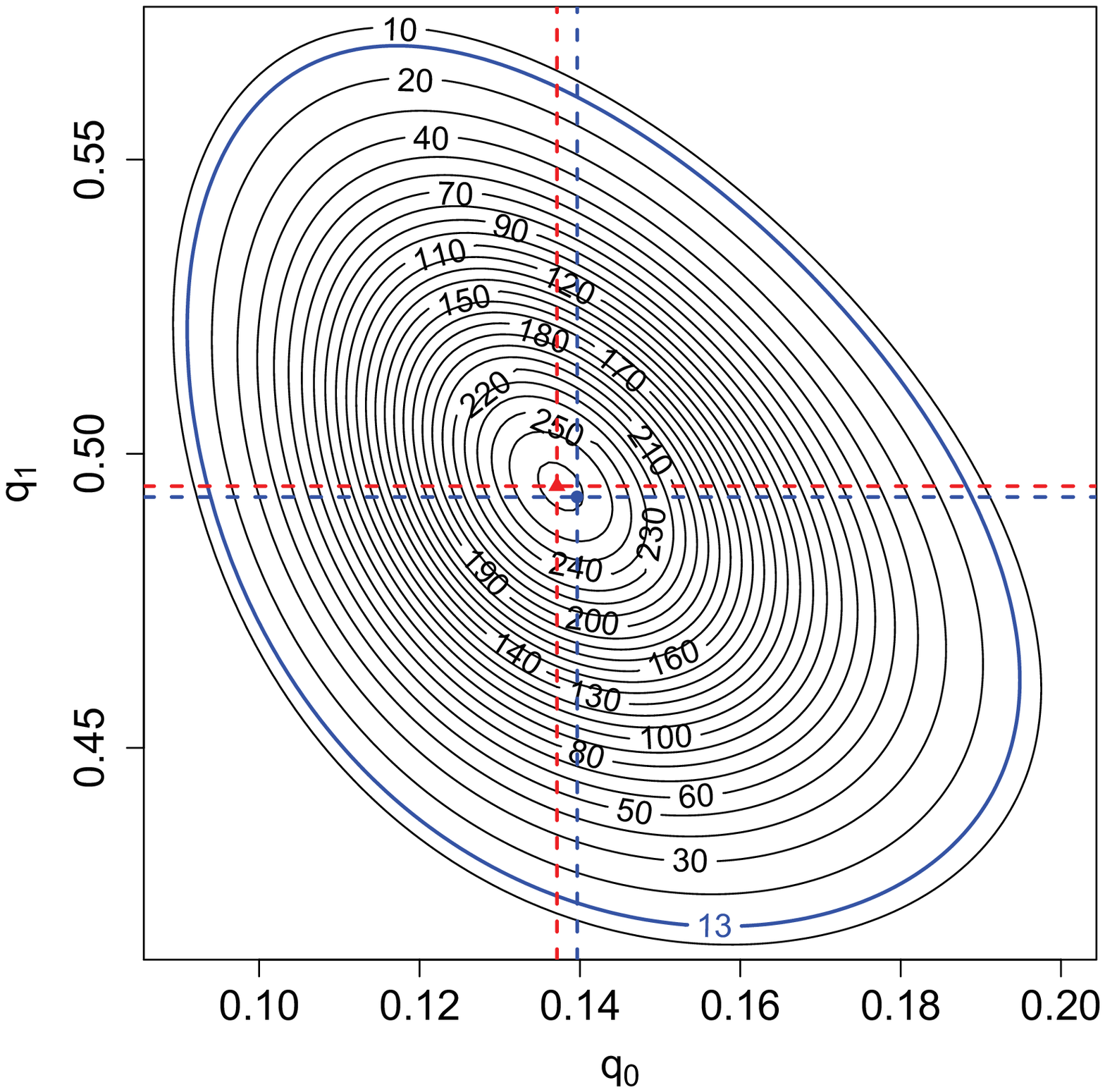}\hspace*{0.03\textwidth}
\includegraphics[width=0.3\textwidth]{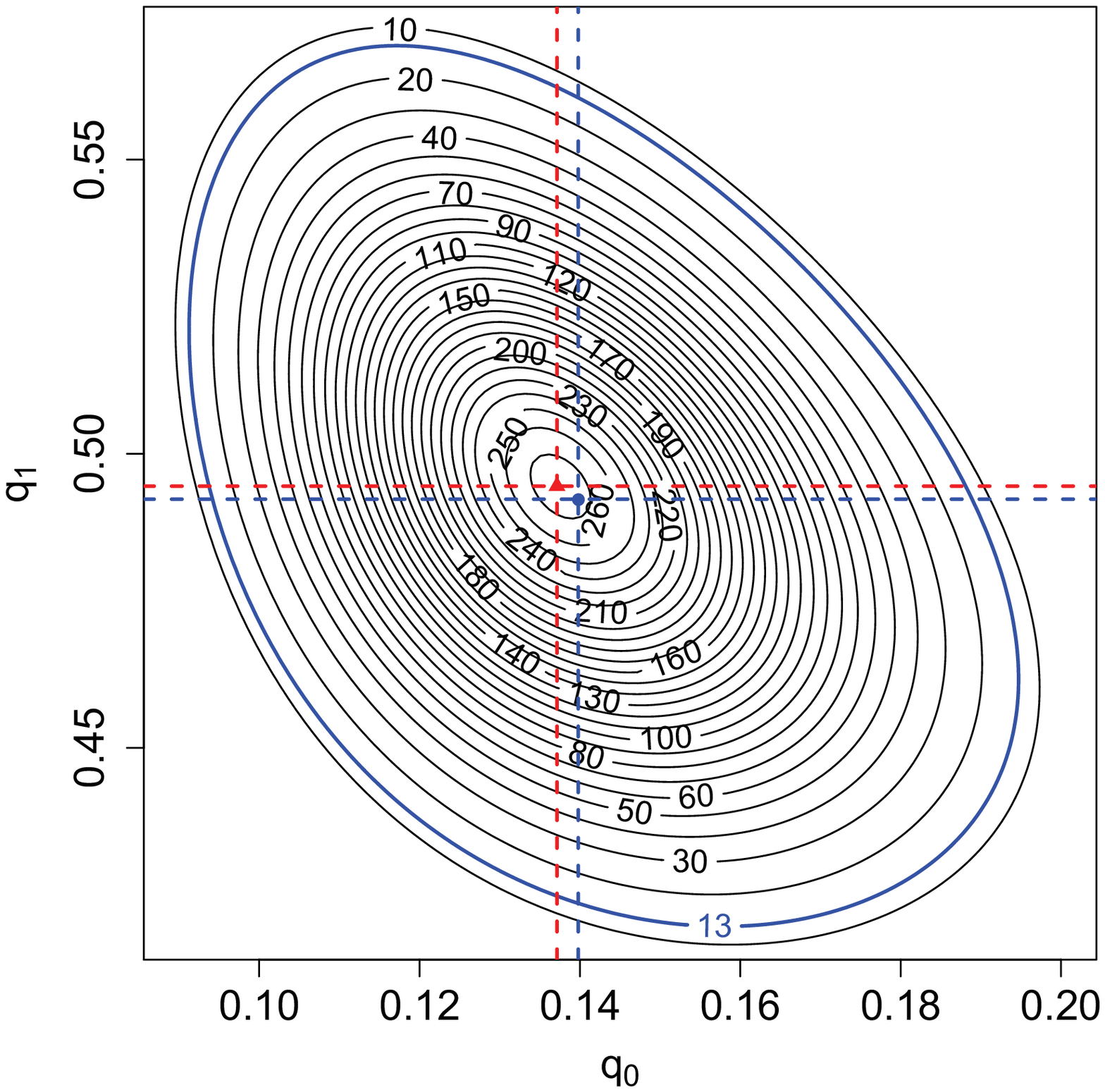}\caption{Contour plot for the $D$-posterior density of $(q_0,q_1)$ of the second experiment, together with the EDAP estimate (intersection of blue dashed lines) and the MDAP estimate (intersection of red dashed lines). Blue contour line represents the 95\% HPD region. Left: HD. Centre: NED. Right: KL.}\label{Dposterior:fig:contour2}
\end{figure}

\begin{table}[H]
\centering\begin{tabular}{|c|cc|cc|cc|}
\cline{2-7}
\multicolumn{1}{c|}{ }& \multicolumn{2}{|c|}{HD}& \multicolumn{2}{|c|}{NED}& \multicolumn{2}{|c|}{KL}\\
\hline
Parameter & EDAP & MDAP & EDAP & MDAP & EDAP & MDAP \\
\hline
$p_0$ & $0.3858$ & $0.3850$ & $0.3852$ & $0.3850$ & $0.3853$ & $0.3850$ \\
$p_1$ & $0.4895$ & $0.4910$ & $0.4897$ & $0.4910$ & $0.4893$ & $0.4910$ \\
$p_2$ & $0.1247$ & $0.1240$ & $0.1251$ & $0.1240$ & $0.1254$ & $0.1240$ \\
\hline
\end{tabular}
 \caption{Estimates of $p_0$, $p_1$ and $p_2$ based on the EDAP and MDAP estimators for HD, NED and KL in the first experiment.} \label{Dposterior:table:ex-yanev-exp-1}
\end{table}

\begin{table}[H]
\centering\begin{tabular}{|c|cc|cc|cc|}
\cline{2-7}
\multicolumn{1}{c|}{ }& \multicolumn{2}{|c|}{HD}& \multicolumn{2}{|c|}{NED}& \multicolumn{2}{|c|}{KL}\\
\hline
Parameter & EDAP & MDAP & EDAP & MDAP & EDAP & MDAP \\
\hline
$p_0$ & $0.1380$ & $0.1360$ & $0.1388$ & $0.1360$ & $0.1386$ & $0.1360$ \\
$p_1$ & $0.4939$ & $0.4960$ & $0.4935$ & $0.4960$ & $0.4934$ & $0.4960$ \\
$p_2$ & $0.3680$ & $0.3680$ & $0.3677$ & $0.3680$ & $0.3681$ & $0.3680$ \\
\hline
\end{tabular}
 \caption{Estimates of $p_0$, $p_1$ and $p_2$ based on the EDAP and MDAP estimators for HD, NED and KL in the second experiment.} \label{Dposterior:table:ex-yanev-exp-2}
\end{table}

An interesting issue to tackle in this Bayesian framework is to determine the probability of having a supercritical probability distribution governing the production of precursor cells. This is an important problem since if the mean of the variable $X^{(1)}_{0,1}$, $m=2p_1$, is less than 1, then the population of cells of type $T_1$ will become extinct with probability 1, so does the population of type $T_2$ cells. Using the approximation of the $D$-posterior distribution at $\hat{p}_n$ in both experiments, we present an estimation of the probability that $m>1$, with respect to the $D$-posterior distribution at $\hat{p}_n$, in Table \ref{Dposterior:table:ex-yanev-criticality}. Note that the aforementioned probability is greater in the second experiment and this could mean that the solution with the hormone is effective at promoting cell reproduction. Indeed, by comparing Tables \ref{Dposterior:table:ex-yanev-exp-1} and \ref{Dposterior:table:ex-yanev-exp-2}, one also sees that the probability that a cell of type $T_1$ dies with no offspring is smaller in the second experiment whereas the probability of differentiating into a type $T_2$ cell is greater.

\begin{table}[H]
\centering
\begin{tabular}{|c|cc|}
\cline{2-3}
\multicolumn{1}{c|}{ } & Experiment 1 & Experiment 2 \\
\hline
HD & $0.3424$ & $0.4234$ \\
NED & $0.3382$ & $0.4177$ \\
KL & $0.3380$ & $0.4174$ \\
\hline
\end{tabular}
 \caption{Probability that $m>1$, with respect to the $D$-posterior distribution at $\hat{p}_n$, for HD, NED, and KL.} \label{Dposterior:table:ex-yanev-criticality}
\end{table}

\begin{remark}
With regard to computational purposes, the integrals involved in this example have been approximated by Monte Carlo methods.
\end{remark}

\subsection{Simulated example}\label{Dposterior:ex:sim}

This example is a continuation of Example \ref{Dposterior:ex:LD} introduced in Section \ref{Dposterior:sec:method} and its purpose is to show the behaviour of EDAP and MDAP estimators in presence of contamination. Recall that the process starts with $Z_0=1$ individual, the offspring distribution is a geometric distribution with parameter $\theta_0=0.3$ and which is affected by outliers, as described in Example \ref{Dposterior:ex:LD}, and the control variables $\phi_n(k)$ follow Poisson distributions with mean $\lambda k$, for each $k\in\N_0$ and each $n\in\N_0$. It is important to mention that the geometric distribution is a natural offspring distribution to use in the context of branching processes. For instance, this distribution is proposed as reproduction law for the GWP resulting from embedding the continuous time Markov branching process applied to model data from a yeast cell colony in \cite{Guttorp}, p.158. Moreover, branching processes with a geometric offspring distribution have been used in other fields as, for example, Physics (e.g. see the recent paper \cite{corral-garcia-font-2016}). Observe that we consider an extreme contamination framework; indeed, in this process, each progenitor has exactly 11 offspring with probability 0.15 and, with probability 0.85, it has offspring according to the aforementioned geometric distribution. From a practical viewpoint, this might seem to be too much extreme example, however, its choice is motivated by the fact that it allows to illustrate appropriately the accuracy of the method in a contamination context.

In order to approximate the $D$-posterior density functions based on the Hellinger distance (HD) and the negative exponential disparity (NED), one can consider different approaches. One of them is to obtain a sample of such a density function by applying the Metropolis-Hastings algorithm, and use this sample to estimate the corresponding density function (this approach was used in \cite{Hooker-Vidyshankar-2014}). However, since an expression of the $D$-posterior density function is available (see \eqref{Dposterior:eq:D-posterior-tree}), we have opted for estimating both the $D$-posterior density function and the EDAP and MDAP estimators by using numerical methods. The reason for the choice of the latter approach is that it is computationally quite faster than the former; this fact is especially remarkable when considering a disparity measure different from the KL disparity since in that case it is not possible to use a conjugate families of distributions (see \cite{Hooker-Vidyshankar-2014}).

First, in Figure \ref{Dposterior:fig:estim-D-posterior} we show the estimates of the HD-posterior density and NED-posterior density upon the sample $z_{45}^*$ and a beta distribution with parameters $1/2$ and $1/2$ as a  prior distribution. Recall that the choice of this prior distribution is motivated by the fact we are not assuming any prior knowledge about the offspring parameter. One can observe that both functions are approximately centred around the true value of the offspring parameter and the EDAP and MDAP estimators based on both of them provide accurate estimates of this parameter: $\theta_{45}^{*HD}=0.2962$ and $\theta_{45}^{+HD}=0.2953$, in the case of HD, and $\theta_{45}^{*NED}=0.2953$ and $\theta_{45}^{+NED}=0.2940$, for the NED.

\begin{figure}[H]
\centering\includegraphics[width=0.35\textwidth]{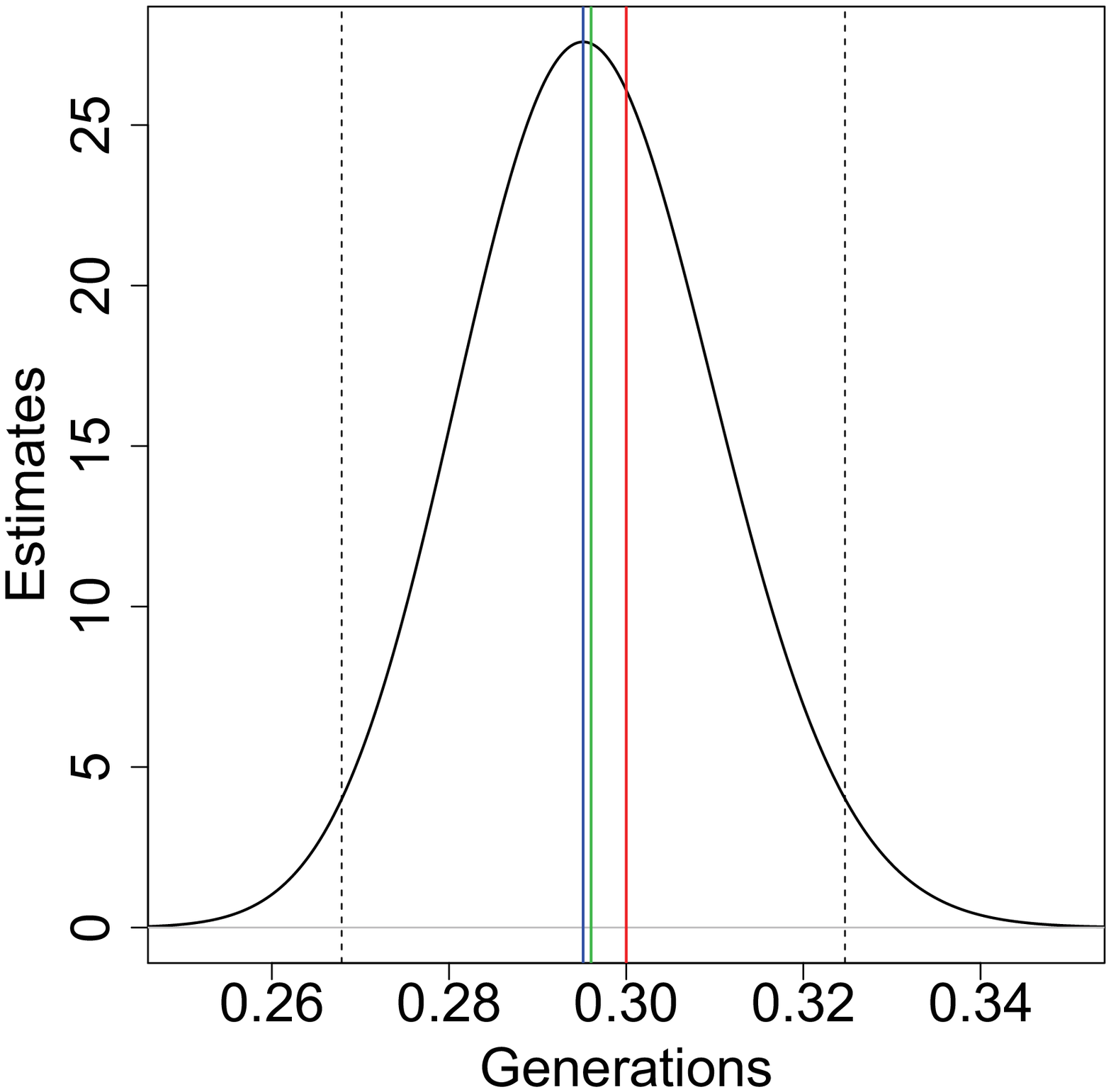}\hspace*{0.03\textwidth}
\includegraphics[width=0.35\textwidth]{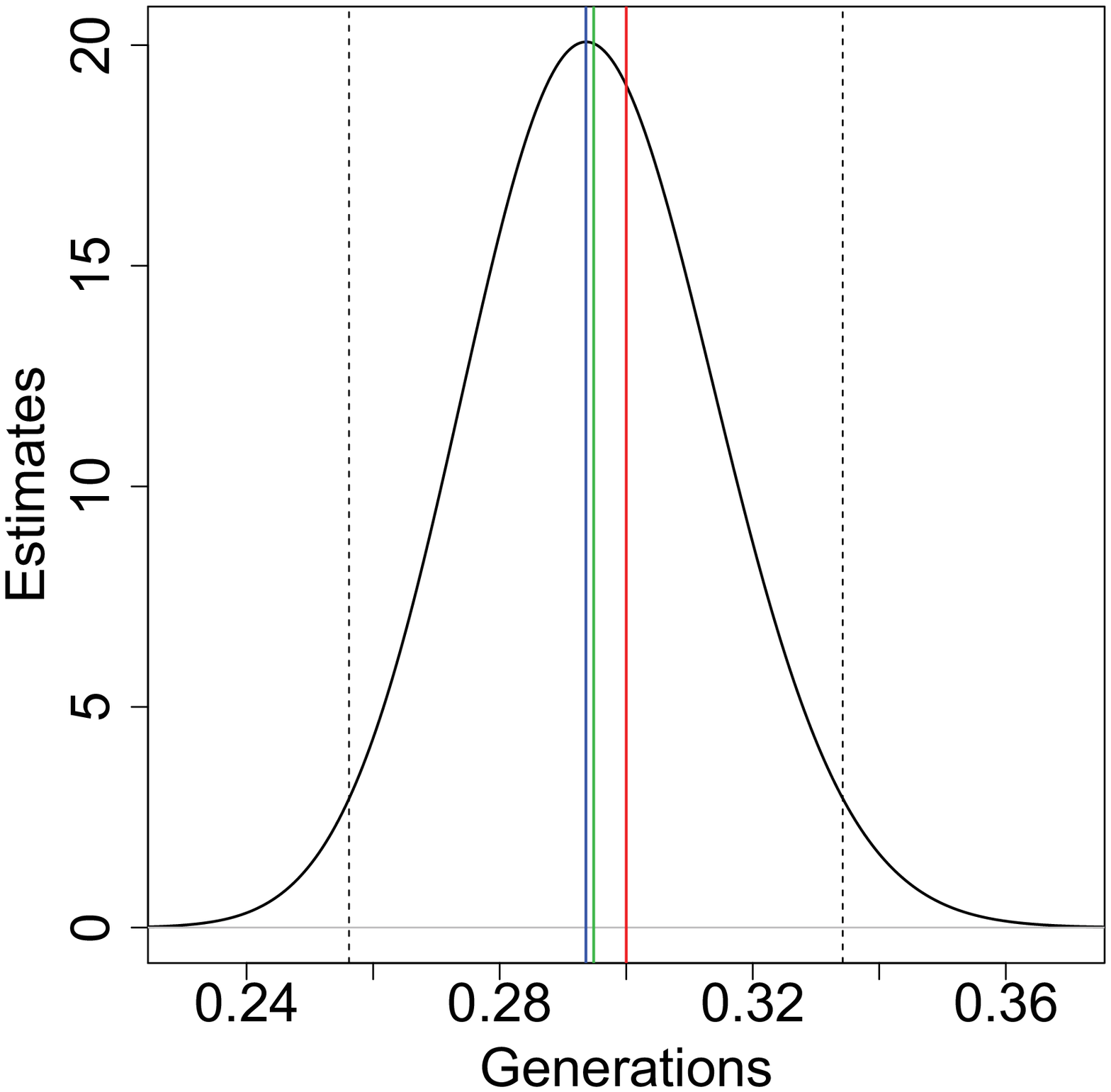}
\caption{Estimate of the D-posterior density of $\theta$ given the sample $z_{45}^*$, together with EDAP (blue line) and MDAP (green line), HPD interval (dashed line) and true value of $\theta_0$ (red line). Left: HD. Right: NED.}\label{Dposterior:fig:estim-D-posterior}
\end{figure}

The strong consistency of EDAP and MDAP estimates based on HD and NED are illustrated in Figure \ref{Dposterior:fig:consistency-EDAP} and \ref{Dposterior:fig:consistency-MDAP}, respectively. The evolution of HPD intervals are also plotted. These estimates have shown to be accurate in contrast to those plotted in Figure \ref{Dposterior:fig:cont-pop-post-KL} (right).

\begin{figure}[H]
\centering\includegraphics[width=0.35\textwidth]{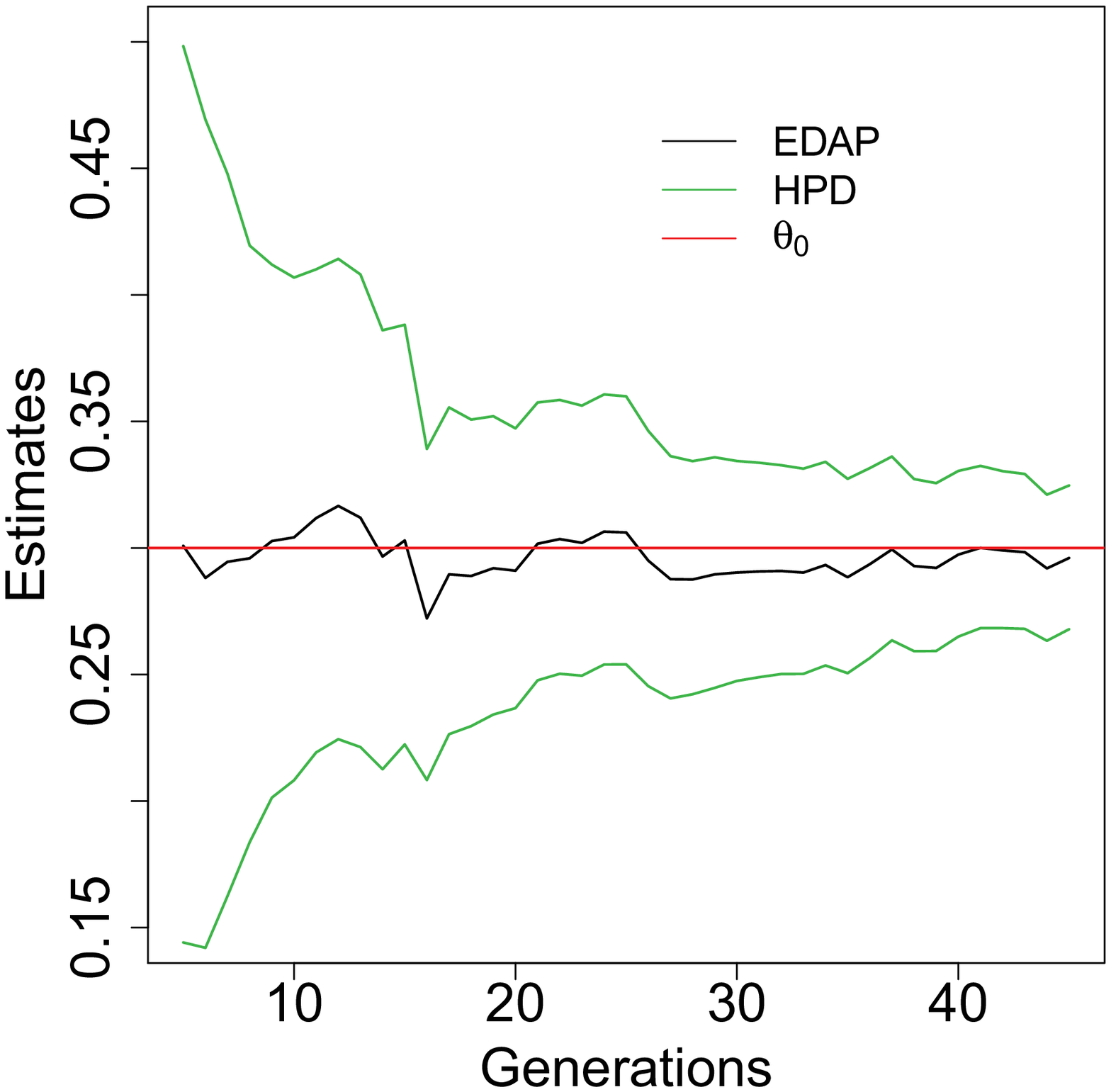}\hspace*{0.03\textwidth}
\includegraphics[width=0.35\textwidth]{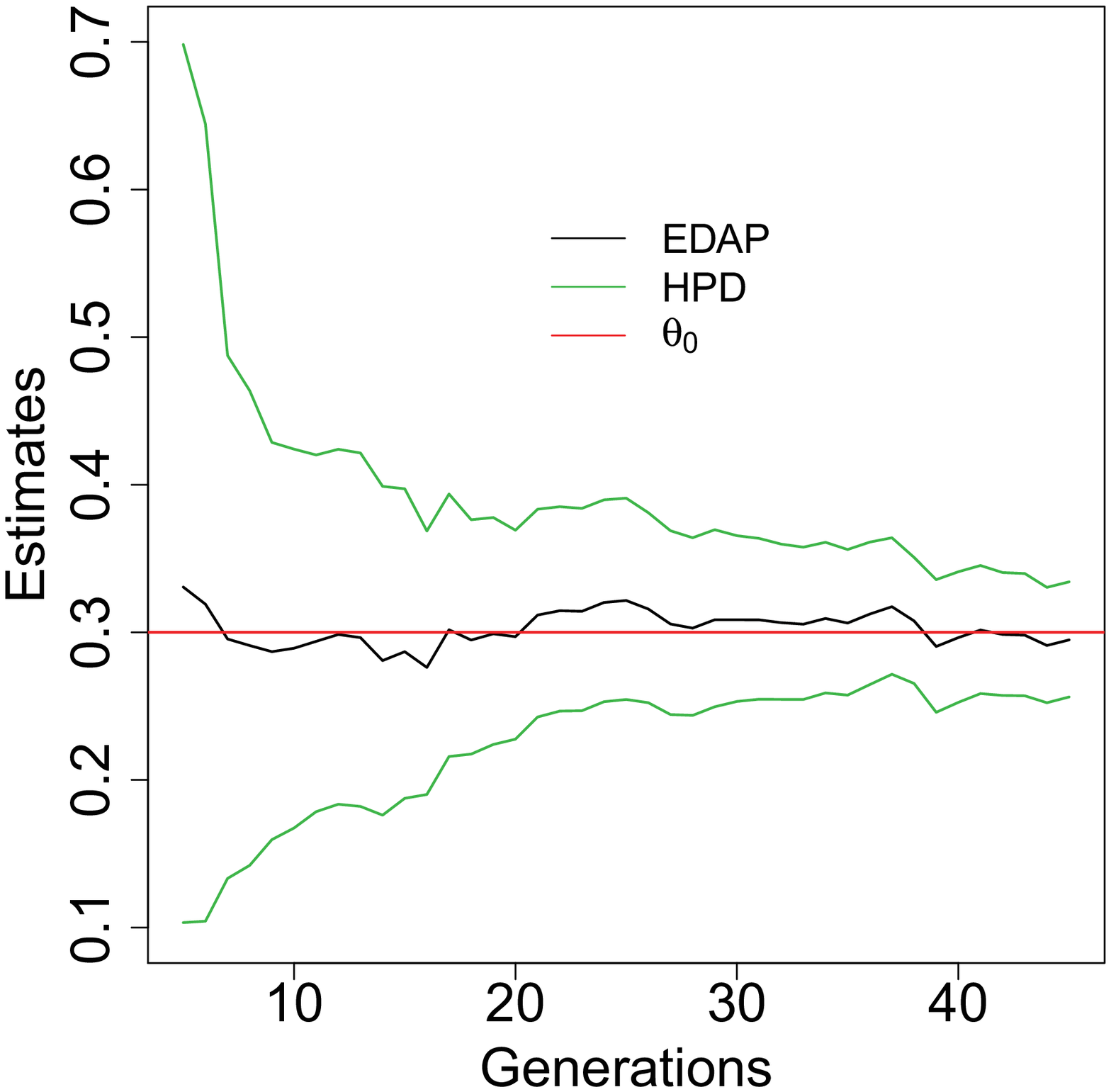}
\caption{Evolution of EDAP estimates (black line) and of the HPD intervals (green lines). Left: HD. Right: NED.  Horizontal red lines represent the true value of $\theta_0$.}\label{Dposterior:fig:consistency-EDAP}
\end{figure}

\begin{figure}[H]
\centering\includegraphics[width=0.35\textwidth]{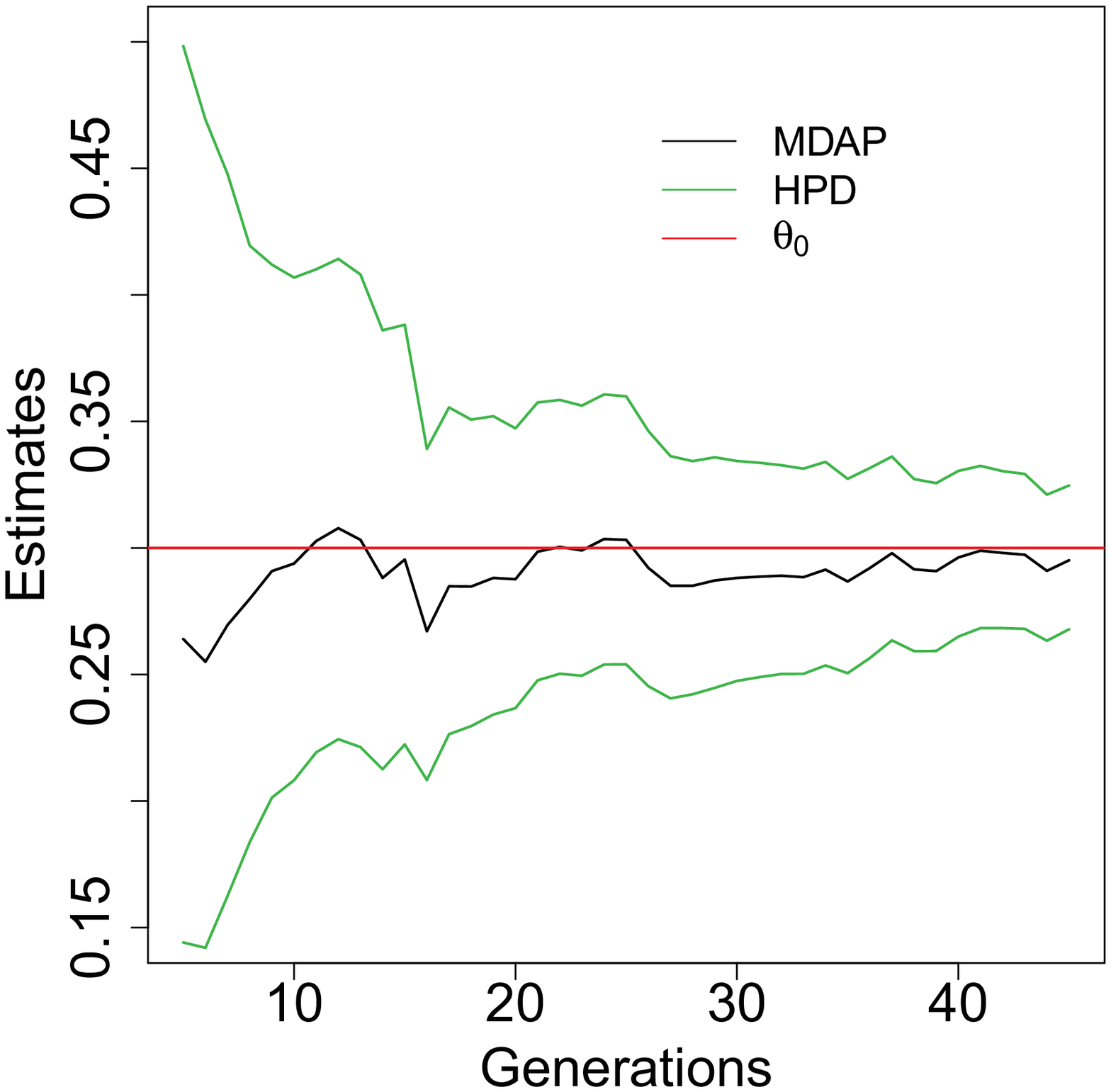}\hspace*{0.03\textwidth}
\includegraphics[width=0.35\textwidth]{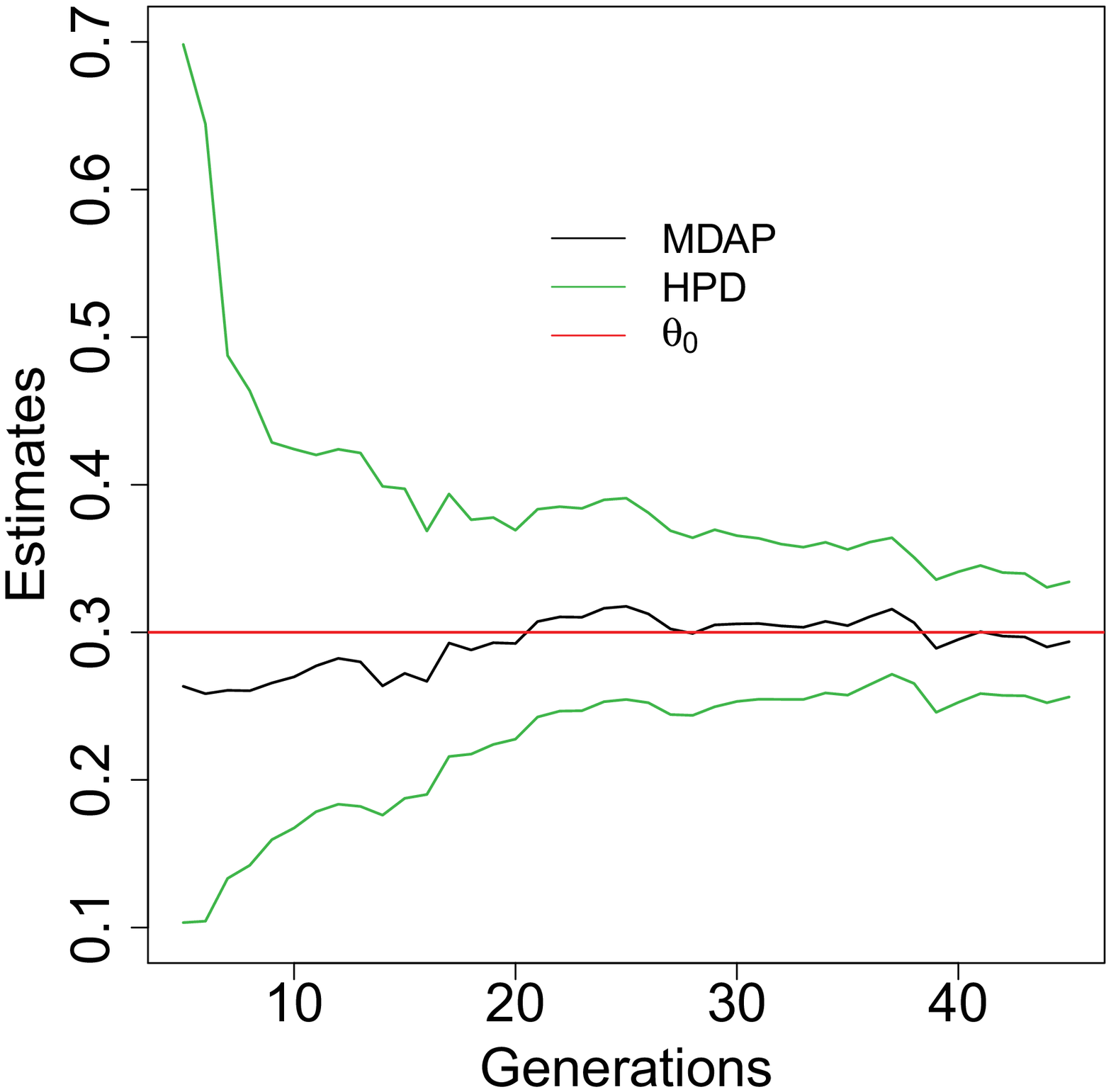}
\caption{Evolution of MDAP estimates (black line) and of the HPD intervals (green lines). Left: HD. Right: NED.  Horizontal red lines represent the true value of $\theta_0$.}\label{Dposterior:fig:consistency-MDAP}
\end{figure}

Taking into account that the offspring mean and variance can be written as continuous functions of the offspring parameter, we have used the EDAP and MDAP estimates of the offspring parameter to obtain estimates of them. The evolution of $m(\theta_n^{*D})$, and $\sigma^2(\theta_n^{*D})$, for $n=5,\ldots,45$, for $D\in\{HD,NED,KL\}$, are shown in Figure \ref{Dposterior:fig:evol-estim-m-sigma}, whereas $m(\theta_n^{+D})$, and $\sigma^2(\theta_n^{+D})$, for $n=5,\ldots,45$, for $D\in\{HD,NED,KL\}$, are shown in Figure \ref{Dposterior:fig:evol-estim-m-sigma-mdap}. Moreover, note that due to the continuity, these estimators prove to be strongly consistent estimators for the corresponding parameter.

\begin{figure}[H]
\centering\includegraphics[width=0.35\textwidth]{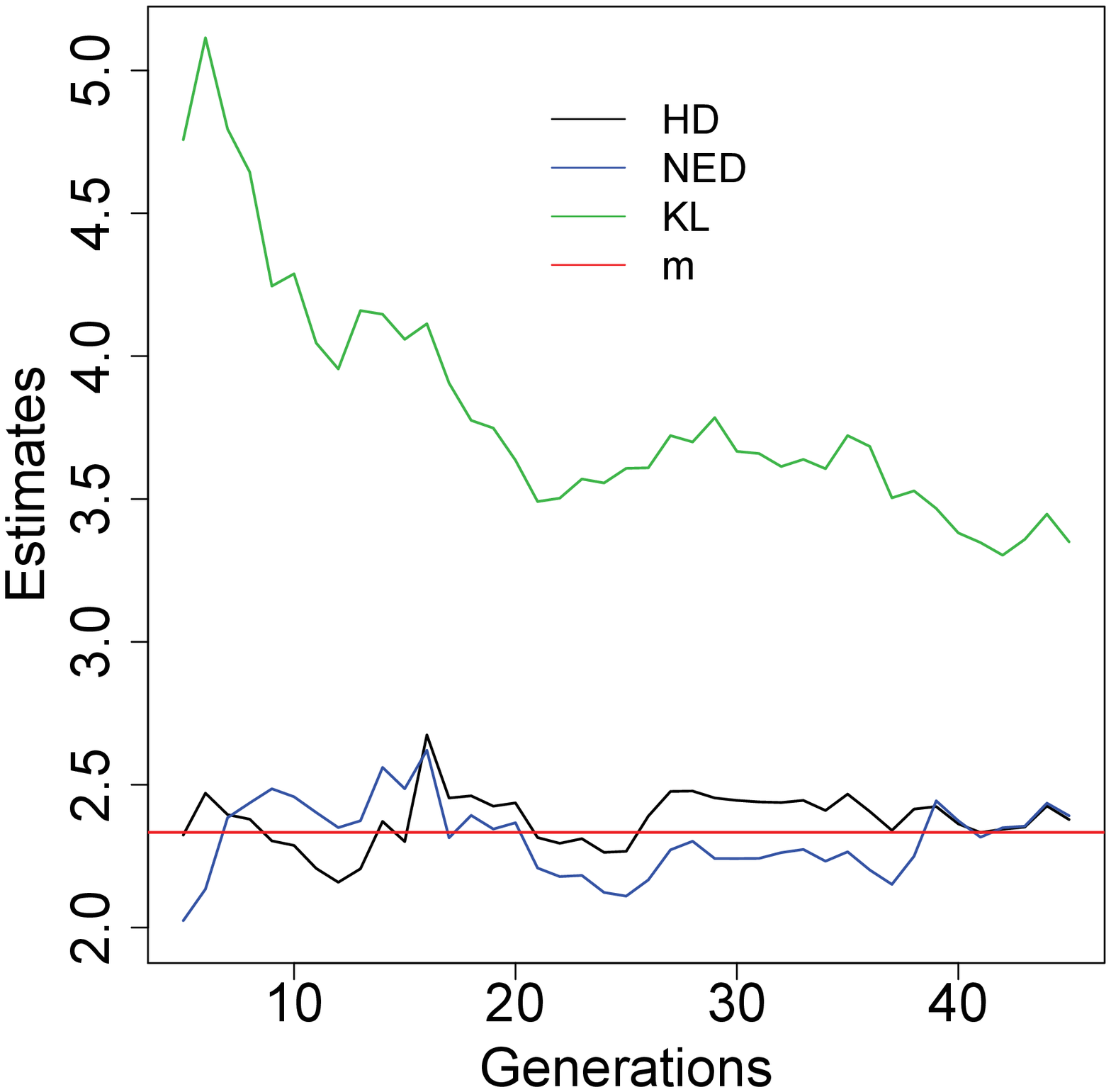}\hspace*{0.03\textwidth}
\includegraphics[width=0.35\textwidth]{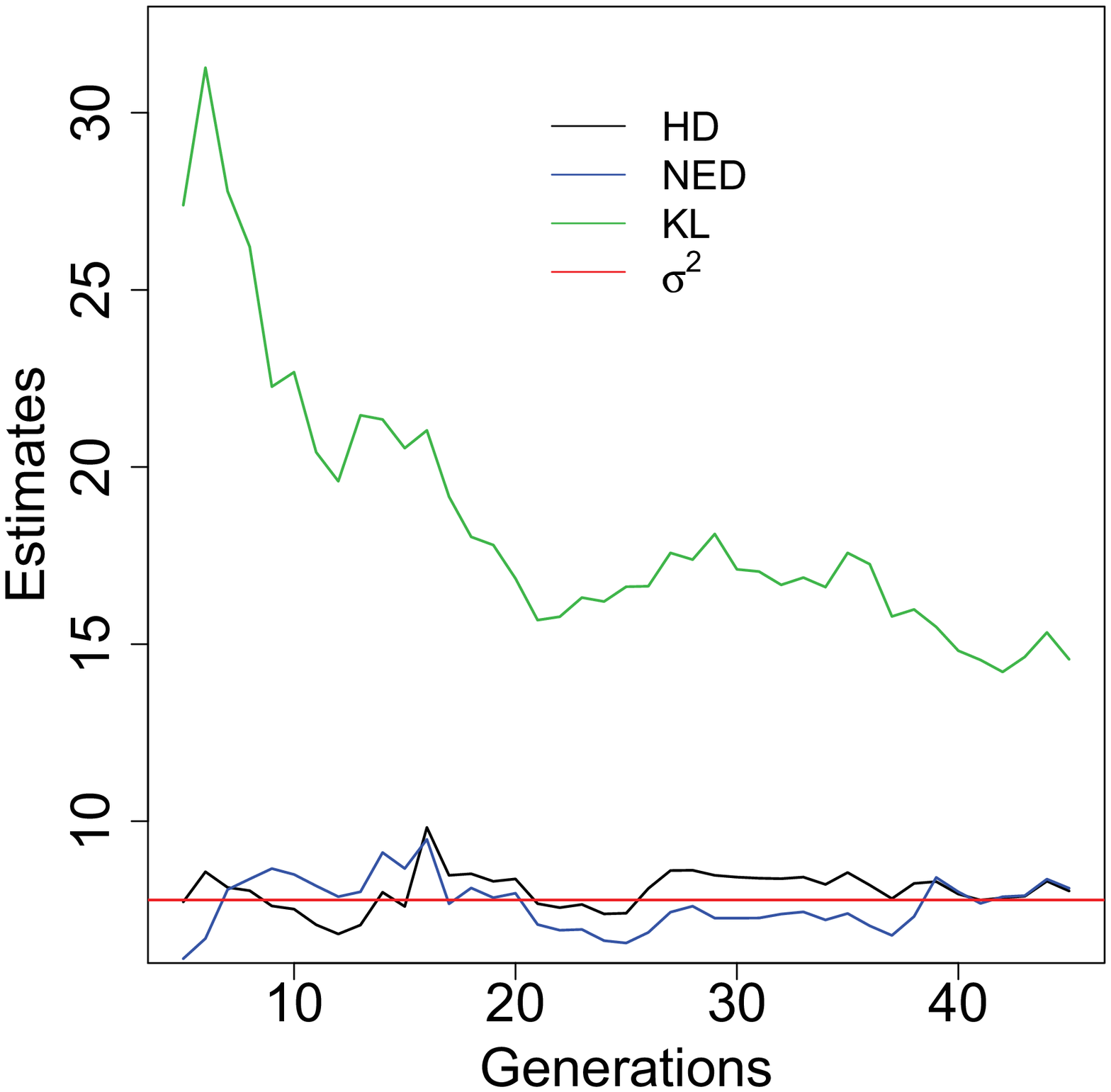}
\caption{Left: evolution of the estimates of $m$ based on the EDAP estimates  of $\theta_0$  for the different disparities. Right: evolution of the estimates of $\sigma^2$ based on the EDAP estimates  of $\theta_0$  for the different disparities. Horizontal red lines represent the true value of the corresponding parameter.}\label{Dposterior:fig:evol-estim-m-sigma}
\end{figure}

\begin{figure}[H]
\centering\includegraphics[width=0.35\textwidth]{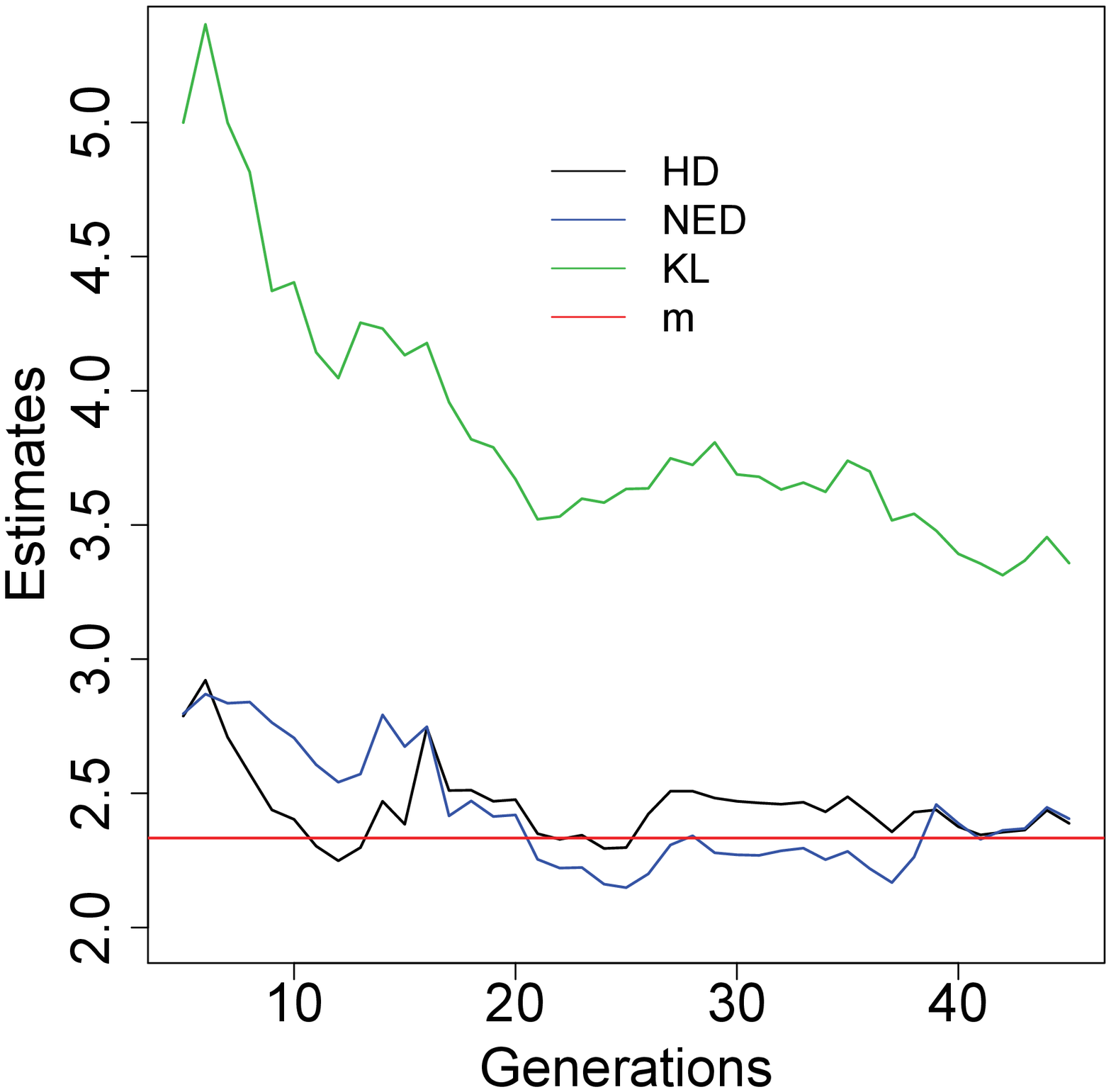}\hspace*{0.03\textwidth}
\includegraphics[width=0.35\textwidth]{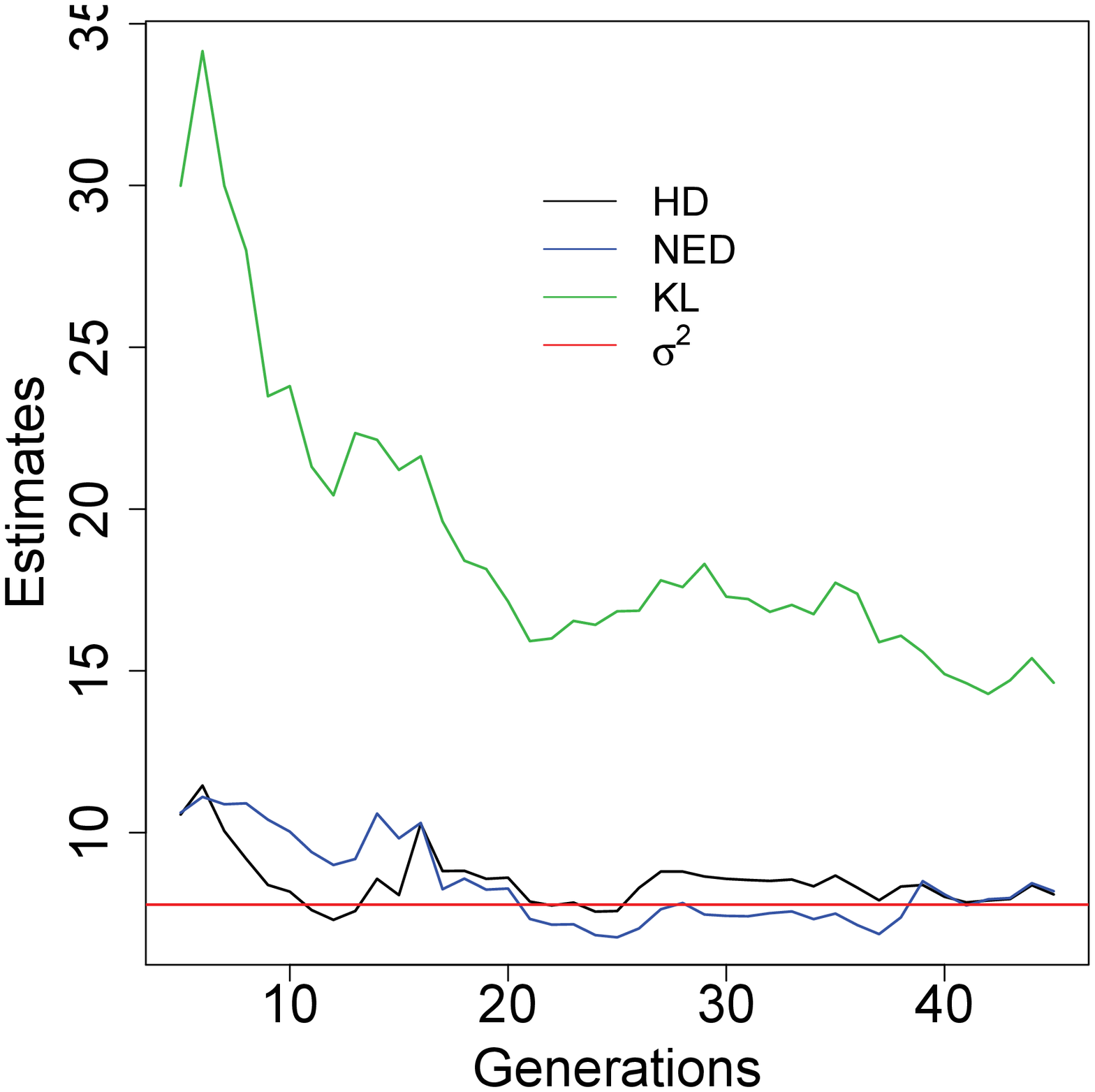}
\caption{Left: evolution of the estimates of $m$ based on the MDAP estimates  of $\theta_0$  for the different disparities. Right: evolution of the estimates of $\sigma^2$ based on the MDAP estimates  of $\theta_0$  for the different disparities. Horizontal red lines represent the true value of the corresponding parameter.}\label{Dposterior:fig:evol-estim-m-sigma-mdap}
\end{figure}

Finally, a sensitivity analysis has been performed in order to determine the influence of the choice of the prior distribution in the $D$-posterior one and the corresponding point estimators. Since the parameter to be estimated lies in the interval $[0,1]$, we use beta distributions as possible priors. More specifically, to apply the method we make use of beta distributions with parameters $(\rho,\beta)$ belonging to the grid given by the Cartesian product of the set $\{n+0.1*k:\ n=0,1,\ldots,4;\ k=1,2,\ldots10\}$ with itself as prior distributions. In Table \ref{Dposterior:table:sens} some obtained results are summarized, especially those corresponding to prior distributions which are very concentrated at extreme values far away from the  parameter of interest. Despite this fact, one can observe the slight influence of the choice of the prior distribution on the EDAP and MDAP estimates. In addition, although the results are not shown here, it is interesting to mention that the effect of the prior distribution on the estimates decreases as the generations go up (see Supplementary material for further details).

\begin{table}[!htbp] \centering
\begin{tabular}{|cccccccc|}
\hline
 $\rho$ & $\beta$ & Prior mean & Prior variance & $\theta_{45}^{*HD}$ & $\theta_{45}^{*NED}$ & $\theta_{45}^{+HD}$ & $\theta_{45}^{+NED}$ \\
\hline
$0.1$ & $5$ & $0.020$ & $0.003$ & $0.294$ & $0.292$ & $0.294$ & $0.291$ \\
$0.1$ & $1$ & $0.091$ & $0.039$ & $0.296$ & $0.294$ & $0.295$ & $0.293$ \\
$2$ & $5$ & $0.286$ & $0.026$ & $0.296$ & $0.294$ & $0.295$ & $0.293$ \\
$1$ & $2$ & $0.333$ & $0.056$ & $0.296$ & $0.295$ & $0.295$ & $0.294$ \\
$0.1$ & $0.1$ & $0.500$ & $0.208$ & $0.296$ & $0.295$ & $0.295$ & $0.293$ \\
$2.5$ & $2.5$ & $0.500$ & $0.042$ & $0.297$ & $0.296$ & $0.296$ & $0.295$ \\
$2$ & $1$ & $0.667$ & $0.056$ & $0.297$ & $0.297$ & $0.296$ & $0.295$ \\
$5$ & $2$ & $0.714$ & $0.026$ & $0.299$ & $0.300$ & $0.298$ & $0.299$ \\
$1$ & $0.1$ & $0.909$ & $0.039$ & $0.296$ & $0.296$ & $0.296$ & $0.295$ \\
$5$ & $0.1$ & $0.980$ & $0.003$ & $0.299$ & $0.301$ & $0.298$ & $0.300$ \\
\hline
\end{tabular}
  \caption{EDAP and MDAP estimates for the HD and NED considering different beta distributions as prior distributions.}
  \label{Dposterior:table:sens}
\end{table}

\begin{remark}
Regarding computational purposes, to approximate the integrals involved in \eqref{Dposterior:eq:D-posterior-tree} and \eqref{Dposterior:def:EDAP-estimator}, we have made use of the function \texttt{cotes()} of the library \texttt{pracma} (see \cite{pracma}).
\end{remark}

\section{Concluding remark}\label{Dposterior:sec:conclusion}

For controlled branching processes, this paper addresses the issue of obtaining robust estimators of the offspring distribution in a Bayesian context making use of disparity-based methods. It has been assumed that  the offspring distribution belongs to a general parametric family and  the inference on the main parameter of the reproduction law is based on the sample given by the entire family tree.

Firstly, we have defined the $D$-posterior density, a density function which is obtained by replacing the log-likelihood in Bayes rule with a conveniently scaled disparity measure. The expectation and mode of the $D$-posterior density, denoted as EDAP and MDAP estimators, respectively,   are  proposed as Bayes estimators for the offspring parameter, emulating the point estimators under the squared error loss function or under $0-1$ loss function, respectively, for the posterior density. As initial step for the analysis of the asymptotic and robustness properties of these estimators, their existence and measurability are studied. Moreover, sufficient conditions for their strong consistency and asymptotic normality, once suitably normalized, are provided.

Special attention has been paid in the research of the robustness properties of the proposed EDAP and MDAP estimators. To that end we have introduced the so-called EDAP and MDAP functions for dependent tree-structured data. These functions present the novelty and added difficulty of their random nature as a consequence of that EDAP and MDAP functions depend on the total number of progenitors up to certain generation. This leads to examine their continuity and asymptotic behaviour carefully. These features play an important role in the study of the  robustness of EDAP and MDAP estimators, where the classical measures have been examined. However, the randomness in the EDAP and MDAP functions has made necessary the introduction of some new measures for the robustness of these estimators, which have shown the robust behaviour of the EDAP and MDAP estimators when the sample size grows.

Although results are provided for a class of disparity measures, we have focused our attention on the Hellinger distance and the negative exponential disparity in the examples presented. In the first one, we applied the described methodology to real data from oligodendrocyte cell populations and estimated the main parameters of the process suggested to model such populations. The results in this example illustrate the asymptotic properties of the EDAP and MDAP estimators. The second example considers simulated data  with the purpose of showing and analysing the robustness properties of the EDAP and MDAP estimators established in this paper. Thus, both examples illustrate the good properties of EDAP and MDAP estimators, showing that they are efficient in a free-contamination context and robust to model misspecification and presence of aberrant ouliers. We finally remark that since our ultimate goal is to provide estimators of the offspring parameter, in both examples we have made use of numerical methods to approximate it instead of using a MCMC methodology (as used in \cite{Hooker-Vidyshankar-2014}), which has resulted in a reduction of the computational time to calculate the EDAP and MDAP estimators.

\section*{Acknowledgements}

This research was partially developed while Carmen Minuesa was a visiting PhD student at the Department of Statistics, George Mason University, and she is grateful for the hospitality and collaboration. The authors would like to thank Jacinto Mart\'in (University of Extremadura) for his valuable comments and suggestions. This research has been supported by the Ministerio de Educaci\'on, Cultura y Deporte (grant FPU13/03213), the Junta de Extremadura (grant GR15105 and IB16099), the Ministerio de Econom\'ia y Competitividad of Spain (grant MTM2015-70522-P) and the Fondo Europeo de Desarrollo Regional.

\section{Supplementary material}

\subsection{Proof of properties of the EDAP and MDAP functions}\label{Dposterior:ap:method}

\begin{Prf}[Proposition \ref{Dposterior:prop:existence-EDAP}]
For \ref{Dposterior:prop:existence-EDAP-functional}, observe that since $G(\cdot)$ is strictly convex, the disparity $D$ is non-negative, and as a consequence, for each $\omega\in\Omega$, $|\overline{T}_n(q)(\omega)|<\infty$.

For the continuity of EDAP function, due to the facts that $e^{-\Delta_{n-1}(\omega)D(q_j,\theta)}\leq 1$, for each $\omega\in\Omega$, $j\in\N$, $\theta\in\Theta$, and $D(q_j,\theta)\to D(q,\theta)$, as $j\to\infty$, one has that $\int_\Theta e^{-\Delta_{n-1}(\omega)D(q_j,\theta)}\pi(\theta)d\theta\to\int_\Theta e^{-\Delta_{n-1}(\omega)D(q,\theta)}\pi(\theta)d\theta$, as $j\to\infty$, for each $n\in\N$, and $\omega\in\Omega$, and consequently, making use of a generalized version of the dominated convergence theorem  (see \cite{Royden}, p.92), one has that, as $j\to\infty$,
\begin{equation*}
\overline{T}_n(q_j)(\omega)=\frac{\int_\Theta\theta e^{-\Delta_{n-1}(\omega)D(q_j,\theta)}\pi(\theta)d\theta}{\int_\Theta e^{-\Delta_{n-1}(\omega)D(q_j,\theta)}\pi(\theta)d\theta}\to\frac{\int_\Theta\theta e^{-\Delta_{n-1}(\omega)D(q,\theta)}\pi(\theta)d\theta}{\int_\Theta e^{-\Delta_{n-1}(\omega)D(q,\theta)}\pi(\theta)d\theta}=\overline{T}_n(q)(\omega).
\end{equation*}
Now, by the continuity of $\overline{T}_n(\cdot)(\omega)$ for each $\omega\in\Omega$ and the measurability of $\overline{T}_n(q)$ for each $q\in\Gamma$, the measurability of $\overline{T}_n$ follows from Theorem 2 in \cite{Gowrisankaran-1972}.
\end{Prf}

\begin{Prf}[Proposition \ref{Dposterior:prop:existence-MDAP}]
\ref{Dposterior:prop:existence-MDAP-functional} For each $\omega\in\Omega$, the finiteness and existence of $\widetilde{T}_n(q)(\omega)$ is immediate from the definition of the family $\Gamma_{n,\omega}^+$ and the continuity of the function $g_n(q,\omega,\cdot)$. For the measurability  of $\widetilde{T}_n(q)$ we apply Theorem 4.5 in \cite{Debreu-1967} bearing in mind that $g_n(q,\cdot,\cdot)$ is measurable, $g_n(q,\omega,\cdot)$ is continuous for each $\omega\in\Omega$, and since $\widetilde{T}_n(q)(\omega)$ is unique, the function $S:\omega\in\Omega\mapsto S(\omega)=\arg\min_{\theta\in\Theta} g_n(q,\omega,\theta)$ is well defined.

For \ref{Dposterior:prop:existence-continuity-MDAP-functional}, let denote $t_\omega=\widetilde{T}_n(q)(\omega)$ in order to ease the notation (it exists by \ref{Dposterior:prop:existence-MDAP-functional} and is unique) and consider that all the limits are taken as $j\to\infty$ unless specified otherwise. First of all, from the continuity of $D(q,\cdot)$ and the fact that
\begin{equation}\label{Dposterior:eq:lim-sup-g-n}
\sup_{\theta\in\Theta} |D(q_j,\theta)-D(q,\theta)|\to 0,
\end{equation}
it follows that if $\{\theta_n\}_{n\in\N}$ is a sequence in $\Theta$ converging to $\theta^*\in\Theta$, given $\epsilon>0$, one can find $J=J(\epsilon)\in\N$ and $n_0=n_0(\epsilon)\in\N$ such that $|D(q_j,\theta_n)-D(q_j,\theta^*)|<\epsilon$, for each $j\geq J$ and $n\geq n_0$.

Now, fix $\omega\in\Omega$; from \eqref{Dposterior:eq:lim-sup-g-n}, it follows that $\sup_{\theta\in\Theta} |g_n(q_j,\omega,\theta)-g_n(q,\omega,\theta)|\to 0$. From this latter, it is deduced that $q_j\in\Gamma_{n,\omega}^+$ eventually. Thus, using the same arguments as in \ref{Dposterior:prop:existence-MDAP-functional}, $\widetilde{T}_n(q_j)(\omega)$ eventually exists and $\widetilde{T}_n(q_j)(\omega)\in C_{n,\omega}^+$. Let us denote $\widetilde{T}_n(q_j)(\omega)$ by $t_{j,\omega}$ to ease the notation; next we show that $t_{j,\omega}\to t_\omega$, as $j\to\infty$.

From \eqref{Dposterior:eq:lim-sup-g-n}, the convergence of $g_n(q_j,\omega,t_{j,\omega})\to g_n(q,\omega,t_\omega)$ and $|g_n(q_j,\omega,t_{j,\omega})-g_n(q,\omega,t_{j,\omega})|\to 0$,  are deduced, so $g_n(q,\omega,t_{j,\omega})\to g_n(q,\omega,t_\omega)$. If the sequence $\{t_{j,\omega}\}_{j\in\N_0}$ did not converge to $t_\omega$, then there would exist a subsequence $\{t_{j_l,\omega}\}_{l\in\N}\subseteq\{t_{j,\omega}\}_{j\in\N}$ such that $t_{j_l,\omega}\to t_\omega^*\neq t_\omega$, as $l\to\infty$. Since $g_n(q,\omega,\cdot)$ is continuous, $g_n(q,\omega,t_{j_l\omega})\to g_n(q,\omega,t_\omega^*)$, as $l\to\infty$. Due to all of the above, one would have $g_n(q,\omega,t_\omega)=g_n(q,\omega,t_\omega^*)$, which would contradict the uniqueness of $t_\omega$.

If $\widetilde{T}_n(q_j)(\omega)$ exists for each $\omega\in\Omega$, the measurability of $\widetilde{T}_n(q_j)$ is immediate with the same reasoning as in \ref{Dposterior:prop:existence-MDAP-functional}.

Finally, one obtains the measurability of $\widetilde{T}_n$ from Theorem 2 in \cite{Gowrisankaran-1972} and the facts that for each $\omega\in\Omega$, $\widetilde{T}_n(\cdot)(\omega)$ is continuous and for each $q\in\Gamma_{n,\omega}^+$, $\widetilde{T}_n(q)$ is measurable.
\end{Prf}

\begin{Prf}[Theorem \ref{Dposterior:thm:aprox-MDE}]
Throughout this proof, all statements are made on the set $\{Z_n\to\infty\}$ and for $n$ large enough. Let fix $\omega\in\{Z_n\to\infty\}$ then $\Delta_n(\omega)\to\infty$, and in the remainder of this section, we consider that all the random variables are evaluated at $\omega$, although we do not write it explicitly in order to ease the notation.

To facilitate the proof of Theorem \ref{Dposterior:thm:aprox-MDE}, we will make use of the following lemma.

\begin{lem}\label{Dposterior:lem:posterior-t}
Let $q\in\widetilde{\Gamma}$, $\pi_D^n(\theta|q)$ be the $D$-posterior density of $\theta$ at $q$, $\overline{\pi}_D^{n}(t|q)$ be the $D$-posterior density function of $t=\Delta_{n-1}^{1/2}(\theta-T(q))$ at $q$ and $R_{n}=\{\Delta_{n-1}^{1/2}(\theta-T(q)):\theta\in\Theta\}$. Under conditions of Theorem \ref{Dposterior:thm:aprox-MDE}, for each $t\in\R$
\begin{align}
\overline{\pi}_D^{n}(t|q)&=\left(\frac{I^D(T(q))}{2\pi}\right)^{1/2}e^{-\frac{t^2I^D(\theta_n'(t))}{2}}\bigg(1+\frac{tb_1}{\Delta_{n-1}^{1/2}}+\frac{t^2b_2}{2\Delta_{n-1}}+\frac{t^3 \pi'''(\theta_n^*(t))}{6\Delta_{n-1}^{3/2}\pi(T(q))}\bigg)\nonumber\\
&\phantom{=}\cdot(1+C_n)I_{R_n}(t),\label{Dposterior:eq:expansion-Dposterior-t}
\end{align}
a.s. on $\{Z_n\to\infty\}$, where $I^D(\theta)=\ddot{D}\left(q,\theta\right)$, $b_1=\frac{\pi'\left(T(q)\right)}{\pi\left(T(q)\right)}$, $b_2=\frac{\pi''\left(T(q)\right)}{\pi\left(T(q)\right)}$, $\pi'(\cdot)$ and $\pi''(\cdot)$ denote the first and the second derivative of the function $\pi(\cdot)$, $\theta_n'(t)$  and $\theta_n^*(t)$ are both points between $T(q)$ and $T(q)+\frac{t}{\Delta_{n-1}^{1/2}}$, for each $n\in\N$, and $t\in\R$, and  $\{C_n\}_{n\in\N}$ is a sequence of real valued random variables converging to 0 a.s.
\end{lem}

\begin{proof}
The idea of the proof is similar to the proof of Theorem 1 in \cite{Hooker-Vidyshankar-2014} and uses arguments in \cite{Ghosh-1994} (pp. 46-47). Notice that, for each $t\in R_{n}$,
\begin{equation*}
\overline{\pi}_D^{n}(t|q)=\frac{e^{-\Delta_{n-1} D\left(q,T(q)+\frac{t}{\Delta_{n-1}^{1/2}}\right)+\Delta_{n-1} D\left(q,T(q)\right)}\pi\left(T(q)+\frac{t}{\Delta_{n-1}^{1/2}}\right)}{\bigint_{R_n} e^{-\Delta_{n-1} D\left(q,T(q)+\frac{t}{\Delta_{n-1}^{1/2}}\right)+\Delta_{n-1} D\left(q,T(q)\right)}\pi\left(T(q)+\frac{t}{\Delta_{n-1}^{1/2}}\right)dt}.\label{Dposterior:eq:posterior-t}
\end{equation*}

Moreover, since $T(q)\in int(\Theta)$, one has that $\cup_{n=1}^\infty R_n =\R$. This is immediate from the fact that there exists $\eta>0$ such that $(T(q)-\eta,T(q)+\eta)\subseteq\Theta$, consequently $(-\Delta_{n-1}^{1/2}\eta,\Delta_{n-1}^{1/2}\eta)\subseteq R_n$, and taking limit as $n\to\infty$, one obtains $\cup_{n=1}^\infty R_n =\R$.

On the one hand, using a second-order Taylor series expansion of the prior density one has that, for each $t\in R_{n}$,
\begin{equation}\label{Dposterior:eq:expansionn-prior}
\pi\left(T(q)+\frac{t}{\Delta_{n-1}^{1/2}}\right)=\pi\left(T(q)\right)+\frac{t\pi'\left(T(q)\right)}{\Delta_{n-1}^{1/2}}+\frac{t^2\pi''\left(T(q)\right)}{2\Delta_{n-1}}+\frac{t^3\pi'''(\theta_n^*(t))}{6\Delta_{n-1}^{3/2}},
\end{equation}
where $\theta_n^*(t)$ is a point between $T(q)$ and $T(q)+\frac{t}{\Delta_{n-1}^{1/2}}$. On the other hand, using a first-order Taylor series expansion of the function $D(q,\cdot)$, for each $t\in R_{n}$,
\begin{equation}\label{Dposterior:eq:expansionn-disparity}
\Delta_{n-1}D\left(q,T(q)+\frac{t}{\Delta_{n-1}^{1/2}}\right)-\Delta_{n-1}D\left(q,T(q)\right)=\frac{t^2I^D(\theta_n'(t))}{2},
\end{equation}
where $\theta_n'(t)$ is a point between $T(q)$ and $T(q)+\frac{t}{\Delta_{n-1}^{1/2}}$.

Let denote $h_n(t)=\pi\bigg(T(q)+\frac{t}{\Delta_{n-1}^{1/2}}\bigg)e^{-\Delta_{n-1}D\left(q,T(q)+\frac{t}{\Delta_{n-1}^{1/2}}\right)+\Delta_{n-1}D\left(q,T(q)\right)}I_{R_n}(t)$, for each $n\in\N$, and $t\in\R$. Combining \eqref{Dposterior:eq:expansionn-prior} and \eqref{Dposterior:eq:expansionn-disparity}, one has
\begin{equation}
h_n(t)=\pi(T(q))\bigg(1+\frac{tb_1}{\Delta_{n-1}^{1/2}}+\frac{t^2b_2}{2\Delta_{n-1}}+\frac{t^3 \pi'''(\theta_n^*(t))}{6\Delta_{n-1}^{3/2}\pi(T(q))}\bigg)e^{-\frac{t^2I^D(\theta_n'(t))}{2}}I_{R_n}(t).\label{Dposterior:eq:integrando}
\end{equation}

Let define, for each $n\in\N$, the integrals $J_n=\int h_n(t)dt$, and
\begin{eqnarray*}
I_{n}&=&\int_{R_n}\bigg(1+\frac{tb_1}{\Delta_{n-1}^{1/2}}+\frac{t^2b_2}{2\Delta_{n-1}}+\frac{t^3 \pi'''(\theta_n^*(t))}{6\Delta_{n-1}^{3/2}\pi(T(q))}\bigg)e^{-\frac{t^2I^D(T(q))}{2}}dt\\
&=&\left(\frac{2\pi}{I^D(T(q))}\right)^{1/2}\left(1+\frac{b_2}{2\Delta_{n-1}I^D(T(q))}+\frac{C}{6\Delta_{n-1}^{3/2}\pi(T(q))}\right)+o(1),
\end{eqnarray*}
where to obtain the last equality we have applied the dominated convergence theorem and that $\left(\frac{I^D(T(q))}{2\pi}\right)^{1/2}\cdot\int t^3 \pi'''(\theta_n^*(t)) e^{-\frac{t^2I^D(T(q))}{2}}dt\leq C$, for some constant $C>0$, due to the boundedness of the function $\pi'''(\cdot)$. If one proves that $J_n=\pi(T(q))I_{n}+o(1)$, then, from all the above, one deduces \eqref{Dposterior:eq:expansion-Dposterior-t}.

To prove $J_n=\pi(T(q))I_{n}+o(1)$, let fix $0<\epsilon<I^D(T(q))$. Since $I^D(\cdot)$ is continuous at $T(q)$, there exists $\delta=\delta(\epsilon)>0$ such that if $|\theta-T(q)|\leq\delta$, then $|I^D(\theta)-I^D(T(q))|\leq \epsilon$. Let define, for each $n\in\N$, the set $B_n=\{t\in\R:|t|\leq \delta\Delta_{n-1}^{1/2}\}$, and note that $J_n=J_{1n}+J_{2n}$, where $J_{1n}=\int_{B_n} h_n(t)dt$, and $J_{2n}=\int_{B_n^c} h_n(t)dt$. As a consequence, it is enough to prove that, for each $n\in\N$,  $J_{1n}=\pi(T(q))I_{n}+a_{n}$, and $\{a_{n}\}_{n\in\N}$ and $\{J_{2n}\}_{n\in\N}$ are both sequences of real numbers converging to 0.

For the former, observe that $a_{n}=\int x_n(t)dt$, where
\begin{eqnarray}
x_n(t)&=&\bigg(1+\frac{tb_1}{\Delta_{n-1}^{1/2}}+\frac{t^2b_2}{2\Delta_{n-1}}+\frac{t^3 \pi'''(\theta_n^*(t))}{6\Delta_{n-1}^{3/2}\pi(T(q))}\bigg)\left(e^{-\frac{t^2I^D(\theta_n'(t))}{2}}-e^{-\frac{t^2I^D(T(q))}{2}}\right)\nonumber\\
&\phantom{=}&\cdot\ \pi(T(q))I_{B_n\cap R_n}(t)-\pi(T(q))\bigg(1+\frac{tb_1}{\Delta_{n-1}^{1/2}}+\frac{t^2b_2}{2\Delta_{n-1}}+\frac{t^3 \pi'''(\theta_n^*(t))}{6\Delta_{n-1}^{3/2}\pi(T(q))}\bigg)\nonumber\\
&\phantom{=}&\cdot\ e^{-\frac{t^2I^D(T(q))}{2}}I_{B_n^c\cap R_n}(t),\label{Dposterior:eq:integral-xnt}
\end{eqnarray}
consequently, to prove $a_{n}\to 0$ it is enough to prove that the integral of both terms in \eqref{Dposterior:eq:integral-xnt} converge to 0. The convergence of the first one is obtained by applying the dominated convergence theorem bearing in mind that, for each $t\in B_n\cap R_n$, $0<I^D(T(q))-\epsilon<I^D(\theta_n'(t))<I^D(T(q))+\epsilon$ (due to the fact that $\theta_n'(t)$ is a point between $T(q)$ and $T(q)+\frac{t}{\Delta_{n-1}^{1/2}}$). For the convergence of the integral of the second term, one has
\begin{align*}
\int \bigg(1+\frac{tb_1}{\Delta_{n-1}^{1/2}}&+\frac{t^2b_2}{2\Delta_{n-1}}+\frac{t^3 \pi'''(\theta_n^*(t))}{6\Delta_{n-1}^{3/2}\pi(T(q))}\bigg)e^{-\frac{t^2I^D(T(q))}{2}}I_{B_n^c\cap R_n}(t)dt\leq \\
&\leq \bigg(\frac{2\pi}{I^D(T(q))}\bigg)^{1/2}\cdot\Bigg( P[Z\in B_n^c\cap R_n]+\frac{|b_1|}{\Delta_{n-1}^{1/2}}E[|Z|I_{ B_n^c\cap R_n}] \\
&\phantom{\leq}+\frac{|b_2|}{2\Delta_{n-1}}E[Z^2I_{ B_n^c\cap R_n}]++\frac{M}{6\Delta_{n-1}^{3/2}\pi(T(q))}E[|Z|^3 I_{ B_n^c\cap R_n}]\Bigg)\to 0,
\end{align*}
where $M>0$ is an upper bound of $\pi'''(\cdot)$, and $Z$ is a random variable following a normal distribution with mean equal to 0 and variance $I^D(T(q))^{-1}$.

The convergence $J_{2n}\to 0$ follows from \eqref{Dposterior:cond:sep}. To that end, note that for the fixed $\delta>0$, there exists $\rho>0$ such that
$$\inf_{t\in B_n^c\cap R_n} D\left(q,T(q)+\frac{t}{\Delta_{n-1}^{1/2}}\right)-D\left(q,T(q)\right)= \inf_{\theta\in\Theta:|\theta-T(q)|>\delta} D\left(q,\theta\right)-D\left(q,T(q)\right)>\rho,$$
and as a result,
\begin{align*}
J_{2n}&\leq\int_{B_n^c\cap R_n} \pi\left(T(q)+\frac{t}{\Delta_{n-1}^{1/2}}\right)e^{-\rho\Delta_{n-1}}dt\\
&=\Delta_{n-1}^{1/2}e^{-\rho\Delta_{n-1}}\int_{\{\theta\in\Theta:|\theta-T(q)|>\delta\}}\pi(\theta)d\theta\to 0.
\end{align*}
\end{proof}

\vspace*{2ex}

Now, we use the notation and the approximation for the $D$-posterior density function of $t=\Delta_{n-1}^{1/2}(\theta-T(q))$ at $q$ given in the previous lemma to prove Theorem \ref{Dposterior:thm:aprox-MDE}. Throughout this proof, all the limits are taken as $n\to\infty$ unless specified otherwise.

\ref{Dposterior:thm:aprox-MDE-i-EDAP} Observe that $\Delta_{n-1}^{1/2}(\overline{T}_n(q)-T(q))=\int_{R_n} t \overline{\pi}_D^n(t|q)dt$. Hence, to finish the proof it is enough to prove that both integrals $\int_{B_n\cap R_n} t \overline{\pi}_D^n(t|q)dt$ and $\int_{B_n^c\cap R_n} t \overline{\pi}_D^n(t|q)dt$ converge to 0. For the former, one has $\int_{B_n\cap R_n} t \overline{\pi}_D^n(t|q)dt=(1+c_n)\sum_{i=0}^3 I_{in}$,
where $I_{in}=\int_{B_n\cap R_n} t f_{in}(t)dt$, $f_{in}(t)=\bigg(\frac{I^D(T(q))}{2\pi}\bigg)^{1/2}c_{in}(t)e^{-\frac{t^2I^D(\theta_n'(t))}{2}}$, for $i=0,1,2,3$, and $c_{0n}(t)=1$, $c_{1n}(t)=\frac{tb_1}{\Delta_{n-1}^{1/2}}$, $c_{2n}(t)=\frac{t^2b_2}{2\Delta_{n-1}}$, and $c_{3n}(t)=\frac{t^3\pi'''(\theta_n^*(t))}{6\Delta_{n-1}^{3/2}\pi(T(q))}$. If $h(t)=\big(\frac{I^D(T(q))}{2\pi}\big)^{1/2}e^{-\frac{t^2I^D(T(q))}{2}}$, and we prove that $I_{in}-\int_{R_n} t c_{in}(t)h(t)dt\to 0$, for $i=0,1,2,3$, then, taking into account that $\int_{R_n} t c_{in}(t)h(t)dt\to 0$, for $i=0,1,2,3$, we obtain $I_{in}\to 0$, for $i=0,1,2,3$.

To prove that $I_{in}-\int_{R_n} t c_{in}(t)h(t)dt\to 0$, for $i=0,1,2,3$, observe that, for each $i=0,1,2,3$, $t f_{in}(t)I_{B_n\cap R_n}(t)=t c_{in}(t)h(t)I_{R_n}(t)+x_n^i(t)$, with
\begin{equation}\label{Dposterior:eq:function-x-n-i-t}
x_n^i(t)=(t f_{in}(t)-t c_{in}(t)h(t))I_{B_n\cap R_n}(t)-t c_{in}(t)h(t)I_{B_n^c\cap R_n}(t).
\end{equation}
Consequently, it is sufficient to prove that the integrals of both terms in \eqref{Dposterior:eq:function-x-n-i-t} converge to 0. For the first term, similarly to Lemma \ref{Dposterior:lem:posterior-t}, we apply the dominated convergence theorem. The integrals of the second terms are bounded by absolute moments of the variable $Z I_{B_n^c\cap R_n}$, where $Z$ follows a normal distribution with mean equal to 0 and variance equal to $I^D(T(q))^{-1}$, hence, applying again the dominated convergence theorem one has $\int t c_{in}(t)h(t)I_{B_n^c\cap R_n}(t)dt\to 0$, for $i=0,1,2,3$.

Finally, to prove that $\int_{B_n^c\cap R_n} t \overline{\pi}_D^n(t|q)dt\to 0$, making use of \eqref{Dposterior:cond:sep}, \eqref{Dposterior:eq:expansion-Dposterior-t} and \eqref{Dposterior:eq:integrando}, one obtains
\begin{align*}
\bigg|\int_{B_n^c\cap R_n} t \overline{\pi}_D^n(t|q)dt\bigg| &\leq\bigg|\pi(T(q))^{-1}\left(\frac{I^D(T(q))}{2\pi}\right)^{1/2}(1+c_n)e^{-\Delta_{n-1}\rho}\\
&\phantom{\leq}\cdot\int_{B_n^c\cap R_n}t\pi\bigg(T(q)+\frac{t}{\Delta_{n-1}^{1/2}}\bigg)dt\bigg|\\
&\leq \pi(T(q))^{-1}\left(\frac{I^D(T(q))}{2\pi}\right)^{1/2}|1+c_n|\Delta_{n-1}e^{-\Delta_{n-1}\rho}\\
&\phantom{\leq}\cdot\left(\int_{\Theta}|\theta|\pi(\theta)d\theta+|T(q)|\right)\to 0.
\end{align*}

\vspace*{1ex}

For \ref{Dposterior:thm:aprox-MDE-ii-MDAP}, let denote $f_{n}(\theta)=\Delta_{n-1}^{1/2}(\theta-T(q))$ and $\widetilde{C}_{n}^{+}=f_{n}(C_{n}^+)$. Without loss of generality, we can assume that $T(q)\in C_{n}^+$ for each $n\geq n_0$. Note that $f_{n}(\cdot)$ is a strictly increasing homeomorphism between $\Theta$ and $R_{n}$. As a consequence, $\widetilde{C}_{n}^{+}$ is a compact set and $0\in \widetilde{C}_{n}^{+}$, for each $n\geq n_0$. For $n\geq n_0$, since $q\in\Gamma_{n}^+$, one obtains $\widetilde{T}_n(q)=T(q) + \frac{t_{n}^*}{\Delta_{n-1}^{1/2}}$, with $t_{n}^*=\arg\max_{t\in\widetilde{C}_{n}^{+}} \overline{\pi}_D^{n}(t|q)=\arg\max_{t\in R_n} \overline{\pi}_D^{n}(t|q)=\arg\max_{t\in\R} \overline{\pi}_D^{n}(t|q)$; hence, it suffices to prove that $t_{n}^*\to 0$.

As $h(t)$ is the density function of a normal distribution with mean equal to 0 and variance equal to $I^D(T(q))^{-1}$, it has a unique maximum at $t=0$; thus, we shall prove that $\arg\max_{t\in \R}\overline{\pi}_D^{n}(t|q)\to \arg\max_{t\in \R}h(t)$, as $n\to\infty$, by proving that
\begin{equation}\label{Dposterior:eq:lim-sup-R}
\lim_{n\to\infty}\sup_{t\in\R} |\overline{\pi}_D^{n}(t|q)-h(t)|= 0.
\end{equation}

To that end, recall that $\overline{\pi}_D^{n}(t|q)=\frac{h_n(t)}{J_n}$; consequently,
{\small\begin{align}\label{Dposterior:eq:aprox-MDE_MDAP-sup-conv}
\sup_{t\in\R} |\overline{\pi}_D^{n}(t|q)- h(t)| &\leq\sup_{t\in\R} \left|\overline{\pi}_D^{n}(t|q)- \frac{h_n(t)}{\pi(T(q))}\left(\frac{I^D(T(q))}{2\pi}\right)^{1/2}\right|I_{B_n}(t)+ \sup_{t\in\R}\overline{\pi}_D^{n}(t|q)I_{B_n^c}(t)\nonumber\\
&\phantom{=}+\sup_{t\in\R} \left|\frac{h_n(t)}{\pi(T(q))}\left(\frac{I^D(T(q))}{2\pi}\right)^{1/2}- h(t)\right|I_{B_n}(t)+\sup_{t\in\R} h(t)I_{B_n^c}(t),
\end{align}}
thus, it is enough to prove that all the terms in \eqref{Dposterior:eq:aprox-MDE_MDAP-sup-conv} converge to 0, as $n\to\infty$.

For the first term in \eqref{Dposterior:eq:aprox-MDE_MDAP-sup-conv}, one has the following inequality
\begin{align*}
\sup_{t\in\R} &\left|\overline{\pi}_D^{n}(t|q)- \frac{h_n(t)}{\pi(T(q))}\left(\frac{I^D(T(q))}{2\pi}\right)^{1/2}\right|I_{B_n}(t)\leq\left|\frac{1}{J_n}-\frac{1}{\pi(T(q))}\left(\frac{I^D(T(q))}{2\pi}\right)^{1/2}\right|\\
&\phantom{\leq} \cdot \bigg(\pi(T(q))+\frac{|b_1| \pi(T(q))}{\Delta_{n-1}^{1/2}}\sup_{t\in\R}|t| e^{-\frac{t^2 (I^D(T(q))-\epsilon)}{2}}\\
&\phantom{\leq}+\frac{|b_2|\pi(T(q)) }{2\Delta_{n-1}}\sup_{t\in\R}t^2 e^{-\frac{t^2 (I^D(T(q))-\epsilon)}{2}}+\frac{M}{6\Delta_{n-1}^{3/2}}\sup_{t\in\R} |t|^3 e^{-\frac{t^2 (I^D(T(q))-\epsilon)}{2}}\bigg),
\end{align*}
where $M$ is an upper bound of the function $\pi'''(\cdot)$. Thus, it converges to 0 due to the fact that $\frac{1}{J_n}\to \frac{1}{\pi(T(q))}\left(\frac{I^D(T(q))}{2\pi}\right)^{1/2}$, and $|t|^i e^{-\frac{t^2 (I^D(T(q))-\epsilon)}{2}}$ is a bounded function for $i=1,2,3$.

Since $\pi(\cdot)$ is bounded, for the second term in \eqref{Dposterior:eq:aprox-MDE_MDAP-sup-conv}, one has
\begin{equation*}
\sup_{t\in\R} \overline{\pi}_D^{n}(t|q)I_{B_n^c}(t)\leq\sup_{t\in\R} \frac{e^{-\rho\Delta_{n-1}}}{J_n}\pi\left(T(q)+\frac{t}{\Delta_{n-1}}\right)I_{B_n^c\cap R_n}(t)\leq\frac{e^{-\rho\Delta_{n-1}}}{J_n}M^* \to 0,
\end{equation*}
where $M^*$ denotes an upper bound of the function $\pi(\cdot)$.

For the third term in \eqref{Dposterior:eq:aprox-MDE_MDAP-sup-conv}, one has
\begin{align*}
\sup_{t\in\R} &\left|\left(\frac{I^D(T(q))}{2\pi}\right)^{1/2}\frac{h_n(t)}{\pi(T(q))}- h(t)\right|I_{B_n}(t)\leq \sup_{t\in\R} h(t)I_{R_n^c}(t)\\
&\phantom{\leq} +\left(\frac{I^D(T(q))}{2\pi}\right)^{1/2} \sup_{t\in\R} \left|e^{-\frac{t^2 I^D(\theta_n'(t))}{2}}-e^{-\frac{t^2 I^D(T(q))}{2}}\right|I_{B_n\cap R_n}(t)\\
&\phantom{\leq}+\left(\frac{I^D(T(q))}{2\pi}\right)^{1/2} \cdot \bigg(\frac{|b_1|}{\Delta_{n-1}^{1/2}}\sup_{t\in\R}|t| e^{-\frac{t^2 (I^D(T(q))-\epsilon)}{2}}+\frac{|b_2|}{2\Delta_{n-1}}\sup_{t\in\R}t^2 e^{-\frac{t^2 (I^D(T(q))-\epsilon)}{2}}\\
&\phantom{\leq}+\frac{M}{6\Delta_{n-1}^{3/2}\pi(T(q))}\sup_{t\in\R} |t|^3 e^{-\frac{t^2 (I^D(T(q))-\epsilon)}{2}}\bigg).
\end{align*}
As a result, since $\sup_{t\in\R} h(t)I_{R_n^c}(t)\leq h(\eta\Delta_{n-1}^{1/2})\to 0$, $\Delta_n\to\infty$, and $|t|^i e^{-\frac{t^2 (I^D(T(q))-\epsilon)}{2}}$ is a bounded function for $i=1,2,3$, to establish the convergence of the third term in \eqref{Dposterior:eq:aprox-MDE_MDAP-sup-conv}, we shall prove
\begin{equation}
\sup_{t\in\R} \left|e^{-\frac{t^2 I^D(\theta_n'(t))}{2}}-e^{-\frac{t^2 I^D(T(q))}{2}}\right|I_{B_n\cap R_n}(t)\to 0.\label{Dposterior:eq:aprox-MDE_MDAP-sup-1a}
\end{equation}
To that end, we prove that for each $0<\epsilon<I^D(T(q))$,
\begin{equation}\label{Dposterior:eq:lim-MDAP}
\lim_{n\to\infty}\sup_{t\in\R} \left|e^{-\frac{t^2 I^D(\theta_n'(t))}{2}}-e^{-\frac{t^2 I^D(T(q))}{2}}\right|I_{B_n\cap R_n}(t)\leq \frac{\epsilon}{I^D(T(q))-\epsilon},
\end{equation}
thus, by taking limit as $\epsilon\to 0$,  \eqref{Dposterior:eq:aprox-MDE_MDAP-sup-1a} follows. For each $t\in B_{n}\cap R_n$, let consider the function $h_t(x)=e^{-\frac{t^2}{2}x}$, $x\geq 0$, which satisfies $h_t'(x)=-\frac{t^2}{2}e^{-\frac{t^2}{2}x}$. From the mean value theorem, for each $t\in B_{n}\cap R_n$,
\begin{equation*}
\left|e^{-\frac{t^2 I^D(\theta_n'(t))}{2}}-e^{-\frac{t^2 I^D(T(q))}{2}}\right|\leq \frac{t^2}{2}e^{-\frac{t^2}{2}a_{t,n}} |I^D(\theta_n'(t))-I^D(T(q))|,
\end{equation*}
for $a_{t,n}$ between $I^D(\theta_n'(t))$ and $I^D(T(q))$. Thus, $a_{t,n}>I^D(T(q))-\epsilon$, and $e^{-\frac{t^2}{2}a_{t,n}}\leq e^{-\frac{t^2}{2}(I^D(T(q))-\epsilon)}$, where $\frac{t^2}{2}e^{-\frac{t^2}{2}(I^D(T(q))-\epsilon)}$ is a function bounded by $\frac{1}{I^D(T(q))-\epsilon}$. As a consequence,
\begin{equation*}
\sup_{t\in B_{n}\cap R_n}\left|e^{-\frac{t^2 I^D(\theta_n'(t))}{2}}-e^{-\frac{t^2 I^D(T(q))}{2}}\right|\leq \frac{\epsilon}{I^D(T(q))-\epsilon},
\end{equation*}
and \eqref{Dposterior:eq:lim-MDAP} follows.

For the forth term in \eqref{Dposterior:eq:aprox-MDE_MDAP-sup-conv}, one has $\sup_{t\in\R} h(t)I_{B_n^c}(t)= h(\delta\Delta_{n-1}^{1/2})\to 0$.

From \eqref{Dposterior:eq:lim-sup-R}, one has that $|\overline{\pi}_D^{n}(t_{n}^*|q)- h(t_{n}^*)|\to 0$, and $\overline{\pi}_D^{n}(t_{n}^*|q)\to h(0)$. The former is immediate from the inequality $|\overline{\pi}_D^{n}(t_{n}^*|q)- h(t_{n}^*)|\leq\sup_{t\in\R} |\overline{\pi}_D^{n}(t|q)- h(t)|$, and the latter follows from $|\overline{\pi}_D^{n}(t_{n}^*|q)- h(0)|\leq \sup_{t\in\R} |\overline{\pi}_D^{n}(t|q)- h(t)|$; now, it is straightforward that $h(t_{n}^*)\to h(0)$.

If $\{t_{n}^*\}_{n\in\N}$ does not converge to 0, then there exists $\varepsilon >0$, such that for each $N_0\in\N$, $|t_{N}^*|>\varepsilon$, for some $N>N_0$. Consequently, $|h(t_{N}^*)-h(0)|>h(0)-h(\varepsilon)>0$. Therefore, for $\nu=h(0)-h(\varepsilon)$, and each $N_0\in\N$, there exists $N>N_0$ satisfying $|h(t_{N}^*)-h(0)|>\nu$, which contradicts $h(t_{n}^*)\to h(0)$. This completes the proof of the theorem.

\end{Prf}

\subsection{Proof of asymptotic properties of the EDAP and MDAP estimators}\label{Dposterior:ap:asymp}

For simplicity, we will assume that $P[Z_n\to\infty]=1$. Moreover, unless specified otherwise, we will assume that all the limits are taken as $n\to\infty$ and we shall keep the same notation throughout this section. Let fix $\omega\in\{Z_n\to\infty\}$ and henceforth, in order to lighten the notation, we will consider that all the random variables are evaluated at $\omega$.

\vspace{0.75cm}
\begin{Prf}[Theorem \ref{Dposterior:thm:consistency}]
In order to prove this theorem, we are to make use of next lemma.

\begin{lem}\label{Dposterior:lemma:wn}
Under the hypotheses of Theorem \ref{Dposterior:thm:consistency}, let write for $t\in\R$,
\begin{equation*}
w_n(t)=\pi\left(\hat{\theta}_n^D + \frac{t}{\Delta_{n-1}^{1/2}}\right)e^{-\Delta_{n-1}\left(D\left(\hat{p}_n,\hat{\theta}_n^D+\frac{t}{\Delta_{n-1}^{1/2}}\right)- D(\hat{p}_n,\hat{\theta}_n^D)\right)}-\pi(\theta_p)e^{-\frac{1}{2}t^2 I^D(\theta_p)}.
\end{equation*}
Then, as $n\to\infty$,
\begin{enumerate}[label=(\roman*),ref=\emph{(\roman*)}]
\item $\int |w_n(t)|dt\rightarrow 0$\quad a.s. on $\{Z_n\to\infty\}$.\label{Dposterior:lemma:i}
\item $\int |t| |w_n(t)|dt\rightarrow 0$\quad a.s. on $\{Z_n\to\infty\}.$\label{Dposterior:lemma:ii}
\end{enumerate}
\end{lem}

\begin{proof}
Let denote $R_n=\{\Delta_{n-1}^{1/2}(\theta-\hat\theta_n^D):\theta\in\Theta\}$, $n\in\N$. Since $\theta_p\in int(\Theta)$, from \eqref{Dposterior:eq:consistencyMDE}, one has that $\hat{\theta}_n^D$ is eventually in $int(\Theta)$ and, as a result, in an analogous manner to that in Lemma \ref{Dposterior:lem:posterior-t}, one obtains $\cup_{n=1}^\infty R_n =\R$.

\ref{Dposterior:lemma:i} Observe that $w_n(t)I_{R_n^c}(t)=-\pi(\theta_p)e^{-\frac{1}{2}t^2 I^D(\theta_p)}I_{R_n^c}(t)$, consequently, using the dominated convergence theorem, since $I_{R_n^c}(t)\to 0$, $|w_n(t)I_{R_n^c}(t)|\leq \pi(\theta_p)e^{-\frac{1}{2}t^2 I^D(\theta_p)}$, for each $t\in\R$, and $\int e^{-\frac{1}{2}t^2 I^D(\theta_p)} dt<\infty$, one has that $\int_{R_n^c}w_n(t)dt\to 0$.

On the other hand, for each fixed $t\in\R$, using a first-order Taylor series expansion, one has
\begin{equation*}
D\left(\hat{p}_n,\hat{\theta}_n^D+\frac{t}{\Delta_{n-1}^{1/2}}\right)=D(\hat{p}_n,\hat{\theta}_n^D)+\frac{t}{\Delta_{n-1}^{1/2}}\dot{D}(\hat{p}_n,\hat{\theta}_n^D)+\frac{1}{2}\left(\frac{t}{\Delta_{n-1}^{1/2}}\right)^2\ddot{D}(\hat{p}_n,\theta_n'(t)),
\end{equation*}
where $\theta_n'(t)$ is a real number between $\hat{\theta}_n^D$ and $\hat{\theta}_n^D+\frac{t}{\Delta_{n-1}^{1/2}}$. Observe that $\theta_n'(t)\in\Theta$ if and only if $t\in R_n$; hence, for each $t\in\R$ fixed, there exists $n_0=n_0(t)$ such that $\theta_n'(t)\in\Theta$ for each $n\geq n_0$. Thus,
\begin{equation*}
w_n(t)=\pi\left(\hat{\theta}_n^D+\frac{t}{\Delta_{n-1}^{1/2}}\right)e^{-\frac{t^2}{2}I_n^D(\theta_n'(t))}-\pi\left(\theta_p\right)e^{-\frac{t^2}{2}I^D(\theta_p)}.
\end{equation*}

For each $t\in\R$, from \eqref{Dposterior:eq:consistencyMDE}, one has that $\{\theta_n'(t)\}_{n\in\N}$ is a sequence that converges to $\theta_p$, and by \ref{Dposterior:dis:assum}~\ref{Dposterior:cond:cont-second-deriv-D}, $I_n^D(\theta_n'(t))\to I^D(\theta_p)$. Moreover, the continuity of $\pi(\cdot)$ and the fact that $\hat{\theta}_n^D+\frac{t}{\Delta_{n-1}^{1/2}}\to\theta_p$, for each $t\in\R$ fixed, guarantee that $\pi\left(\hat{\theta}_n^D+\frac{t}{\Delta_{n-1}^{1/2}}\right)\to\pi(\theta_p)$; hence, $w_n(t)\to 0$, for each $t\in\R$.

Let $0<\epsilon<I^D(\theta_p)$. Since $\pi(\cdot)$ is continuous, there exists $\delta_1=\delta_1(\epsilon)$ such that if $|\theta-\theta_p|\leq\delta_1$, then $|\pi(\theta)-\pi(\theta_p)|\leq\epsilon$. From \ref{Dposterior:dis:assum}¬ \ref{Dposterior:cond:cont-second-deriv-D}, one also has that there exist $n_1=n_1(\epsilon)$, and $\delta_2=\delta_2(\epsilon)$ such that if $|\theta-\theta_p|\leq\delta_2$, and $n\geq n_1$, then $|I_n^D(\theta)-I^D(\theta_p)|\leq\epsilon$. In addition, we can safely assume that $\delta_1=\delta_2$. On the other hand, from \eqref{Dposterior:eq:consistencyMDE}, the existence of $n_2=n_2(\delta)\in\N$ such that $|\hat{\theta}_n^D-\theta_p|\leq\delta_1/2$, for each $n\geq n_2$, is guaranteed.

Let fix $\delta=\delta_1/2$, $N_0=\max(n_1,n_2)=N_0(\delta_1,\epsilon)$, and define the sets $B_n=\{t\in\R:|t|\leq \delta\Delta_{n-1}^{1/2}\}$, $n\in\N$. Firstly, we shall prove $\int_{B_n\cap R_n}|w_n(t)|dt\to 0$.

Note that for each $t\in B_n\cap R_n$, $n\geq N_0$, $\left|\pi\left(\hat{\theta}_n^D+\frac{t}{\Delta_{n-1}^{1/2}}\right)-\pi(\theta_p)\right|\leq\epsilon$, and $|I_n^D(\theta_n'(t))-I^D(\theta_p)|\leq\epsilon$. Consequently, for each $n\geq N_0$,
\begin{eqnarray*}
|w_n(t)|I_{B_{n}\cap R_n}(t)&\leq&\left(\pi\left(\hat{\theta}_n^D+\frac{t}{\Delta_{n-1}^{1/2}}\right)e^{-\frac{t^2}{2}I_n^D(\theta_n'(t))}+\pi\left(\theta_p\right)e^{-\frac{t^2}{2}I^D(\theta_p)}\right)I_{B_{n}\cap R_n}(t)\\
&\leq& 2\left(\pi(\theta_p)+\epsilon\right)e^{-\frac{t^2}{2}\left(I^D(\theta_p)-\epsilon\right)},
\end{eqnarray*}
with $\int e^{-\frac{t^2}{2}\left(I^D(\theta_p)-\epsilon\right)} dt =\left(\frac{2\pi}{I^D(\theta_p)-\epsilon}\right)^{1/2}<\infty$.

Since $w_n(t)\to 0$, in particular, $|w_n(t)|I_{B_n\cap R_n}(t)\to 0$, and by applying the dominated convergence theorem one obtains $\int_{B_n\cap R_n}|w_n(t)|dt\to 0$.

Now, we shall prove $\int_{B_n^c\cap R_n}|w_n(t)|dt\to 0$. From the previous Taylor series expansion, for $t\in\R$, one has
\begin{equation*}
D\left(\hat{p}_n,\hat{\theta}_n^D+\frac{t}{\Delta_{n-1}^{1/2}}\right)-D(\hat{p}_n,\hat{\theta}_n^D)=\frac{I_n^D(\theta_n'(t))}{2}\frac{t^2}{\Delta_{n-1}}.
\end{equation*}

Bearing in mind \eqref{Dposterior:cond:sep}, \eqref{Dposterior:eq:consistencyMDE}, \ref{Dposterior:dis:assum}~\ref{Dposterior:cond:unif-cont-disparity}, and the fact that $\hat{p}_n\to p$ in $l_1$, the existence of some $\rho>0$ and $N_*\in\N$ satisfying
\begin{equation}\label{Dposterior:eq:inf-lem-i}
\inf_{t\in B_n^c\cap R_n} D\left(\hat{p}_n,\hat{\theta}_n^D+\frac{t}{\Delta_{n-1}^{1/2}}\right)-D(\hat{p}_n,\hat{\theta}_n^D)>\rho,
\end{equation}
for $n\geq N_*$ is guaranteed, and as a consequence, for each $t\in B_n^c\cap R_n$, $n\geq N_*$, $I_n^D(\theta_n'(t))>\frac{2}{t^2}\rho\Delta_{n-1}$, therefore $e^{-\frac{t^2}{2}I_n^D(\theta_n'(t))}\leq e^{-\rho\Delta_{n-1}}$. Hence, for each $n\geq N_*$, one has
{\small\begin{eqnarray*}
\int_{B_n^c\cap R_n} |w_n(t)|dt&\leq& e^{-\Delta_{n-1}\rho}\int_{B_n^c\cap R_n}\pi\left(\hat{\theta}_n^D + \frac{t}{\Delta_{n-1}^{1/2}}\right) dt+\pi(\theta_p)\int_{B_n^c\cap R_n}e^{-\frac{1}{2}t^2 I^D(\theta_p)} dt\\
&\leq&e^{-\Delta_{n-1}\rho}\Delta_{n-1}^{1/2}+ \pi(\theta_p)\left(\frac{2\pi}{I^D(\theta_p)}\right)^{1/2} P[|Z|>\delta\Delta_{n-1}^{1/2}],
\end{eqnarray*}}
where $Z$ follows a normal distribution with mean equal to 0 and variance equal to $I^D(\theta_p)^{-1}$; therefore, $\int_{B_n^c\cap R_n} |w_n(t)|dt\rightarrow 0$.

\ref{Dposterior:lemma:ii} Firstly, note that for each $t\in\R$ fixed, $|t| |w_n(t)|\to 0$. As done before, $\int_{R_n^c} |t| |w_n(t)|dt\rightarrow 0$, thus, it is enough to prove that $\int_{B_n\cap R_n} |t| |w_n(t)|dt\rightarrow 0$, and $\int_{B_n^c\cap R_n} |t| |w_n(t)|dt\rightarrow 0$.

On the one hand, for each $n\geq N_0$,
\begin{equation*}
|t||w_n(t)|I_{B_n\cap R_n}(t)\leq 2|t| \left(\pi(\theta_p)+\epsilon\right)e^{-\frac{t^2}{2}\left(I^D(\theta_p)-\epsilon\right)},
\end{equation*}
and $\int_{B_n\cap R_n} |t| |w_n(t)|dt\rightarrow 0$ follows from dominated convergence theorem and bearing in mind that
\begin{equation*}
\int |t|e^{-\frac{t^2}{2}\left(I^D(\theta_p)-\epsilon\right)} dt = \frac{2}{I^D(\theta_p)-\epsilon}.
\end{equation*}

On the other hand, for each $n\geq N_*$, one has
\begin{eqnarray*}
\int_{B_n^c\cap R_n} |t||w_n(t)|dt&\leq & e^{-\Delta_{n-1}\rho}\int_{B_n^c\cap R_n}|t|\pi\left(\hat{\theta}_n^D + \frac{t}{\Delta_{n-1}^{1/2}}\right) dt\\
&\phantom{=}&+\pi(\theta_p)\int_{B_n^c\cap R_n}|t|e^{-\frac{1}{2}t^2 I^D(\theta_p)} dt\\
&\leq & e^{-\Delta_{n-1}\rho}\Delta_{n-1}^{1/2}\int_\Theta|\theta|\pi\left(\theta\right) d\theta+e^{-\Delta_{n-1}\rho}\Delta_{n-1}|\hat{\theta}_n^D|\\
&\phantom{=}&+ \pi(\theta_p)\left(\frac{2\pi}{I^D(\theta_p)}\right)^{1/2} E[|Z|I_{B_n^c}],
\end{eqnarray*}
where $Z$ follows a normal distribution with mean equal to 0 and variance equal to $I^D(\theta_p)^{-1}$; therefore, $\int_{B_n^c\cap R_n} |w_n(t)|dt\rightarrow 0$.
\end{proof}

\vspace{2ex}

Now, we show the proof of Theorem \ref{Dposterior:thm:consistency}.

\ref{Dposterior:thm:consistency-l1-param} From the fact that $\hat{p}_n\to p$ in $l_1$, \ref{Dposterior:dis:assum}~\ref{Dposterior:cond:unif-cont-disparity} and the continuity of $D(q,\cdot)$ on $\Theta$ for each $q\in\Gamma$, one has that $\sup_{\theta\in\Theta}|D(\hat{p}_n,\theta)-D(p,\theta)|\to 0$, hence $D(\hat{p}_n,\hat{\theta}_n^D)\to D(p,\theta_p)$.

Recall that $\dot{D}(\hat{p}_n,\hat{\theta}_n^D)= \dot{D}(p,\theta_p)=0$, and $I_n^D(\hat{\theta}_n^D)\to I^D(\theta_p)$.

For each $t\in R_n$, we can write
\begin{equation*}
\overline{\pi}_D^{n}(t|\hat{p}_n)=K_n^{-1}\pi\left(\hat{\theta}_n^D + \frac{t}{\Delta_{n-1}^{1/2}}\right)e^{-\Delta_{n-1}\left(D\left(\hat{p}_n,\hat{\theta}_n^D+\frac{t}{\Delta_{n-1}^{1/2}}\right)-D\left(\hat{p}_n,\hat{\theta}_n^D\right)\right)},
\end{equation*}
where $K_n=\bigint_{R_n} \pi\left(\hat{\theta}_n^D + \frac{t}{\Delta_{n-1}^{1/2}}\right)e^{-\Delta_{n-1}\left(D\left(\hat{p}_n,\hat{\theta}_n^D+\frac{t}{\Delta_{n-1}^{1/2}}\right)- D(\hat{p}_n,\hat{\theta}_n^D)\right)} dt$, and \linebreak $\overline{\pi}_D^{n}(t|\hat{p}_n)=0$, for $t\in R_n^c$. As a consequence,
\begin{align*}
\int &\Big|\overline{\pi}_D^{n}(t|\hat{p}_n)-\left(\frac{I^D(\theta_p)}{2\pi}\right)^{1/2} e^{-\frac{t^2}{2}I^D(\theta_p)}\Big|dt=\int_{R_n^c}\left(\frac{I^D(\theta_p)}{2\pi}\right)^{1/2} e^{-\frac{t^2}{2}I^D(\theta_p)}dt\\
&\phantom{=} +K_n^{-1}\int_{R_n}\Big|\pi\left(\hat{\theta}_n^D + \frac{t}{\Delta_{n-1}^{1/2}}\right)e^{-\Delta_{n-1}\left(D\left(\hat{p}_n,\hat{\theta}_n^D+\frac{t}{\Delta_{n-1}^{1/2}}\right)- D(\hat{p}_n,\hat{\theta}_n^D)\right)}\\
&\phantom{=}-K_n\left(\frac{I^D(\theta_p)}{2\pi}\right)^{1/2} e^{-\frac{t^2}{2}I^D(\theta_p)}\Big|dt\\
&\leq \int_{R_n^c}\left(\frac{I^D(\theta_p)}{2\pi}\right)^{1/2} e^{-\frac{t^2}{2}I^D(\theta_p)}dt+K_n^{-1}(I_{1n}+I_{2n}),
\end{align*}
where the first integral converges to 0 by applying the dominated convergence theorem, and

\begin{eqnarray*}
I_{1n}&=&\int_{R_n} |w_n(t)|dt,\\
I_{2n}&=& \int_{R_n} \Big|\pi(\theta_p)e^{-\frac{1}{2}t^2 I^D(\theta_p)}-K_n\left(\frac{I^D(\theta_p)}{2\pi}\right)^{1/2} e^{-\frac{t^2}{2}I^D(\theta_p)}\Big|dt.
\end{eqnarray*}

In the proof of Lemma \ref{Dposterior:lemma:wn}, we showed that $K_n\to \pi(\theta_p)\left(\frac{2\pi}{I^D(\theta_p)}\right)^{1/2}\neq 0$; hence, it is enough to prove that $I_{in}\to 0$, for $i=1,2$. From Lemma \ref{Dposterior:lemma:wn}, we also have that $I_{1n}\to 0$, and $\Big|\pi(\theta_p)-K_n\left(\frac{I^D(\theta_p)}{2\pi}\right)^{1/2}\Big|\bigint_{R_n} e^{-\frac{t^2}{2} I^D(\theta_p)}dt\to 0$.

\ref{Dposterior:thm:consistency-l1-param-t} We have
\begin{align*}
\int&|t|\Big|\overline{\pi}_D^{n}(t|\hat{p}_n)-\left(\frac{I^D(\theta_p)}{2\pi}\right)^{1/2} e^{-\frac{t^2}{2}I^D(\theta_p)}\Big|dt=\int_{R_n^c}|t|\left(\frac{I^D(\theta_p)}{2\pi}\right)^{1/2} e^{-\frac{t^2}{2}I^D(\theta_p)}dt\\
&\phantom{=} +K_n^{-1}\int_{R_n}|t|\Big|\pi\left(\hat{\theta}_n^D + \frac{t}{\Delta_{n-1}^{1/2}}\right)e^{-\Delta_{n-1}\left(D\left(\hat{p}_n,\hat{\theta}_n^D+\frac{t}{\Delta_{n-1}^{1/2}}\right)- D(\hat{p}_n,\hat{\theta}_n^D)\right)}\\
&\phantom{=}-K_n\left(\frac{I^D(\theta_p)}{2\pi}\right)^{1/2} e^{-\frac{t^2}{2}I^D(\theta_p)}\Big|dt\\
&\leq \int_{R_n^c}|t|\left(\frac{I^D(\theta_p)}{2\pi}\right)^{1/2} e^{-\frac{t^2}{2}I^D(\theta_p)}dt+K_n^{-1}(\bar{I}_{1n}+\bar{I}_{2n}),
\end{align*}
where again the first integral converges to 0 by applying the dominated convergence theorem, and
\begin{eqnarray*}
\bar{I}_{1n}&=&\int_{R_n} |t||w_n(t)|dt\\
\bar{I}_{2n}&=& \int_{R_n} |t|\Big|\pi(\theta_p)e^{-\frac{1}{2}t^2 I^D(\theta_p)}-K_n\left(\frac{I^D(\theta_p)}{2\pi}\right)^{1/2} e^{-\frac{t^2}{2}I^D(\theta_p)}\Big|dt.
\end{eqnarray*}

From Lemma \ref{Dposterior:lemma:wn}~\ref{Dposterior:lemma:ii}, we also have that $\bar{I}_{1n}\to 0$, and bearing in mind that $K_n\to \pi(\theta_p)\left(\frac{2\pi}{I^D(\theta_p)}\right)^{1/2}\neq 0$, we obtain $\bar{I}_{2n}=\Big|\pi(\theta_p)-K_n\left(\frac{I^D(\theta_p)}{2\pi}\right)^{1/2}\Big|\int_{R_n} |t|e^{-\frac{t^2}{2} I^D(\theta_p)}dt\to 0$.

\ref{Dposterior:thm:consistency-l1-estim} Note that
\begin{align*}
\int \bigg|\overline{\pi}_D^{n}(t|\hat{p}_n) &-\left(\frac{I_n^D(\hat{\theta}_n^D)}{2\pi}\right)^{1/2}e^{-\frac{1}{2}t^2 I_n^D(\hat{\theta}_n^D)}\bigg|dt\leq \\
&\leq\int \bigg|\overline{\pi}_D^{n}(t|\hat{p}_n) -\left(\frac{I^D(\theta_p)}{2\pi}\right)^{1/2}e^{-\frac{1}{2}t^2 I^D(\theta_p)}\bigg|dt\\
&\phantom{\leq}+\int\bigg|\left(\frac{I^D(\theta_p)}{2\pi}\right)^{1/2}e^{-\frac{1}{2}t^2 I^D(\theta_p)} -\left(\frac{I_n^D(\hat{\theta}_n^D)}{2\pi}\right)^{1/2}e^{-\frac{1}{2}t^2 I_n^D(\hat{\theta}_n^D)}\bigg|dt
\end{align*}

From \ref{Dposterior:thm:consistency-l1-param}, we have that the first integral converges to 0. Since $I_n^D(\hat{\theta}_n^D)\to I^D(\theta_p)$, using Glick's theorem (see \cite{Devroye-Gyorfi}, p.10), we also have that the second integral converges to 0.

\end{Prf}

\begin{Prf}[Theorem \ref{Dposterior:thm:consistency-EDAP}]
Let $t=\Delta_{n-1}^{1/2}(\theta-\hat{\theta}^D_n)$ and $\overline{\pi}_D^{n}(t|\hat{p}_n)$ its $D$-posterior density function at $\hat{p}_n$.

\ref{Dposterior:thm:conv-MDE-EDAP} Bearing in mind that $\theta_n^{*D}=\int_\Theta\theta \pi_D^n(\theta|\hat{p}_n)d\theta=\int_{R_n}\left(\hat{\theta}_n^D+\frac{t}{\Delta_{n-1}^{1/2}}\right)\overline{\pi}_D^{n}(t|\hat{p}_n)dt$, and Theorem \ref{Dposterior:thm:consistency}~\ref{Dposterior:thm:consistency-l1-param-t}, we obtain
\begin{equation*}
\Delta_{n-1}^{1/2}(\theta_n^{*D}-\hat\theta_n^D)=\int_{R_n} t \overline{\pi}_D^{n}(t|\hat{p}_n)dt\rightarrow\left(\frac{I^D(\theta_p)}{2\pi}\right)^{1/2}\int te^{-\frac{1}{2}t^2 I^D(\theta_p)}dt=0\quad a.s.,
\end{equation*}
on $\{Z_n\to\infty\}$.

\ref{Dposterior:thm:normality-EDAP} Note that $\Delta_{n-1}^{1/2}(\theta_n^{*D}-\theta_p)=\Delta_{n-1}^{1/2}(\theta_n^{*D}-\hat\theta_n^D)+\Delta_{n-1}^{1/2}(\hat\theta_n^D-\theta_p)$. From \eqref{Dposterior:eq:normalityMDE} and \ref{Dposterior:thm:conv-MDE-EDAP}, making use of Slutsky's theorem, one has
$$\Delta_{n-1}^{1/2}(\theta_n^{*D}-\theta_p)\xrightarrow[n\to\infty]{d} N(0,I^D(\theta_p)^{-1}).$$
\end{Prf}


\begin{Prf}[Theorem \ref{Dposterior:thm:lim-sup-D-post-estim}]
The proof is analogous of that in Theorem \ref{Dposterior:thm:aprox-MDE}~\ref{Dposterior:thm:aprox-MDE-ii-MDAP}. Let $R_{n}=\{\Delta_{n-1}^{1/2}(\theta-\hat{\theta}_n^D):\theta\in\Theta\}$. With analogous arguments to those in Lemma \ref{Dposterior:lem:posterior-t}, one can prove that for each $t\in\R$,
\begin{align*}
\overline{\pi}_D^{n}(t|\hat{p}_n)&=\left(\frac{I_n^D(\hat{\theta}_n^D)}{2\pi}\right)^{1/2}e^{-\frac{t^2 I_n^D(\theta_n'(t))}{2}}\bigg(1+\frac{tb_{1n}}{\Delta_{n-1}^{1/2}}+\frac{t^2b_{2n}}{2\Delta_{n-1}}+\frac{t^3 \pi'''(\theta_n^*(t))}{6\Delta_{n-1}^{3/2}\pi(\hat{\theta}_n^D))}\bigg)\\
&\phantom{=}\cdot(1+C_n)I_{R_n}(t),
\end{align*}
a.s. on $\{Z_n\to\infty\}$, where $b_{1n}=\frac{\pi'\left(\hat{\theta}_n^D\right)}{\pi\left(\hat{\theta}_n^D\right)}$, $b_{2n}=\frac{\pi''\left(\hat{\theta}_n^D\right)}{\pi\left(\hat{\theta}_n^D\right)}$, $\theta_n'(t)$ and $\theta_n^*(t)$ are both points between $\hat{\theta}_n^D$ and $\hat{\theta}_n^D+\frac{t}{\Delta_{n-1}^{1/2}}$, for each $n\in\N$, and $t\in\R$, and  $\{C_n\}_{n\in\N}$ is a sequence of real valued random variables converging to 0 a.s.

\vspace{0.5cm}

Moreover, the following inequality holds
\begin{align}
\sup_{t\in\R} |\overline{\pi}_D^{n}(t|\hat{p}_n)- \varphi(t;\theta_p)| &\leq\sup_{t\in\R} \left|\overline{\pi}_D^{n}(t|\hat{p}_n)-\frac{h_n(t)}{\pi(\hat{\theta}_n^D)}\left(\frac{I_n^D(\hat{\theta}_n^D)}{2\pi}\right)^{1/2}\right|I_{B_n}(t)\nonumber\\
&\phantom{\leq}+\sup_{t\in\R} \left|\left(\frac{I_n^D(\hat{\theta}_n^D)}{2\pi}\right)^{1/2}\frac{h_n(t)}{\pi(\hat{\theta}_n^D)}- \varphi(t;\theta_p)\right|I_{B_n}(t)\nonumber\\
&\phantom{\leq}+ \sup_{t\in\R} \overline{\pi}_D^{n}(t|\hat{p}_n)I_{B_n^c}(t)+ \sup_{t\in\R} \varphi(t;\theta_p)I_{B_n^c}(t),\label{Dposterior:eq:aprox-MDE_MDAP-sup-conv-asymp}
\end{align}
where $h_n(t)=\pi\bigg(\hat{\theta}_n^D+\frac{t}{\Delta_{n-1}^{1/2}}\bigg)e^{-\Delta_{n-1}D\left(\hat{p}_n,\hat{\theta}_n^D+\frac{t}{\Delta_{n-1}^{1/2}}\right)+\Delta_{n-1}D\left(\hat{p}_n,\hat{\theta}_n^D\right)}I_{R_n}(t)$, for each $n\in\N$ and $t\in\R$.

Following the same arguments as in Theorem \ref{Dposterior:thm:aprox-MDE}~\ref{Dposterior:thm:aprox-MDE-ii-MDAP}, one can verify that all the terms in \eqref{Dposterior:eq:aprox-MDE_MDAP-sup-conv-asymp} converge to 0 using that $\Delta_n\to\infty$, $\pi(\hat{\theta}_n^D)\to\pi(\theta_p)$, $I_n^D(\hat{\theta}_n^D)\to I^D(\theta_p)$, $J_n=\int h_n(t) dt$, $n\in\N$, satisfies $\frac{1}{J_n}-\frac{1}{\pi(\hat{\theta}_n^D)}\left(\frac{I_n^D(\hat{\theta}_n^D)}{2\pi}\right)^{1/2}\to 0$, $\pi(\cdot)$ is bounded, $\sup_{t\in\R} e^{-\frac{t^2 I^D(\theta_p)}{2}}=1$, and $|t|^i e^{-\frac{t^2 (I^D(\theta_p)-\epsilon)}{2}}$ is a bounded function for $i=1,2,3$.

\end{Prf}

\begin{Prf}[Theorem \ref{Dposterior:thm:consistency-MDAP}]
The proof is analogous of that in Theorem \ref{Dposterior:thm:aprox-MDE}~\ref{Dposterior:thm:aprox-MDE-ii-MDAP}.

To prove \ref{Dposterior:thm:conv-MDE-MDAP}, observe that there exists $n_0\in\N$ such that $\widetilde{T}_n(\hat{p}_n)$ exists and is unique for $n\geq n_0$.
Since $\hat{p}_n\in\Gamma_{n}^+$ for $n\geq n_0$, one has $\widetilde{T}_n(\hat{p}_n)=\hat{\theta}_n^D + \frac{t_{n}^*}{\Delta_{n-1}^{1/2}}$,
with $t_{n}^*=\arg\max_{t\in R_n} \overline{\pi}_D^{n}(t|\hat{p}_n)=\arg\max_{t\in\R} \overline{\pi}_D^{n}(t|\hat{p}_n)$; hence, it is sufficient to prove that $t_{n}^*\to 0$.

Note that, since $\varphi(t;\theta_p)$ is the density function of a normal distribution with mean equal to 0 and variance equal to $I^D(\theta_p)^{-1}$, it has a unique maximum at $t=0$; thus, we shall prove that $\arg\max_{t\in \R}\overline{\pi}_D^{n}(t|\hat{p}_n)\to\arg\max_{t\in \R}\varphi(t;\theta_p)$.

From Theorem \ref{Dposterior:thm:lim-sup-D-post-estim}, one has that $|\overline{\pi}_D^{n}(t_{n}^*|\hat{p}_n)- \varphi(t_{n}^*;\theta_p)|\to 0$, and $|\overline{\pi}_D^{n}(t_{n}^*|\hat{p}_n)-\varphi(0;\theta_p)|\to 0$. Now, it is straightforward that $|\varphi(t_{n}^*;\theta_p)- \varphi(0;\theta_p)|\to 0$, and as a result, $t_{n}^*\to 0$.

\ref{Dposterior:thm:normality-MDAP} is straightforward, using Slutsky's theorem, from \ref{Dposterior:dis:assum}~\ref{Dposterior:cond:MDE}, \ref{Dposterior:thm:conv-MDE-MDAP} and
\begin{equation*}
\Delta_{n-1}^{1/2}(\theta_n^{+D}-\theta_p)=\Delta_{n-1}^{1/2}(\theta_n^{+D}-\hat\theta_n^D)+\Delta_{n-1}^{1/2}(\hat\theta_n^D-\theta_p).
\end{equation*}
\end{Prf}

\subsection{Proof of robustness properties of the EDAP and MDAP estimators}\label{Dposterior:ap:robust}

As it was done in Section \ref{Dposterior:ap:asymp}, we will assume that $P[Z_n\to\infty]=1$ and we shall keep the same notation throughout this section. Let fix $\omega\in\{Z_n\to\infty\}$ and henceforth, in order to lighten the notation, we will consider that all the random variables are evaluated at $\omega$.

\newpage

\begin{Prf}[Theorem \ref{Dposterior:thm:influence-function}]
First of all, note that for each $L\in\N_0$ and $n\in\N$, the influence function satisfies
\begin{align*}
IF(L,\overline{T}_n,p)=
\frac{\partial}{\partial\alpha}\left(\frac{\int_\Theta\theta e^{-\Delta_{n-1}D(p(\theta_0,\alpha,L),\theta)}\pi(\theta)d\theta}{\int_\Theta e^{-\Delta_{n-1}D(p(\theta_0,\alpha,L),\theta)}\pi(\theta)d\theta}\right)_{\big|\alpha=0}.\label{Dposterior:eq:IF-deriv}
\end{align*}

Let denote $M=\sup_{\delta\in[-1,\infty)}|G(\delta)|<\infty$, $\overline{M}=\sup_{\delta\in[-1,\infty)}|G'(\delta)|<\infty$, and $F_{\theta_0}(\theta,\alpha,L)=e^{-\Delta_{n-1}D(p(\theta_0,\alpha,L),\theta)}\pi(\theta)$, for each $\theta\in\Theta$, $\alpha\in(0,1)$, $L\in\N_0$. Observe that, $F_{\theta_0}(\theta,\alpha,L)\geq e^{-\Delta_{n-1} M}\pi(\theta)$, and as a consequence,
\begin{equation*}
\int_{\Theta} F_{\theta_0}(\theta,\alpha,L)d\theta \geq e^{-\Delta_{n-1} M} \int_{\Theta}\pi(\theta)d\theta=e^{-\Delta_{n-1} M},
\end{equation*}
moreover,
\begin{eqnarray*}
\bigg|\frac
{\partial}{\partial\alpha}\left(\frac{\theta F_{\theta_0}(\theta,\alpha,L)}{\int_\Theta F_{\theta_0}(\theta,\alpha,L)d\theta}\right)\bigg|&=&\bigg|\frac{\theta\frac{\partial}{\partial\alpha} F_{\theta_0}(\theta,\alpha,L)}{\int_\Theta F_{\theta_0}(\theta,\alpha,L)d\theta}\\
&\phantom{=}& -\frac{\theta F_{\theta_0}(\theta,\alpha,L)\frac{\partial}{\partial\alpha}\left(\int_\Theta F_{\theta_0}(\theta,\alpha,L) d\theta\right)}{\left(\int_\Theta F_{\theta_0}(\theta,\alpha,L)d\theta\right)^2}\bigg|\\
&\leq& 4\overline{M}\Delta_{n-1}|\theta|\pi(\theta)e^{2\Delta_{n-1}M}.
\end{eqnarray*}
Let denote $C_L(p(\theta_0),\theta)=\frac{\partial}{\partial\alpha}D(p(\theta_0,\alpha,L),\theta)_{|_{\alpha=0}}$, and the expectation with respect the $D$-posterior of $\theta$ given the probability distribution $p=p(\theta_0)$ by $E_{\pi_D^n(\theta|p)}[\cdot]$. One obtains
\begin{align*}
IF(L,\overline{T}_n,p)&=\Delta_{n-1}\frac{\int_\Theta -\theta  C_L(p(\theta_0),\theta)e^{-\Delta_{n-1}D(p,\theta)}\pi(\theta)d\theta}{\int_\Theta e^{-\Delta_{n-1}D(p,\theta)}\pi(\theta)d\theta}\\
&\phantom{=}+\Delta_{n-1}\frac{\left(\int_\Theta C_L(p(\theta_0),\theta)e^{-\Delta_{n-1}D(p,\theta)}\pi(\theta)d\theta\right)\left(\int_\Theta\theta e^{-\Delta_{n-1}D(p,\theta)}\pi(\theta)d\theta\right)}{\left(\int_\Theta e^{-\Delta_{n-1}D(p,\theta)}\pi(\theta)d\theta\right)^2}\\
&=-\Delta_{n-1} E_{\pi_D^n(\theta|p)}\left[\theta C_L(p(\theta_0),\theta)\right]\\
&\phantom{=}+\Delta_{n-1} E_{\pi_D^n(\theta|p)}\left[\theta\right]E_{\pi_D^n(\theta|p)}\left[C_L(p(\theta_0),\theta)\right].
\end{align*}

Thus, it is enough to prove that, for each $n\in\N$, $|E_{\pi_D^n(\theta|p)}\left[C_L(p(\theta_0),\theta)\right]|<\infty$, and $|E_{\pi_D^n(\theta|p)}\left[\theta C_L(p(\theta_0),\theta)\right]|<\infty$. For the former, one has
\begin{equation*}
|E_{\pi_D^n(\theta|p)}\left[C_L(p(\theta_0),\theta)\right]|\leq\frac{\int_\Theta |C_L(p(\theta_0),\theta)|e^{-\Delta_{n-1}r}\pi(\theta)d\theta}{\int_\Theta e^{-\Delta_{n-1}R}\pi(\theta)d\theta} \leq 2\overline{M}e^{\Delta_{n-1}(R-r)}<\infty,
\end{equation*}
where $\inf_{\theta\in\Theta} D(p,\theta)= r\geq 0$, and $\sup_{\theta\in\Theta} D(p,\theta)= R\leq M <\infty$. Analogously,
\begin{align*}
|E_{\pi_D^n(\theta|p)}\left[\theta C_L(p(\theta_0),\theta)\right]|&\leq\frac{\int_\Theta |\theta| |C_L(p(\theta_0),\theta)|e^{-\Delta_{n-1}r}\pi(\theta)d\theta}{\int_\Theta e^{-\Delta_{n-1}R}\pi(\theta)d\theta}\\
&\leq 2\overline{M}e^{\Delta_{n-1}(R-r)}\int_\Theta |\theta|\pi(\theta)d\theta<\infty,
\end{align*}
hence, the results follows.
\end{Prf}

\begin{Prf}[Theorem \ref{Dposterior:thm:breakdown-point-EDAP-MDAP}]
To prove this theorem and some of the following results, we will make use of the following lemma.

\begin{lem}\label{Dposterior:lem:1-breakdown-EDAP}
Let $\alpha\in (0,1)$ and consider a family of probability distributions $\{\bar{q}_L\}_{L\in\N_0}$ satisfying \ref{Dposterior:cond:contam-distr}~\ref{Dposterior:cond:ortog-distrib-contam}-\ref{Dposterior:cond:ortog-contam-family}, and assume that conditions \ref{Dposterior:eq:con:bound-disp}~\ref{Dposterior:cond:G-disparity} and \ref{Dposterior:cond:contam-distr}~\ref{Dposterior:cond:ortog-distrib-family} hold. Then, for any $\delta>0$:
\begin{enumerate}[label=(\roman*),ref=\emph{(\roman*)}]
\item There exist $L_1\in\N_0$, and $M_1>0$ such that
$$\sup_{|\theta|>M_1}\sup_{L>L_1} |D((1-\alpha)p+\alpha\bar{q}_{L},\theta)-D(\alpha\bar{q}_L,\theta)|<\delta.$$\label{Dposterior:lem:1-breakdown-EDAP-i}
\item For all $M_2>0$, there exists $L_2\in\N_0$ such that
$$\sup_{|\theta|<M_2}\sup_{L>L_2} |D((1-\alpha)p+\alpha\bar{q}_{L},\theta)-D((1-\alpha)p,\theta)|<\delta.$$\label{Dposterior:lem:1-breakdown-EDAP-ii}
\end{enumerate}
The results also hold for the Hellinger distance.
\end{lem}

\noindent The proof is analogous to that for Lemma 1 in \cite{Hooker-Vidyshankar-2014} and it is omitted.
\vspace{0.5cm}

Now, we shall prove Theorem \ref{Dposterior:thm:breakdown-point-EDAP-MDAP}.

\ref{Dposterior:thm:breakdown-point-i-EDAP} We follow the same steps as in the proof of Corollary 2 in \cite{Hooker-Vidyshankar-2014}. Since $G(\cdot)$ is strictly convex, $\inf_{\theta\in\Theta,q\in\Gamma} D(q,\theta)= r\geq 0$, and under the assumptions, $\sup_{\theta\in\Theta,q\in\Gamma} D(q,\theta)= R<\infty$. For each $\theta\in\Theta$, $\alpha\in [0,1]$, and $q\in\Gamma$, one has that  $e^{-\Delta_{n-1}R}\leq e^{-\Delta_{n-1}D((1-\alpha)p+\alpha q,\theta)} \leq e^{-\Delta_{n-1}r}$, and as a consequence,
\begin{equation*}
e^{-\Delta_{n-1}(R-r)} \int_\Theta|\theta|\pi(\theta)d\theta\leq |\overline{T}_n((1-\alpha)p+\alpha q)|\leq e^{-\Delta_{n-1}(r-R)}\int_\Theta|\theta|\pi(\theta)d\theta,
\end{equation*}
hence, the result yields.

\ref{Dposterior:thm:breakdown-point-ii-MDAP} 
Throughout this proof, all the limits are taken as $L\to\infty$ unless specified otherwise. Let fix $\alpha\in (0,1)$, and assume that there is a breakdown at level $\alpha$, that is, $|\theta_L|\to\infty$, with $\theta_L=\arg\max_{\theta\in\Theta}\pi_D^n(\theta | (1-\alpha)p+\alpha \overline{q}_L)$.
From the definition of $\theta_L$,
\begin{equation}\label{Dposterior:eq:ineq-breakdown-point-MDAP}
\pi_D^n(\theta_0 | (1-\alpha)p+\alpha \overline{q}_L)\leq \pi_D^n(\theta_L | (1-\alpha)p+\alpha \overline{q}_L).
\end{equation}
From Lemma \ref{Dposterior:lem:1-breakdown-EDAP}~\ref{Dposterior:lem:1-breakdown-EDAP-i} and \ref{Dposterior:lem:1-breakdown-EDAP-ii} one obtains, respectively,
\begin{align}
\lim_{L\to\infty} D((1-\alpha)p+\alpha \overline{q}_L,\theta_L)&=\lim_{L\to\infty}D(\alpha\overline{q}_L,\theta_L),\label{Dposterior:eq:lim1-breakdown-point-MDAP}\\
\lim_{L\to\infty} D((1-\alpha)p+\alpha \overline{q}_L,\theta)&=D((1-\alpha)p,\theta),\quad \text{ for each }\theta\in\Theta,\label{Dposterior:eq:lim2-breakdown-point-MDAP}
\end{align}
then, taking limit as $L\to\infty$ in \eqref{Dposterior:eq:ineq-breakdown-point-MDAP} we obtain a contradiction, and as a consequence, there is no breakdown at level $\alpha$, for each $\alpha\in (0,1)$. Indeed, from \eqref{Dposterior:eq:lim2-breakdown-point-MDAP} and the dominated convergence theorem one has
\begin{equation*}
\lim_{L\to\infty} \int_{\Theta}e^{-\Delta_{n-1}D((1-\alpha)p+\alpha \overline{q}_L,\theta)}\pi(\theta)d\theta= \int_{\Theta}e^{-\Delta_{n-1}D((1-\alpha)p,\theta)}\pi(\theta)d\theta,
\end{equation*}
thus, for the first term in \eqref{Dposterior:eq:ineq-breakdown-point-MDAP}, since $D((1-\alpha)p,\theta_0)=G(-\alpha)$, we have the following limit
\begin{align*}
\lim_{L\to\infty} \pi_D^n(\theta_0 | (1-\alpha)p+\alpha \overline{q}_L)=\frac{e^{-\Delta_{n-1}G(-\alpha)}\pi(\theta_0)}{\int_{\Theta}e^{-\Delta_{n-1}D((1-\alpha)p,\theta)}\pi(\theta)d\theta}>0.
\end{align*}
For the right hand side in \eqref{Dposterior:eq:ineq-breakdown-point-MDAP}, from Jensen's inequality, one has $D(\alpha\overline{q}_L,\theta)\geq G(\alpha-1)$, and $D((1-\alpha)p,\theta)\geq G(-\alpha)$, for each $\theta\in\Theta$. Consequently, using \eqref{Dposterior:eq:lim1-breakdown-point-MDAP},
\begin{align*}
\lim_{L\to\infty} \pi_D^n(\theta_L | (1-\alpha)p+\alpha \overline{q}_L)
&=\frac{\lim_{L\to\infty} e^{-\Delta_{n-1}D(\alpha\overline{q}_L,\theta_L)}\pi(\theta_L)}{\int_{\Theta}e^{-\Delta_{n-1}D((1-\alpha)p,\theta)}\pi(\theta)d\theta}\\
&\leq \frac{e^{-\Delta_{n-1} G(\alpha-1)}}{\int_{\Theta}e^{-\Delta_{n-1}D((1-\alpha)p,\theta)}\pi(\theta)d\theta} \lim_{L\to\infty}\pi(\theta_L)=0,
\end{align*}
where we have applied that $\lim_{|\theta|\to\infty}\pi(\theta)=0$.
\end{Prf}


\begin{Prf}[Theorem \ref{Dposterior:thm:asympt-breakdown}]
In order to prove this theorem, we will make use of the following lemma.

\begin{lem}\label{Dposterior:lem:2-breakdown-EDAP}
Let us fix $\alpha\in (0,1)$. Under conditions of Lemma \ref{Dposterior:lem:1-breakdown-EDAP}, if there exists $\delta>0$ such that
\begin{equation*}
\inf_{L\in\N_0} \inf_{\theta\in\Theta} D(\alpha\bar{q}_{L},\theta) > \inf_{\theta\in\Theta} D((1-\alpha)p,\theta)+\delta,
\end{equation*}
then, there exists $L_0\in\N_0$ such that for every $\lambda\in(0,\delta)$, there is $M^*>0$ satisfying
\begin{equation*}\label{Dposterior:eq:lem2-breakdown-EDAP}
\inf_{L>L_0} \inf_{|\theta|>M^*} \left(D((1-\alpha)p+\alpha\bar{q}_L,\theta)-D\left((1-\alpha)p+\alpha\bar{q}_L,T\left((1-\alpha)p+\alpha\bar{q}_L\right)\right)\right)\geq\lambda.
\end{equation*}
\end{lem}

\noindent The proof is similar to that for Lemma 2 in \cite{Hooker-Vidyshankar-2014} and it is omitted.

\vspace{0.5cm}

We shall prove Theorem \ref{Dposterior:thm:asympt-breakdown}.

\ref{Dposterior:thm:asympt-breakdown-EDAP} To that end, let denote
\begin{align*}
B_1&=\{\alpha\in (0,1):b(\alpha,T,p)<\infty\},\\
B_2&=\{\alpha\in (0,1):\limsup_{n\to\infty} b(\alpha,\overline{T}_n,p)<\infty\}.
\end{align*}
We shall prove that $B_1=B_2$ and, as a result, we obtain $B(T,p)=B(\{\overline{T}_n\}_{n\in\N},p)$.

\vspace{1ex}

Let assume that $\alpha\in B_1^c$; then, we can find a sequence of probability distributions $\{\bar{q}_{L}\}_{L\in\N}$ such that $|T((1-\alpha)p+\alpha\bar{q}_{L})|\to\infty$, as $L\to\infty$.  From Theorem \ref{Dposterior:thm:aprox-MDE}~\ref{Dposterior:thm:aprox-MDE-i-EDAP}, one deduces that for each $L\in\N$,
\begin{equation*}
\forall \epsilon>0,\ \exists k_L=k_L(\epsilon)\in\N:\  |\overline{T}_n((1-\alpha)p+\alpha\bar{q}_{L})-T((1-\alpha)p+\alpha\bar{q}_{L})|<\epsilon,\quad \forall n\geq k_L,
\end{equation*}
and with $k_L\to\infty$, as $L\to\infty$. Let $M>0$, and let fix $\epsilon=1$, and $k_L=k_L(1)$, $L\in\N$. From the convergence $|T((1-\alpha)p+\alpha\bar{q}_{L})|\to\infty$, as $L\to\infty$, one has
\begin{equation*}
\exists L_0\in\N:\  |T((1-\alpha)p+\alpha\bar{q}_{L})|>M+1,\quad \forall L\geq L_0.
\end{equation*}
Now, by using the triangle inequality, for $L\geq L_0$ and $n\geq k_L$,
\begin{align*}
|\overline{T}_{n}((1-\alpha)p+\alpha\bar{q}_{L})|&\geq |T((1-\alpha)p+\alpha\bar{q}_{L})|\\
&\phantom{\geq}-|\overline{T}_{n}((1-\alpha)p+\alpha\bar{q}_{L})-T((1-\alpha)p+\alpha\bar{q}_{L})|>M.
\end{align*}
From all the above, and the fact that  $\overline{T}_n(p)\to \theta_0$, as $n\to\infty$, one has that $\limsup_{n\to\infty} b(\alpha,\overline{T}_n,p)=\infty$, and $\alpha\in B_2^c$.

\vspace{1ex}

Now, we shall prove that if $\alpha\in B_1$, then $\alpha\in B_2$. To that end, we prove that for each $n\in\N$, and any sequence of contaminating distributions $\{\bar{q}_L\}_{L\in\N}$ in $\overline{\Gamma}$, there exists $L_0\in\N$ such that $|\overline{T}_n((1-\alpha)p+\alpha\bar{q}_{L})|\leq M_n$, if $L\geq L_0$, where $\overline{\Gamma}$ denotes the family of contaminating distributions satisfying \ref{Dposterior:cond:contam-distr}~\ref{Dposterior:cond:ortog-distrib-contam}-\ref{Dposterior:cond:ortog-contam-family}, and $M_n\to M$, for some $M>0$.

On the one hand, since $\alpha\in B_1$, $\sup_{\bar{q}\in\overline{\Gamma}} |T((1-\alpha)p+\alpha\bar{q})|<\infty$.

Let fix $\lambda\in (0,\delta)$. From Lemma \ref{Dposterior:lem:2-breakdown-EDAP}, one can find $L_0\in\N_0$ and $M^*>0$ satisfying
\begin{equation*}\label{Dposterior:eq:asympt-break-point-EDAP}
D((1-\alpha)p+\alpha\bar{q}_{L},\theta)\geq D\left((1-\alpha)p+\alpha\bar{q}_{L},T\left((1-\alpha)p+\alpha\bar{q}_{L}\right)\right)+\lambda,
\end{equation*}
for each $L\geq L_0$ and each $|\theta|\geq M^*$, consequently,
$$e^{-\Delta_{n-1}D((1-\alpha)p+\alpha\bar{q}_{L},\theta)}\leq e^{-\Delta_{n-1}D\left((1-\alpha)p+\alpha\bar{q}_{L},T\left((1-\alpha)p+\alpha\bar{q}_{L}\right)\right)-\Delta_{n-1}\lambda},$$
for each $L\geq L_0$ and each $|\theta|\geq M^*$.

Moreover, since $\alpha\in B_1$, one can deduced, using the mean value theorem, condition (A4) in \cite{art-MDE} and \ref{Dposterior:cond:bound-RAF} (or alternatively condition (A6) in \cite{art-MDE} for the Hellinger distance), that for the fixed $\lambda$, there exists $\epsilon>0$ such that if $|\theta-T((1-\alpha)p+\alpha\bar{q}_{L})|<\epsilon$, then
$$D((1-\alpha)p+\alpha\bar{q}_{L},\theta)<D((1-\alpha)p+\alpha\bar{q}_{L},T((1-\alpha)p+\alpha\bar{q}_{L}))+\lambda/2,$$
for any $L\in N_0$, and hence
$$e^{-\Delta_{n-1}D((1-\alpha)p+\alpha\bar{q}_{L},\theta)}>e^{-\Delta_{n-1}D((1-\alpha)p+\alpha\bar{q}_{L},T((1-\alpha)p+\alpha\bar{q}_{L}))-\Delta_{n-1}\lambda/2}.$$

Let denote $B=\{\theta\in\Theta:|\theta|<M^*\}$ and $K(x)=\left(\int_{x-\epsilon}^{x+\epsilon} \pi(\theta)d\theta\right)^{-1}\int_{M^*}^{\infty} |\theta|\pi(\theta)d\theta$. Then, for $L\geq L_0$,
\begin{align*}
|\overline{T}_n&((1-\alpha)p+\alpha\bar{q}_{L})|
=\int_{B} |\theta| \pi_D^n(\theta|(1-\alpha)p+\alpha\bar{q}_{L})d\theta\\
&\phantom{=} +\int_{B^c} |\theta| \pi_D^n(\theta|(1-\alpha)p+\alpha\bar{q}_{L})d\theta\\
&\leq  M^* +\frac{\int_{B^c} |\theta| e^{-\Delta_{n-1}D((1-\alpha)p+\alpha\bar{q}_{L},\theta)}\pi(\theta)d\theta}{\int_{T((1-\alpha)p+\alpha\bar{q}_{L})-\epsilon}^{T((1-\alpha)p+\alpha\bar{q}_{L})+\epsilon} e^{-\Delta_{n-1}D((1-\alpha)p+\alpha\bar{q}_{L},\theta)}\pi(\theta)d\theta}\\
&\leq  M^* +\frac{e^{-\Delta_{n-1}D((1-\alpha)p+\alpha\bar{q}_{L},T((1-\alpha)p+\alpha\bar{q}_{L}))-\Delta_{n-1}\lambda} \int_{B^c} |\theta| \pi(\theta)d\theta}{e^{-\Delta_{n-1}D((1-\alpha)p+\alpha\bar{q}_{L},T((1-\alpha)p+\alpha\bar{q}_{L}))-\Delta_{n-1}\lambda/2}\int_{T((1-\alpha)p+\alpha\bar{q}_{L})-\epsilon}^{T((1-\alpha)p+\alpha\bar{q}_{L})+\epsilon} \pi(\theta)d\theta}\\
&\leq  M^*+e^{-\Delta_{n-1}\lambda/2}K^*,
\end{align*}
where $K^*=\sup\{K(t):|t|\leq \sup_{L^*>L_0} |T((1-\alpha)p+\alpha\bar{q}_{L^*})|\}$. Note that $K^*<\infty$ due to the fact that $\alpha\in B_1$. Thus, $M^*+e^{-\Delta_{n-1}\lambda/2}K^*\to M^*$, as $n\to\infty$.

\ref{Dposterior:thm:asympt-breakdown-MDAP} Let write $\theta^*=\arg\max_{\theta\in\Theta}\pi(\theta)$, $\overline{g}_n(q,\theta)=\Delta_{n-1}D(q,\theta)-\log (\pi(\theta^*))$, and
\begin{align*}
\widetilde{B}_2=\{\alpha\in (0,1):\limsup_{n\to\infty} b(\alpha,\widetilde{T}_n,p)<\infty\},
\end{align*}
and, as done before, we shall prove that $B_1=\widetilde{B}_2$ to obtain $B(T,p)=B(\{\widetilde{T}_n\}_{n\in\N},p)$.

Assume that $\alpha\in B_1$, and observe that for each $q\in\Gamma$, $g_n(q,\theta)\geq\overline{g}_n(q,\theta)$, for each $\theta\in\Theta$, and $g_n(q,\theta^*)=\overline{g}_n(q,\theta^*)$,
hence, for $n\in\N$ sufficiently large $|\widetilde{T}_n(q)-\theta^*|\leq|T(q)-\theta^*|$. As a consequence,
\begin{eqnarray*}
|\widetilde{T}_n((1-\alpha)p+\alpha\bar{q})-\widetilde{T}_n(p)|&\leq&|\widetilde{T}_n((1-\alpha)p+\alpha\bar{q})-\theta^*|+|\theta^*-\widetilde{T}_n(p)|\\
&\leq &|T((1-\alpha)p+\alpha\bar{q})-T(p)|+|T(p)-\theta^*|\\
&\phantom{\leq}&+|\theta^*-T(p)|,
\end{eqnarray*}
and $b(\alpha,\widetilde{T}_n,p)\leq b(\alpha,T,p)+2|T(p)-\theta^*|$; thus, $\alpha\in\widetilde{B}_2$.

The fact that if $\alpha\in B_1^c$, then $\alpha\in \widetilde{B}_2^c$, can be proved in an identical way to that in \ref{Dposterior:thm:asympt-breakdown-EDAP}.
\end{Prf}


\begin{Prf}[Theorem \ref{Dposterior:thm:convergence-Dposterior}]
\ref{Dposterior:thm:convergence-Dposterior-i} is immediate using the same arguments as in the proof of Theorem \ref{Dposterior:thm:breakdown-point-EDAP-MDAP}~\ref{Dposterior:thm:breakdown-point-ii-MDAP}.

\ref{Dposterior:thm:convergence-Dposterior-ii} Let denote
\begin{align*}
J&=\int_{\Theta} e^{-\Delta_{n-1}D((1-\alpha)p,\theta)}\pi(\theta)d\theta,\\
J_L&=\int_\Theta e^{-\Delta_{n-1}D((1-\alpha)p+\alpha \bar{q}_{L},\theta)}\pi(\theta)d\theta,\quad \text{  for each }L\in\N_0.
\end{align*}
With this notation, for each $\theta\in\Theta$, one has
\begin{align*}
\big|\pi_D^n(\theta|(1-\alpha)p+\alpha \bar{q}_{L})&-\pi_D^n(\theta|(1-\alpha)p)\big|
=\frac{\pi(\theta)}{J\cdot J_L}\\
&\phantom{\leq}\cdot\left|J e^{-\Delta_{n-1}D((1-\alpha)p+\alpha \bar{q}_{L},\theta)}- J_L e^{-\Delta_{n-1}D((1-\alpha)p,\theta)}\right|\\
&\leq \frac{1}{J_L}  \left|e^{-\Delta_{n-1}D((1-\alpha)p+\alpha \bar{q}_{L},\theta)}- e^{-\Delta_{n-1}D((1-\alpha)p,\theta)}\right|\pi(\theta)\\
&\phantom{\leq} +\frac{|J-J_L|}{J\cdot J_L} e^{-\Delta_{n-1}D((1-\alpha)p,\theta)}\pi(\theta).
\end{align*}
Observe that $J_L\to J$, as $L\to\infty$, and using Jensen's inequality, one has that\linebreak $D((1-\alpha)p,\theta)\geq G(-\alpha)$, and consequently $\int_\Theta e^{-\Delta_{n-1}D((1-\alpha)p,\theta)}\pi(\theta) d\theta\leq e^{-\Delta_{n-1} G(-\alpha)}$; hence, it suffices to prove that $I_L\to 0$, as $L\to\infty$, where
\begin{equation*}
I_L=\int_\Theta \left|e^{-\Delta_{n-1}D((1-\alpha)p+\alpha \bar{q}_{L},\theta)}- e^{-\Delta_{n-1}D((1-\alpha)p,\theta)}\right| \pi(\theta) d\theta.
\end{equation*}

Let consider the function $f: x\in [c,\infty)\rightarrow f(x)=e^{-\Delta_{n-1}x}$, with \linebreak$c=\min(0,G(-\alpha))$, and let fix $\epsilon>0$. Since $f(\cdot)$ is uniformly continuous, there exists $\delta=\delta(\epsilon)$ such that if $|x-y|<\delta$, then $|e^{-\Delta_{n-1}x}-e^{-\Delta_{n-1}y}|<\epsilon$. Given $M>0$, from Lemma \ref{Dposterior:lem:1-breakdown-EDAP}~\ref{Dposterior:lem:1-breakdown-EDAP-ii}, one has that there exists $L_0=L_0(M,\delta)$ such that
\begin{equation*}
|D((1-\alpha)p+\alpha \bar{q}_{L},\theta)-D((1-\alpha)p,\theta)|<\delta,\quad \forall\theta\in\Theta : |\theta|<M, \forall L>L_0.
\end{equation*}
Let write
\begin{align*}
I_{L,1}^{(M)}&=\int_{\{|\theta|<M\}} \left|e^{-\Delta_{n-1}D((1-\alpha)p+\alpha \bar{q}_{L},\theta)}- e^{-\Delta_{n-1}D((1-\alpha)p,\theta)}\right| \pi(\theta) d\theta,\\
I_{L,2}^{(M)}&=\int_{\{|\theta|\geq M\}} \left|e^{-\Delta_{n-1}D((1-\alpha)p+\alpha \bar{q}_{L},\theta)}- e^{-\Delta_{n-1}D((1-\alpha)p,\theta)}\right| \pi(\theta) d\theta,
\end{align*}
thus, $I_L=I_{L,1}^{(M)}+I_{L,2}^{(M)}$. On the one hand, $I_{L,1}^{(M)}\leq \epsilon\int_{\{|\theta|<M\}} \pi(\theta)d\theta\leq \epsilon$, for each $L>L_0$. On the other hand, $I_{L,2}^{(M)}\leq \left(1+e^{-\Delta_{n-1} G(-\alpha)}\right)\int_{\{|\theta|\geq M\}} \pi(\theta)d\theta$, for each $L\in\N_0$. Combining all the above, one has that for any $M>0$, and $\epsilon>0$,
\begin{equation*}
\limsup_{L\to\infty} I_L\leq \epsilon+\left(1+e^{-\Delta_{n-1} G(-\alpha)}\right)\int_{\{|\theta|\geq M\}} \pi(\theta)d\theta.
\end{equation*}
Thus, since $\lim_{M\to\infty} \int_{\{|\theta|\geq M\}} \pi(\theta)d\theta=0$, one obtains $\limsup_{L\to\infty} I_L\leq \epsilon$ and taking limit as $\epsilon\to 0$, one has that $\lim_{L\to\infty} I_L=0$.
\end{Prf}

\subsection{Additional details of the example in Subsection \ref{Dposterior:ex:oligo}}

\subsubsection{Sensitivity analysis}

A brief description of the sensitivity analysis carried out for the example on oligodendrocyte cell populations is presented in this subsection in Tables \ref{Dposterior:tab:real-ex-sens-edap-1}, \ref{Dposterior:tab:real-ex-sens-mdap-1}, \ref{Dposterior:tab:real-ex-sens-edap-2} and \ref{Dposterior:tab:real-ex-sens-mdap-2}.

\subsection{Additional details of the simulated example in Subsection \ref{Dposterior:ex:sim}}\label{sup:sec:ex-sim}

\subsubsection{Simulated data}

The data for the simulated example are provided in Tables \ref{Dposterior:table:sim-data-1} and \ref{Dposterior:table:sim-data-2}. Recall that in this example, the initial number of individuals is $Z_0=1$, the reproduction law is a geometric distribution with parameter $\theta_0=0.3$ and which is contaminated by outliers, which can happen at the point 11 with probability 0.15. The control variables $\phi_n(k)$ has Poisson distributions with mean $\lambda k$, for each $k,n\in\N_0$.

\subsubsection{Sensitivity analysis}\label{sup:subsec:ex-sim-sensi}

A summary of the sensitivity analysis performed for the simulated example is provided in this subsection. For sake of brevity, we show the results for the generations 25 (in Tables \ref{Dposterior:table-sens25-I}, \ref{Dposterior:table-sens25-II}, \ref{Dposterior:table-sens25-III} and \ref{Dposterior:table-sens25-IV}), 35  (in Tables \ref{Dposterior:table-sens35-I}, \ref{Dposterior:table-sens35-II}, \ref{Dposterior:table-sens35-III} and \ref{Dposterior:table-sens35-IV}) and 45  (in Tables \ref{Dposterior:table-sens45-I}, \ref{Dposterior:table-sens45-II}, \ref{Dposterior:table-sens45-III} and \ref{Dposterior:table-sens45-IV}).

\newpage
\renewcommand{\arraystretch}{1}

\begin{table}[H]\centering
\vspace*{2cm}
\rotatebox{90}{\centering\scalebox{0.8}{
\begin{tabular}{|ccc|ccc|ccc|ccc|}
\cline{4-12}
\multicolumn{3}{c|}{}&\multicolumn{3}{|c|}{HD}&\multicolumn{3}{|c|}{NED}&\multicolumn{3}{|c|}{KL}\\
\hline
$\alpha_1$ & $\alpha_2$  & $\alpha_3$ & $p_0$ & $p_1$ & $p_2$ & $p_0$ & $p_1$ & $p_2$ & $p_0$ & $p_1$ & $p_2$ \\
\hline
$1$ & $1$ & $1$ & $0.3852$ & $0.4892$ & $0.1257$ & $0.3853$ & $0.4889$ & $0.1259$ & $0.3850$ & $0.4891$ & $0.1259$ \\
$0.5$ & $0.5$ & $0.5$ & $0.3852$ & $0.4900$ & $0.1248$ & $0.3848$ & $0.4901$ & $0.1251$ & $0.3854$ & $0.4893$ & $0.1253$ \\
$1.9268$ & $2.4512$ & $0.6220$ & $0.3852$ & $0.4907$ & $0.1241$ & $0.3852$ & $0.4903$ & $0.1244$ & $0.3853$ & $0.4905$ & $0.1242$ \\
$1$ & $1$ & $8$ & $0.4007$ & $0.4563$ & $0.1430$ & $0.3827$ & $0.4625$ & $0.1548$ & $0.3756$ & $0.4713$ & $0.1531$ \\
$1$ & $2$ & $7$ & $0.3746$ & $0.4843$ & $0.1411$ & $0.3818$ & $0.4875$ & $0.1307$ & $0.3714$ & $0.4885$ & $0.1401$ \\
$1$ & $3$ & $6$ & $0.3762$ & $0.4888$ & $0.1349$ & $0.3786$ & $0.4835$ & $0.1379$ & $0.3764$ & $0.4849$ & $0.1387$ \\
\hline
\end{tabular}}}
\caption{Sensitivity analysis for EDAP estimators in experiment 1.}\label{Dposterior:tab:real-ex-sens-edap-1}
\end{table}

\begin{table}[H]\centering
\vspace*{2cm}
\rotatebox{90}{\centering\scalebox{0.8}{\begin{tabular}{|ccc|ccc|ccc|ccc|}
\cline{4-12}
\multicolumn{3}{c|}{}&\multicolumn{3}{|c|}{HD}&\multicolumn{3}{|c|}{NED}&\multicolumn{3}{|c|}{KL}\\
\hline
$\alpha_1$ & $\alpha_2$  & $\alpha_3$ & $p_0$ & $p_1$ & $p_2$ & $p_0$ & $p_1$ & $p_2$ & $p_0$ & $p_1$ & $p_2$ \\
\hline
$1$ & $1$ & $1$ & $0.3860$ & $0.4900$ & $0.1240$ & $0.3860$ & $0.4900$ & $0.1240$ & $0.3860$ & $0.4900$ & $0.1240$ \\
$0.5$ & $0.5$ & $0.5$ & $0.3850$ & $0.4910$ & $0.1240$ & $0.3850$ & $0.4910$ & $0.1240$ & $0.3850$ & $0.4910$ & $0.1240$ \\
$1.9268$ & $2.4512$ & $0.6220$ & $0.3860$ & $0.4910$ & $0.1230$ & $0.3860$ & $0.4910$ & $0.1230$ & $0.3860$ & $0.4910$ & $0.1230$ \\
$1$ & $1$ & $8$ & $0.3790$ & $0.4820$ & $0.1390$ & $0.3790$ & $0.4820$ & $0.1390$ & $0.3790$ & $0.4820$ & $0.1390$ \\
$1$ & $2$ & $7$ & $0.3790$ & $0.4850$ & $0.1360$ & $0.3790$ & $0.4840$ & $0.1370$ & $0.3790$ & $0.4840$ & $0.1370$ \\
$1$ & $3$ & $6$ & $0.3790$ & $0.4870$ & $0.1340$ & $0.3790$ & $0.4870$ & $0.1340$ & $0.3790$ & $0.4870$ & $0.1340$ \\
\hline
\end{tabular}}}
\caption{Sensitivity analysis for MDAP estimator in experiment 1.}\label{Dposterior:tab:real-ex-sens-mdap-1}
\end{table}

\begin{table}[H]\centering
\vspace*{2cm}
\rotatebox{90}{\centering\scalebox{0.8}{\begin{tabular}{|ccc|ccc|ccc|ccc|}
\cline{4-12}
\multicolumn{3}{c|}{}&\multicolumn{3}{|c|}{HD}&\multicolumn{3}{|c|}{NED}&\multicolumn{3}{|c|}{KL}\\
\hline
$\alpha_1$ & $\alpha_2$  & $\alpha_3$ & $p_0$ & $p_1$ & $p_2$ & $p_0$ & $p_1$ & $p_2$ & $p_0$ & $p_1$ & $p_2$ \\
\hline
$1$ & $1$ & $1$ & $0.1389$ & $0.4933$ & $0.3678$ & $0.1397$ & $0.4925$ & $0.3679$ & $0.1397$ & $0.4926$ & $0.3677$ \\
$0.5$ & $0.5$ & $0.5$ & $0.1384$ & $0.4934$ & $0.3682$ & $0.1387$ & $0.4932$ & $0.3681$ & $0.1388$ & $0.4934$ & $0.3678$ \\
$0.6877$ & $2.4721$ & $1.8401$ & $0.1371$ & $0.4947$ & $0.3682$ & $0.1374$ & $0.4945$ & $0.3681$ & $0.1376$ & $0.4942$ & $0.3683$ \\
$1$ & $1$ & $8$ & $0.1348$ & $0.4820$ & $0.3832$ & $0.1368$ & $0.4798$ & $0.3834$ & $0.1368$ & $0.4813$ & $0.3819$ \\
$1$ & $2$ & $7$ & $0.1357$ & $0.4845$ & $0.3798$ & $0.1361$ & $0.4839$ & $0.3800$ & $0.1362$ & $0.4838$ & $0.3800$ \\
$1$ & $3$ & $6$ & $0.1356$ & $0.4879$ & $0.3765$ & $0.1361$ & $0.4877$ & $0.3762$ & $0.1362$ & $0.4875$ & $0.3763$ \\
\hline
\end{tabular}}}
\caption{Sensitivity analysis for EDAP estimator in experiment 2.}\label{Dposterior:tab:real-ex-sens-edap-2}
\end{table}

\begin{table}[H]\centering
\vspace*{2cm}
\rotatebox{90}{\centering\scalebox{0.8}{\begin{tabular}{|ccc|ccc|ccc|ccc|}
\cline{4-12}
\multicolumn{3}{c|}{}&\multicolumn{3}{|c|}{HD}&\multicolumn{3}{|c|}{NED}&\multicolumn{3}{|c|}{KL}\\
\hline
$\alpha_1$ & $\alpha_2$  & $\alpha_3$ & $p_0$ & $p_1$ & $p_2$ & $p_0$ & $p_1$ & $p_2$ & $p_0$ & $p_1$ & $p_2$ \\
\hline
$1$ & $1$ & $1$ & $0.1380$ & $0.4940$ & $0.3680$ & $0.1380$ & $0.4940$ & $0.3680$ & $0.1380$ & $0.4940$ & $0.3680$ \\
$0.5$ & $0.5$ & $0.5$ & $0.1360$ & $0.4960$ & $0.3680$ & $0.1360$ & $0.4960$ & $0.3680$ & $0.1360$ & $0.4960$ & $0.3680$ \\
$0.6877$ & $2.4721$ & $1.8401$ & $0.1350$ & $0.4960$ & $0.3690$ & $0.1350$ & $0.4960$ & $0.3690$ & $0.1350$ & $0.4960$ & $0.3690$ \\
$1$ & $1$ & $8$ & $0.1340$ & $0.4820$ & $0.3840$ & $0.1340$ & $0.4820$ & $0.3840$ & $0.1340$ & $0.4820$ & $0.3840$ \\
$1$ & $2$ & $7$ & $0.1340$ & $0.4860$ & $0.3800$ & $0.1340$ & $0.4860$ & $0.3800$ & $0.1340$ & $0.4860$ & $0.3800$ \\
$1$ & $3$ & $6$ & $0.1340$ & $0.4890$ & $0.3770$ & $0.1340$ & $0.4890$ & $0.3770$ & $0.1340$ & $0.4890$ & $0.3770$ \\
\hline
\end{tabular}}}
\caption{Sensitivity analysis for MDAP estimator in experiment 2.}\label{Dposterior:tab:real-ex-sens-mdap-2}
\end{table}

\begin{table}[H]\centering
\rotatebox{90}{\centering\scalebox{0.85}{\begin{tabular}{|c|c|c|c|c|c|c|c|c|c|c|c|c|c|c|c|c|c|}
\hline
$n$ & $Z_n$ & $\phi_n(Z_n)$ & $Z_n(0)$ & $Z_n(1)$ & $Z_n(2)$ & $Z_n(3)$ & $Z_n(4)$ & $Z_n(5)$ & $Z_n(6)$ & $Z_n(7)$  & $Z_n(8)$ & $Z_n(9)$ & $Z_n(10)$ & $Z_n(11)$ & $Z_n(12)$ & $Z_n(13)$ & $Z_n(14)$ \\
\hline
$0$ & $1$ & $1$ & $0$ & $0$ & $0$ & $0$ & $0$ & $0$ & $0$ & $0$ & $1$ & $0$ & $0$ & $0$ & $0$ & $0$ & $0$ \\
$1$ & $8$ & $4$ & $1$ & $0$ & $0$ & $1$ & $0$ & $0$ & $0$ & $0$ & $1$ & $0$ & $0$ & $1$ & $0$ & $0$ & $0$ \\
$2$ & $22$ & $5$ & $2$ & $1$ & $0$ & $0$ & $1$ & $0$ & $0$ & $0$ & $0$ & $0$ & $0$ & $1$ & $0$ & $0$ & $0$ \\
$3$ & $16$ & $3$ & $0$ & $0$ & $0$ & $0$ & $0$ & $0$ & $1$ & $0$ & $0$ & $0$ & $0$ & $2$ & $0$ & $0$ & $0$ \\
$4$ & $28$ & $3$ & $2$ & $0$ & $0$ & $0$ & $1$ & $0$ & $0$ & $0$ & $0$ & $0$ & $0$ & $0$ & $0$ & $0$ & $0$ \\
$5$ & $4$ & $1$ & $0$ & $0$ & $0$ & $0$ & $0$ & $0$ & $0$ & $0$ & $0$ & $0$ & $0$ & $1$ & $0$ & $0$ & $0$ \\
$6$ & $11$ & $2$ & $0$ & $1$ & $0$ & $1$ & $0$ & $0$ & $0$ & $0$ & $0$ & $0$ & $0$ & $0$ & $0$ & $0$ & $0$ \\
\hline
\end{tabular}}}
\caption{Simulated data from the initial generation to the 22-th generation.}\label{Dposterior:table:sim-data-1}
\end{table}

\begin{table}[H]\centering
\rotatebox{90}{\centering\scalebox{0.85}{\begin{tabular}{|c|c|c|c|c|c|c|c|c|c|c|c|c|c|c|c|c|c|}
\hline
$n$ & $Z_n$ & $\phi_n(Z_n)$ & $Z_n(0)$ & $Z_n(1)$ & $Z_n(2)$ & $Z_n(3)$ & $Z_n(4)$ & $Z_n(5)$ & $Z_n(6)$ & $Z_n(7)$  & $Z_n(8)$ & $Z_n(9)$ & $Z_n(10)$ & $Z_n(11)$ & $Z_n(12)$ & $Z_n(13)$ & $Z_n(14)$ \\
\hline
$23$ & $13$ & $4$ & $2$ & $0$ & $1$ & $0$ & $0$ & $0$ & $0$ & $0$ & $0$ & $0$ & $0$ & $1$ & $0$ & $0$ & $0$ \\
$24$ & $13$ & $2$ & $0$ & $1$ & $0$ & $0$ & $0$ & $0$ & $0$ & $0$ & $0$ & $0$ & $0$ & $1$ & $0$ & $0$ & $0$ \\
$25$ & $12$ & $3$ & $1$ & $0$ & $1$ & $0$ & $0$ & $0$ & $0$ & $0$ & $0$ & $1$ & $0$ & $0$ & $0$ & $0$ & $0$ \\
$26$ & $11$ & $4$ & $0$ & $0$ & $0$ & $1$ & $1$ & $0$ & $0$ & $0$ & $1$ & $0$ & $0$ & $1$ & $0$ & $0$ & $0$ \\
$27$ & $26$ & $9$ & $2$ & $2$ & $2$ & $0$ & $1$ & $0$ & $0$ & $0$ & $0$ & $0$ & $1$ & $1$ & $0$ & $0$ & $0$ \\
$28$ & $31$ & $7$ & $2$ & $1$ & $1$ & $0$ & $0$ & $0$ & $0$ & $0$ & $0$ & $0$ & $0$ & $3$ & $0$ & $0$ & $0$ \\
$29$ & $36$ & $9$ & $3$ & $2$ & $0$ & $1$ & $2$ & $0$ & $1$ & $0$ & $0$ & $0$ & $0$ & $0$ & $0$ & $0$ & $0$ \\
\hline
\end{tabular}}}
\caption{Simulated data from the 23-th generation to the 45-th generation.}\label{Dposterior:table:sim-data-2}
\end{table}

\begin{table}[H]\centering
\scalebox{1}{\begin{tabular}{|c|c|c|c|c|c|c|c|}
\hline
$\rho$ & $\beta$ & Prior mean & Prior variance & $\theta_{25}^{*HD}$ & $\theta_{25}^{+HD}$ & $\theta_{25}^{*NED}$ & $\theta_{25}^{+NED}$ \\
\hline
$0.1$ & $0.1$ & $0.5$ & $0.2083$ & $0.3056$ & $0.3027$ & $0.3056$ & $0.3027$ \\
$0.1$ & $0.5$ & $0.1667$ & $0.0868$ & $0.3052$ & $0.3023$ & $0.3052$ & $0.3023$ \\
$0.1$ & $1$ & $0.0909$ & $0.0394$ & $0.3047$ & $0.3017$ & $0.3047$ & $0.3017$ \\
$0.1$ & $1.5$ & $0.0625$ & $0.0225$ & $0.3041$ & $0.3012$ & $0.3041$ & $0.3012$ \\
$0.1$ & $2$ & $0.0476$ & $0.0146$ & $0.3036$ & $0.3007$ & $0.3036$ & $0.3007$ \\
$0.1$ & $2.5$ & $0.0385$ & $0.0103$ & $0.3031$ & $0.3002$ & $0.3031$ & $0.3002$ \\
$0.1$ & $3$ & $0.0323$ & $0.0076$ & $0.3026$ & $0.2997$ & $0.3026$ & $0.2997$ \\
\hline
\end{tabular}}
\caption{Sensitivity analysis for generation $n=25$.}\label{Dposterior:table-sens25-I}
\end{table}

\begin{table}[H]\centering
\scalebox{1}{\begin{tabular}{|c|c|c|c|c|c|c|c|}
\hline
$\rho$ & $\beta$ & Prior mean & Prior variance & $\theta_{25}^{*HD}$ & $\theta_{25}^{+HD}$ & $\theta_{25}^{*NED}$ & $\theta_{25}^{+NED}$ \\
\hline
$1.5$ & $0.1$ & $0.9375$ & $0.0225$ & $0.3089$ & $0.3060$ & $0.3089$ & $0.3060$ \\
$1.5$ & $0.5$ & $0.7500$ & $0.0625$ & $0.3085$ & $0.3056$ & $0.3085$ & $0.3056$ \\
$1.5$ & $1$ & $0.6000$ & $0.0686$ & $0.3080$ & $0.3051$ & $0.3080$ & $0.3051$ \\
$1.5$ & $1.5$ & $0.5$ & $0.0625$ & $0.3074$ & $0.3045$ & $0.3074$ & $0.3045$ \\
$1.5$ & $2$ & $0.4286$ & $0.0544$ & $0.3069$ & $0.3040$ & $0.3069$ & $0.3040$ \\
$1.5$ & $2.5$ & $0.3750$ & $0.0469$ & $0.3064$ & $0.3035$ & $0.3064$ & $0.3035$ \\
$1.5$ & $3$ & $0.3333$ & $0.0404$ & $0.3059$ & $0.3030$ & $0.3059$ & $0.3030$ \\
\hline
\end{tabular}}
\caption{Sensitivity analysis for generation $n=25$ (continuation).}\label{Dposterior:table-sens25-II}
\end{table}

\begin{table}[H]\centering
\scalebox{1}{\begin{tabular}{|c|c|c|c|c|c|c|c|}
\hline
$\rho$ & $\beta$ & Prior mean & Prior variance & $\theta_{25}^{*HD}$ & $\theta_{25}^{+HD}$ & $\theta_{25}^{*NED}$ & $\theta_{25}^{+NED}$ \\
\hline
$3$ & $0.1$ & $0.9677$ & $0.0076$ & $0.3126$ & $0.3096$ & $0.3126$ & $0.3096$ \\
$3$ & $0.5$ & $0.8571$ & $0.0272$ & $0.3121$ & $0.3091$ & $0.3121$ & $0.3091$ \\
$3$ & $1$ & $0.7500$ & $0.0375$ & $0.3116$ & $0.3086$ & $0.3116$ & $0.3086$ \\
$3$ & $1.5$ & $0.6667$ & $0.0404$ & $0.3110$ & $0.3081$ & $0.3110$ & $0.3081$ \\
$3$ & $2$ & $0.6000$ & $0.0400$ & $0.3105$ & $0.3075$ & $0.3105$ & $0.3075$ \\
$3$ & $2.5$ & $0.5455$ & $0.0381$ & $0.3099$ & $0.3070$ & $0.3099$ & $0.3070$ \\
$3$ & $3$ & $0.5$ & $0.0357$ & $0.3094$ & $0.3065$ & $0.3094$ & $0.3065$ \\
\hline
\end{tabular}}
\caption{Sensitivity analysis for generation $n=25$ (continuation).}\label{Dposterior:table-sens25-III}
\end{table}

\begin{table}[H]\centering
\scalebox{1}{\begin{tabular}{|c|c|c|c|c|c|c|c|}
\hline
$\rho$ & $\beta$ & Prior mean & Prior variance & $\theta_{25}^{*HD}$ & $\theta_{25}^{+HD}$ & $\theta_{25}^{*NED}$ & $\theta_{25}^{+NED}$ \\
\hline
$4.5$ & $0.1$ & $0.9783$ & $0.0038$ & $0.3162$ & $0.3132$ & $0.3162$ & $0.3132$ \\
$4.5$ & $0.5$ & $0.9000$ & $0.0150$ & $0.3157$ & $0.3127$ & $0.3157$ & $0.3127$ \\
$4.5$ & $1$ & $0.8182$ & $0.0229$ & $0.3152$ & $0.3122$ & $0.3152$ & $0.3122$ \\
$4.5$ & $1.5$ & $0.7500$ & $0.0268$ & $0.3146$ & $0.3116$ & $0.3146$ & $0.3116$ \\
$4.5$ & $2$ & $0.6923$ & $0.0284$ & $0.3141$ & $0.3111$ & $0.3141$ & $0.3111$ \\
$4.5$ & $2.5$ & $0.6429$ & $0.0287$ & $0.3135$ & $0.3106$ & $0.3135$ & $0.3106$ \\
$4.5$ & $3$ & $0.6000$ & $0.0282$ & $0.3130$ & $0.3100$ & $0.3130$ & $0.3100$ \\
$4.5$ & $3.5$ & $0.5625$ & $0.0273$ & $0.3124$ & $0.3095$ & $0.3124$ & $0.3095$ \\
$4.5$ & $4$ & $0.5294$ & $0.0262$ & $0.3119$ & $0.3090$ & $0.3119$ & $0.3090$ \\
$4.5$ & $4.5$ & $0.5$ & $0.0250$ & $0.3113$ & $0.3085$ & $0.3113$ & $0.3085$ \\
$4.5$ & $5$ & $0.4737$ & $0.0237$ & $0.3108$ & $0.3079$ & $0.3108$ & $0.3079$ \\
$5$ & $0.1$ & $0.9804$ & $0.0032$ & $0.3174$ & $0.3144$ & $0.3174$ & $0.3144$ \\
$5$ & $0.5$ & $0.9091$ & $0.0127$ & $0.3169$ & $0.3139$ & $0.3169$ & $0.3139$ \\
$5$ & $1$ & $0.8333$ & $0.0198$ & $0.3164$ & $0.3134$ & $0.3164$ & $0.3134$ \\
$5$ & $1.5$ & $0.7692$ & $0.0237$ & $0.3158$ & $0.3128$ & $0.3158$ & $0.3128$ \\
$5$ & $2$ & $0.7143$ & $0.0255$ & $0.3153$ & $0.3123$ & $0.3153$ & $0.3123$ \\
$5$ & $2.5$ & $0.6667$ & $0.0261$ & $0.3147$ & $0.3118$ & $0.3147$ & $0.3118$ \\
$5$ & $3$ & $0.6250$ & $0.0260$ & $0.3141$ & $0.3112$ & $0.3141$ & $0.3112$ \\
$5$ & $3.5$ & $0.5882$ & $0.0255$ & $0.3136$ & $0.3107$ & $0.3136$ & $0.3107$ \\
$5$ & $4$ & $0.5556$ & $0.0247$ & $0.3131$ & $0.3102$ & $0.3131$ & $0.3102$ \\
$5$ & $4.5$ & $0.5263$ & $0.0237$ & $0.3125$ & $0.3096$ & $0.3125$ & $0.3096$ \\
$5$ & $5$ & $0.5$ & $0.0227$ & $0.3120$ & $0.3091$ & $0.3120$ & $0.3091$ \\
\hline
\end{tabular}}
\caption{Sensitivity analysis for generation $n=25$ (continuation).}\label{Dposterior:table-sens25-IV}
\end{table}

\begin{table}[H]\centering
\scalebox{1}{\begin{tabular}{|c|c|c|c|c|c|c|c|}
\hline
$\rho$ & $\beta$  & Prior mean & Prior variance & $\theta_{35}^{*HD}$ & $\theta_{35}^{+HD}$ & $\theta_{35}^{*NED}$ & $\theta_{35}^{+NED}$ \\
\hline
$0.1$ & $0.1$ & $0.5$ & $0.2083$ & $0.2881$ & $0.2865$ & $0.2881$ & $0.2865$ \\
$0.1$ & $0.5$ & $0.1667$ & $0.0868$ & $0.2879$ & $0.2863$ & $0.2879$ & $0.2863$ \\
$0.1$ & $1$ & $0.0909$ & $0.0394$ & $0.2876$ & $0.2860$ & $0.2876$ & $0.2860$ \\
$0.1$ & $1.5$ & $0.0625$ & $0.0225$ & $0.2874$ & $0.2857$ & $0.2874$ & $0.2857$ \\
$0.1$ & $2$ & $0.0476$ & $0.0146$ & $0.2871$ & $0.2855$ & $0.2871$ & $0.2855$ \\
$0.1$ & $2.5$ & $0.0385$ & $0.0103$ & $0.2868$ & $0.2852$ & $0.2868$ & $0.2852$ \\
$0.1$ & $3$ & $0.0323$ & $0.0076$ & $0.2866$ & $0.2849$ & $0.2866$ & $0.2849$ \\
\hline
\end{tabular}}
\caption{Sensitivity analysis for generation $n=35$.}\label{Dposterior:table-sens35-I}
\end{table}

\begin{table}[H]\centering
\scalebox{1}{\begin{tabular}{|c|c|c|c|c|c|c|c|}
\hline
$\rho$ & $\beta$  & Prior mean & Prior variance & $\theta_{35}^{*HD}$ & $\theta_{35}^{+HD}$ & $\theta_{35}^{*NED}$ & $\theta_{35}^{+NED}$ \\
\hline
$1.5$ & $0.1$ & $0.9375$ & $0.0225$ & $0.2900$ & $0.2883$ & $0.2900$ & $0.2883$ \\
$1.5$ & $0.5$ & $0.7500$ & $0.0625$ & $0.2898$ & $0.2881$ & $0.2898$ & $0.2881$ \\
$1.5$ & $1$ & $0.6000$ & $0.0686$ & $0.2895$ & $0.2878$ & $0.2895$ & $0.2878$ \\
$1.5$ & $1.5$ & $0.5$ & $0.0625$ & $0.2892$ & $0.2876$ & $0.2892$ & $0.2876$ \\
$1.5$ & $2$ & $0.4286$ & $0.0544$ & $0.2890$ & $0.2873$ & $0.2890$ & $0.2873$ \\
$1.5$ & $2.5$ & $0.3750$ & $0.0469$ & $0.2887$ & $0.2870$ & $0.2887$ & $0.2870$ \\
$1.5$ & $3$ & $0.3333$ & $0.0404$ & $0.2884$ & $0.2868$ & $0.2884$ & $0.2868$ \\
\hline
\end{tabular}}
\caption{Sensitivity analysis for generation $n=35$ (continuation).}\label{Dposterior:table-sens35-II}
\end{table}

\begin{table}[H]\centering
\scalebox{1}{\begin{tabular}{|c|c|c|c|c|c|c|c|}
\hline
$\rho$ & $\beta$  & Prior mean & Prior variance & $\theta_{35}^{*HD}$ & $\theta_{35}^{+HD}$ & $\theta_{35}^{*NED}$ & $\theta_{35}^{+NED}$ \\
\hline
$3$ & $0.1$ & $0.9677$ & $0.0076$ & $0.2920$ & $0.2903$ & $0.2920$ & $0.2903$ \\
$3$ & $0.5$ & $0.8571$ & $0.0272$ & $0.2918$ & $0.2901$ & $0.2918$ & $0.2901$ \\
$3$ & $1$ & $0.7500$ & $0.0375$ & $0.2915$ & $0.2898$ & $0.2915$ & $0.2898$ \\
$3$ & $1.5$ & $0.6667$ & $0.0404$ & $0.2912$ & $0.2896$ & $0.2912$ & $0.2896$ \\
$3$ & $2$ & $0.6000$ & $0.0400$ & $0.2909$ & $0.2893$ & $0.2909$ & $0.2893$ \\
$3$ & $2.5$ & $0.5455$ & $0.0381$ & $0.2907$ & $0.2890$ & $0.2907$ & $0.2890$ \\
$3$ & $3$ & $0.5$ & $0.0357$ & $0.2904$ & $0.2888$ & $0.2904$ & $0.2888$ \\
\hline
\end{tabular}}
\caption{Sensitivity analysis for generation $n=35$ (continuation).}\label{Dposterior:table-sens35-III}
\end{table}

\begin{table}[H]\centering
\scalebox{1}{\begin{tabular}{|c|c|c|c|c|c|c|c|}
\hline
$\rho$ & $\beta$  & Prior mean & Prior variance & $\theta_{35}^{*HD}$ & $\theta_{35}^{+HD}$ & $\theta_{35}^{*NED}$ & $\theta_{35}^{+NED}$ \\
\hline
$4.5$ & $0.1$ & $0.9783$ & $0.0038$ & $0.2940$ & $0.2923$ & $0.2940$ & $0.2923$ \\
$4.5$ & $0.5$ & $0.9000$ & $0.0150$ & $0.2938$ & $0.2921$ & $0.2938$ & $0.2921$ \\
$4.5$ & $1$ & $0.8182$ & $0.0229$ & $0.2935$ & $0.2918$ & $0.2935$ & $0.2918$ \\
$4.5$ & $1.5$ & $0.7500$ & $0.0268$ & $0.2932$ & $0.2916$ & $0.2932$ & $0.2916$ \\
$4.5$ & $2$ & $0.6923$ & $0.0284$ & $0.2929$ & $0.2913$ & $0.2929$ & $0.2913$ \\
$4.5$ & $2.5$ & $0.6429$ & $0.0287$ & $0.2927$ & $0.2910$ & $0.2927$ & $0.2910$ \\
$4.5$ & $3$ & $0.6000$ & $0.0282$ & $0.2924$ & $0.2907$ & $0.2924$ & $0.2907$ \\
$4.5$ & $3.5$ & $0.5625$ & $0.0273$ & $0.2921$ & $0.2905$ & $0.2921$ & $0.2905$ \\
$4.5$ & $4$ & $0.5294$ & $0.0262$ & $0.2918$ & $0.2902$ & $0.2918$ & $0.2902$ \\
$4.5$ & $4.5$ & $0.5$ & $0.0250$ & $0.2916$ & $0.2899$ & $0.2916$ & $0.2899$ \\
$4.5$ & $5$ & $0.4737$ & $0.0237$ & $0.2913$ & $0.2897$ & $0.2913$ & $0.2897$ \\
$5$ & $0.1$ & $0.9804$ & $0.0032$ & $0.2947$ & $0.2930$ & $0.2947$ & $0.2930$ \\
$5$ & $0.5$ & $0.9091$ & $0.0127$ & $0.2945$ & $0.2928$ & $0.2945$ & $0.2928$ \\
$5$ & $1$ & $0.8333$ & $0.0198$ & $0.2942$ & $0.2925$ & $0.2942$ & $0.2925$ \\
$5$ & $1.5$ & $0.7692$ & $0.0237$ & $0.2939$ & $0.2922$ & $0.2939$ & $0.2922$ \\
$5$ & $2$ & $0.7143$ & $0.0255$ & $0.2936$ & $0.2919$ & $0.2936$ & $0.2919$ \\
$5$ & $2.5$ & $0.6667$ & $0.0261$ & $0.2933$ & $0.2917$ & $0.2933$ & $0.2917$ \\
$5$ & $3$ & $0.6250$ & $0.0260$ & $0.2931$ & $0.2914$ & $0.2931$ & $0.2914$ \\
$5$ & $3.5$ & $0.5882$ & $0.0255$ & $0.2928$ & $0.2911$ & $0.2928$ & $0.2911$ \\
$5$ & $4$ & $0.5556$ & $0.0247$ & $0.2925$ & $0.2909$ & $0.2925$ & $0.2909$ \\
$5$ & $4.5$ & $0.5263$ & $0.0237$ & $0.2922$ & $0.2906$ & $0.2922$ & $0.2906$ \\
$5$ & $5$ & $0.5$ & $0.0227$ & $0.2920$ & $0.2903$ & $0.2920$ & $0.2903$ \\
\hline
\end{tabular}}
\caption{Sensitivity analysis for generation $n=35$ (continuation).}\label{Dposterior:table-sens35-IV}
\end{table}

\begin{table}[H]\centering
\scalebox{1}{\begin{tabular}{|c|c|c|c|c|c|c|c|}
\hline
$\rho$ & $\beta$ & Prior mean & Prior variance & $\theta_{45}^{*HD}$ & $\theta_{45}^{+HD}$ & $\theta_{45}^{*NED}$ & $\theta_{45}^{+NED}$ \\
\hline
$0.1$ & $0.1$ & $0.5$ & $0.2083$ & $0.2959$ & $0.2950$ & $0.2959$ & $0.2950$ \\
$0.1$ & $0.5$ & $0.1667$ & $0.0868$ & $0.2958$ & $0.2948$ & $0.2958$ & $0.2948$ \\
$0.1$ & $1$ & $0.0909$ & $0.0394$ & $0.2956$ & $0.2947$ & $0.2956$ & $0.2947$ \\
$0.1$ & $1.5$ & $0.0625$ & $0.0225$ & $0.2955$ & $0.2946$ & $0.2955$ & $0.2946$ \\
$0.1$ & $2$ & $0.0476$ & $0.0146$ & $0.2953$ & $0.2944$ & $0.2953$ & $0.2944$ \\
$0.1$ & $2.5$ & $0.0385$ & $0.0103$ & $0.2952$ & $0.2943$ & $0.2952$ & $0.2943$ \\
$0.1$ & $3$ & $0.0323$ & $0.0076$ & $0.2950$ & $0.2941$ & $0.2950$ & $0.2941$ \\
\hline
\end{tabular}}
\caption{Sensitivity analysis for generation $n=45$.}\label{Dposterior:table-sens45-I}
\end{table}

\begin{table}[H]\centering
\scalebox{1}{\begin{tabular}{|c|c|c|c|c|c|c|c|}
\hline
$\rho$ & $\beta$ & Prior mean & Prior variance & $\theta_{45}^{*HD}$ & $\theta_{45}^{+HD}$ & $\theta_{45}^{*NED}$ & $\theta_{45}^{+NED}$ \\
\hline
$1.5$ & $0.1$ & $0.9375$ & $0.0225$ & $0.2969$ & $0.2960$ & $0.2969$ & $0.2960$ \\
$1.5$ & $0.5$ & $0.7500$ & $0.0625$ & $0.2968$ & $0.2958$ & $0.2968$ & $0.2958$ \\
$1.5$ & $1$ & $0.6000$ & $0.0686$ & $0.2966$ & $0.2957$ & $0.2966$ & $0.2957$ \\
$1.5$ & $1.5$ & $0.5$ & $0.0625$ & $0.2965$ & $0.2956$ & $0.2965$ & $0.2956$ \\
$1.5$ & $2$ & $0.4286$ & $0.0544$ & $0.2963$ & $0.2954$ & $0.2963$ & $0.2954$ \\
$1.5$ & $2.5$ & $0.3750$ & $0.0469$ & $0.2962$ & $0.2953$ & $0.2962$ & $0.2953$ \\
$1.5$ & $3$ & $0.3333$ & $0.0404$ & $0.2960$ & $0.2951$ & $0.2960$ & $0.2951$ \\
\hline
\end{tabular}}
\caption{Sensitivity analysis for generation $n=45$ (continuation).}\label{Dposterior:table-sens45-II}
\end{table}

\begin{table}[H]\centering
\scalebox{1}{\begin{tabular}{|c|c|c|c|c|c|c|c|}
\hline
$\rho$ & $\beta$ & Prior mean & Prior variance & $\theta_{45}^{*HD}$ & $\theta_{45}^{+HD}$ & $\theta_{45}^{*NED}$ & $\theta_{45}^{+NED}$ \\
\hline
$3$ & $0.1$ & $0.9677$ & $0.0076$ & $0.2979$ & $0.2970$ & $0.2979$ & $0.2970$ \\
$3$ & $0.5$ & $0.8571$ & $0.0272$ & $0.2978$ & $0.2969$ & $0.2978$ & $0.2969$ \\
$3$ & $1$ & $0.7500$ & $0.0375$ & $0.2977$ & $0.2968$ & $0.2977$ & $0.2968$ \\
$3$ & $1.5$ & $0.6667$ & $0.0404$ & $0.2975$ & $0.2966$ & $0.2975$ & $0.2966$ \\
$3$ & $2$ & $0.6000$ & $0.0400$ & $0.2974$ & $0.2965$ & $0.2974$ & $0.2965$ \\
$3$ & $2.5$ & $0.5455$ & $0.0381$ & $0.2972$ & $0.2963$ & $0.2972$ & $0.2963$ \\
$3$ & $3$ & $0.5$ & $0.0357$ & $0.2971$ & $0.2962$ & $0.2971$ & $0.2962$ \\
\hline
\end{tabular}}
\caption{Sensitivity analysis for generation $n=45$ (continuation).}\label{Dposterior:table-sens45-III}
\end{table}

\begin{table}[H]\centering
\scalebox{1}{\begin{tabular}{|c|c|c|c|c|c|c|c|}
\hline
$\rho$ & $\beta$ & Prior mean & Prior variance & $\theta_{45}^{*HD}$ & $\theta_{45}^{+HD}$ & $\theta_{45}^{*NED}$ & $\theta_{45}^{+NED}$ \\
\hline
$4.5$ & $0.1$ & $0.9783$ & $0.0038$ & $0.2990$ & $0.2981$ & $0.2990$ & $0.2981$ \\
$4.5$ & $0.5$ & $0.9000$ & $0.0150$ & $0.2989$ & $0.2980$ & $0.2989$ & $0.2980$ \\
$4.5$ & $1$ & $0.8182$ & $0.0229$ & $0.2987$ & $0.2978$ & $0.2987$ & $0.2978$ \\
$4.5$ & $1.5$ & $0.7500$ & $0.0268$ & $0.2986$ & $0.2977$ & $0.2986$ & $0.2977$ \\
$4.5$ & $2$ & $0.6923$ & $0.0284$ & $0.2984$ & $0.2975$ & $0.2984$ & $0.2975$ \\
$4.5$ & $2.5$ & $0.6429$ & $0.0287$ & $0.2983$ & $0.2974$ & $0.2983$ & $0.2974$ \\
$4.5$ & $3$ & $0.6000$ & $0.0282$ & $0.2981$ & $0.2972$ & $0.2981$ & $0.2972$ \\
$4.5$ & $3.5$ & $0.5625$ & $0.0273$ & $0.2980$ & $0.2971$ & $0.2980$ & $0.2971$ \\
$4.5$ & $4$ & $0.5294$ & $0.0262$ & $0.2978$ & $0.2969$ & $0.2978$ & $0.2969$ \\
$4.5$ & $4.5$ & $0.5$ & $0.0250$ & $0.2977$ & $0.2968$ & $0.2977$ & $0.2968$ \\
$4.5$ & $5$ & $0.4737$ & $0.0237$ & $0.2975$ & $0.2966$ & $0.2975$ & $0.2966$ \\
$5$ & $0.1$ & $0.9804$ & $0.0032$ & $0.2994$ & $0.2984$ & $0.2994$ & $0.2984$ \\
$5$ & $0.5$ & $0.9091$ & $0.0127$ & $0.2992$ & $0.2983$ & $0.2992$ & $0.2983$ \\
$5$ & $1$ & $0.8333$ & $0.0198$ & $0.2991$ & $0.2982$ & $0.2991$ & $0.2982$ \\
$5$ & $1.5$ & $0.7692$ & $0.0237$ & $0.2989$ & $0.2980$ & $0.2989$ & $0.2980$ \\
$5$ & $2$ & $0.7143$ & $0.0255$ & $0.2988$ & $0.2979$ & $0.2988$ & $0.2979$ \\
$5$ & $2.5$ & $0.6667$ & $0.0261$ & $0.2986$ & $0.2977$ & $0.2986$ & $0.2977$ \\
$5$ & $3$ & $0.6250$ & $0.0260$ & $0.2985$ & $0.2976$ & $0.2985$ & $0.2976$ \\
$5$ & $3.5$ & $0.5882$ & $0.0255$ & $0.2983$ & $0.2974$ & $0.2983$ & $0.2974$ \\
$5$ & $4$ & $0.5556$ & $0.0247$ & $0.2982$ & $0.2973$ & $0.2982$ & $0.2973$ \\
$5$ & $4.5$ & $0.5263$ & $0.0237$ & $0.2980$ & $0.2971$ & $0.2980$ & $0.2971$ \\
$5$ & $5$ & $0.5$ & $0.0227$ & $0.2979$ & $0.2970$ & $0.2979$ & $0.2970$ \\
\hline
\end{tabular}}
\caption{Sensitivity analysis for generation $n=45$ (continuation).}\label{Dposterior:table-sens45-IV}
\end{table}


\end{document}